\documentclass[article,12pt]{article}

\usepackage[a4paper, total={6.5in, 8.5in}]{geometry}
\usepackage[numbers,sort&compress]{natbib}
\usepackage{thm-restate}
\usepackage{authblk}
\usepackage{amssymb}
\usepackage{amsmath}
\usepackage{amsthm}
\usepackage{hyperref}
\hypersetup{
    colorlinks=true,
    linkcolor=blue,
    filecolor=magenta,
    urlcolor=cyan,
}
\usepackage{xspace}
\usepackage{mathtools}
\usepackage{enumerate}
\usepackage{subcaption}
\usepackage[linesnumbered,algoruled,boxed,lined]{algorithm2e}
\usepackage{complexity} 
\usepackage{todonotes}
\presetkeys{todonotes}{inline}{}
\usepackage{comment}
\usepackage[capitalise,noabbrev]{cleveref}
\usepackage{tikz}
\usepackage{enumitem}
\setlist{nolistsep}
\setlist{nosep}

\newcommand{\parent}{\operatorname{\mathtt{parent}}}

\newcommand{\deff}{\coloneqq}

\newcommand{\interV}[1]{V_{#1}^{\cap}} 
\newcommand{\interR}[1]{R_{#1}^{\cap}} 
\newcommand{\interB}[1]{B_{#1}^{\cap}} 

\newcommand{\belowV}[1]{V_{#1}^{\subseteq}} 
\newcommand{\belowR}[1]{R_{#1}^{\subseteq}} 
\newcommand{\belowB}[1]{B_{#1}^{\subseteq}} 

\newcommand{\underV}[1]{V_{#1}^{\subseteq\dag}} 
\newcommand{\underR}[1]{R_{#1}^{\subseteq\dag}} 
\newcommand{\underB}[1]{B_{#1}^{\subseteq\dag}} 

\newcommand{\containV}[1]{V_{#1}^{\in}} 
\newcommand{\containR}[1]{R_{#1}^{\in}} 
\newcommand{\containB}[1]{B_{#1}^{\in}} 

\usepackage{color}

\DeclarePairedDelimiter\abs{\lvert}{\rvert}

\newcommand{\at}{\texttt{at}}
\newcommand{\lf}{\texttt{lf}}

\newcommand{\calC}{\mathcal{C}}

\newcommand{\calF}{\mathcal{F}}

\newcommand{\calI}{\mathcal I}

\newcommand{\calM}{\ensuremath{{\mathcal M}}}
\newcommand{\calO}{\ensuremath{{\mathcal O}}}
\newcommand{\OO}{\mathcal{O}}
\newcommand{\Oh}{\mathcal{O}} 
\newcommand{\OPT}{\textsc{OPT}}
\newcommand{\calP}{\mathcal{P}}

\newcommand{\calT}{\mathcal{T}}

\newcommand{\WoneH}{\textup{\textsf{W[1]}}-hard\xspace}
\newcommand{\ETH}{\textsf{ETH}\xspace}

\newcommand{\hard}{{hard}}

\newcommand{\paraNP}{\textsf{paraNP}}

\newcommand{\yes}{\textsc{Yes}}
\newcommand{\no}{\textsc{No}}

\newtheorem{theorem}{Theorem}[section]
\newtheorem{lemma}[theorem]{Lemma}

\newtheorem{claim}[theorem]{Claim}
\newtheorem{observation}[theorem]{Observation}

\newtheorem{reduction rule}[theorem]{Reduction Rule}
\Crefname{reduction rule}{Reduction~Rule}{Reduction~Rules}
\newtheorem{greedy select}[theorem]{Greedy~Select}\Crefname{greedy select}{Greedy~Select}{Greedy~Select}
\newtheorem{marking-scheme}[theorem]{Marking Scheme}
\newtheorem{definition}[theorem]{Definition}

\newenvironment{claimproof}{\begin{proof}\renewcommand{\qedsymbol}{\claimqed}}{\end{proof}\renewcommand{\qedsymbol}{\plainqed}}
\let\plainqed\qedsymbol

\newcommand{\DS}{\textsc{DomSet}\xspace}
\newcommand{\connDS}{\textsc{connDomSet}\xspace}
\newcommand{\RBDS}{\textsc{Red-Blue-DomSet}\xspace}
\newcommand{\RestRBDS}{\textsc{Rest-Red-Blue-DomSet}\xspace}
\newcommand{\conRBDS}{\textsc{Connected Red-Blue-DomSet}\xspace}
\newcommand{\SetCov}{\textsc{Set~Cover}\xspace}
\newcommand{\HitSet}{\textsc{Hitting~Set}\xspace}
\newcommand{\SteinTree}{\textsc{Steiner~Tree}\xspace}

\newcommand{\MCundel}{\textsc{MultiCut with UnDel Term}\xspace}

\newcommand{\defproblem}[3]{
  \vspace{1mm}
\noindent\fbox{
  \begin{minipage}{0.96\textwidth}
  \begin{tabular*}{\textwidth}{@{\extracolsep{\fill}}lr} #1 \\ \end{tabular*}
  {\bf{Input:}} #2  \newline
  {\bf{Question:}} #3
  \end{minipage}
  }
  \vspace{1mm}
}

\newcommand{\defproblemques}[3]{
  \vspace{1mm}
\noindent\fbox{
  \begin{minipage}{0.96\textwidth}
  \begin{tabular*}{\textwidth}{@{\extracolsep{\fill}}lr} #1 \\ \end{tabular*}
  {\bf{Input:}} #2  \\
  {\bf{Question:}} #3
  \end{minipage}
  }
  \vspace{1mm}
}

\newcommand{\wt}{\texttt{wt}}

\newcommand{\mwcfull}{\textup{\textsc{Multiway Cut with Undeletable Terminals}}}
\newcommand{\mwc}{\textsc{MWC}}
\newcommand{\mwcset}{multiway-cut}
\newcommand{\stcut}{\textsc{$(s,t)$-Cut}}

\newcommand{\connection}{connection}
\newcommand{\ver}{\textnormal{\texttt{ver}}}
\newcommand{\mld}{\mathcal{M}}
\newcommand{\rr}{\textsf{r}}
\newcommand{\node}{\gamma}
\newcommand{\topnode}{\mathtt{top}}
\newcommand{\botnode}{\mathtt{bot}}
\newcommand{\sink}{sink}
\newcommand{\Pterm}{P}

\newcommand{\truncated}{truncated}

\newcommand{\Pairs}{\Pterm}

\begin{document}

\title{Domination and Cut Problems on Chordal Graphs with Bounded Leafage}
\date{}

\author[1]{Esther Galby}
\author[1]{D\'aniel Marx}
\author[1]{Philipp Schepper}
\author[2]{Roohani Sharma}
\author[1]{Prafullkumar Tale}

\affil[1]{CISPA Helmholtz Center for Information Security, Germany}
\affil[2] {Max Planck Institute for Informatics, SIC, Saarbr\"ucken, Germany}

\maketitle

\begin{abstract}
The leafage of a chordal graph $G$ is the minimum integer $\ell$ such that $G$ can be realized as an intersection graph of subtrees of a tree with $\ell$ leaves.
We consider structural parameterization by the leafage of classical domination and cut problems on chordal graphs.
Fomin, Golovach, and Raymond~[ESA~$2018$, Algorithmica~$2020$] proved, among other things, that \textsc{Dominating Set} on chordal graphs admits an algorithm running in time $2^{\mathcal{O}(\ell^2)} \cdot n^{\mathcal{O}(1)}$. 
We present a conceptually much simpler algorithm that runs in time $2^{\mathcal{O}(\ell)} \cdot n^{\mathcal{O}(1)}$.
We extend our approach to obtain similar results for \textsc{Connected Dominating Set} and \textsc{Steiner Tree}.
We then consider the two classical cut problems \textsc{MultiCut with Undeletable Terminals} and \textsc{Multiway Cut with Undeletable Terminals}.
We prove that the former is \textsf{W}[1]-hard when parameterized by the leafage and complement this result by presenting a simple $n^{\mathcal{O}(\ell)}$-time algorithm.
To our surprise, we find that \textsc{Multiway Cut with Undeletable Terminals} on chordal graphs can be solved, in contrast,  in $n^{\calO(1)}$-time.
\end{abstract}

\newcommand{\ds}{{\sc Dominating Set}}
\newcommand{\mc}{{\sc Multicut}}
\newcommand{\rbds}{{\sc Red-Blue Dominating Set}}
\newcommand{\connds}{{\sc Connected Dominating Set}}
\newcommand{\rbconnds}{{\sc Red-Blue Connected Dominating Set}}

\newcommand{\dssmall}{{\sc DS}}
\newcommand{\rbdssmall}{{\sc RB-DS}}
\newcommand{\conndssmall}{{\sc Conn-DS}}
\newcommand{\rbconndssmall}{{\sc RB-Conn-DS}}

\newcommand{\polyn}{n^{\OO(1)}}
\newcommand{\SETH}{SETH}

\newcommand{\timefptds}{2^{\OO(\ell)} \cdot \polyn}
\newcommand{\timexpmc}{f(\ell) \cdot n^{\OO(\ell)}}

\section{Introduction}
The intersection graph of a family $\calF$ of nonempty sets is the graph whose vertices are the elements of $\calF$ with two vertices being adjacent if and only if their corresponding sets intersect.
The most natural and famous example of such intersection graphs are \emph{interval graphs}
where $\calF$ is a collection of subpaths of a path.
Due to their applicability in scheduling, interval graphs have received a considerable attention in the realm of algorithmic graph theory.
One useful characterization of an interval graph is that its maximal cliques can be linearly ordered such that for every vertex, the maximal cliques containing that vertex occur consecutively~\cite{gilmore_hoffman_1964}.
This property proves very useful for the design of polynomial-time dynamic programming based or greedy algorithms on interval graphs.

Consider the generalization where $\calF$ is a collection of subtrees of a tree instead of subpaths of a path.
In this case, the corresponding class of intersection graphs is exactly that of \emph{chordal graphs} \cite{walter1972representations,Gavril1974TheIG,DBLP:journals/dm/Buneman74}.
Recall that a graph is chordal if every cycle of length at least $4$ has a chord.
Often, the algorithms of the types mentioned in the previous paragraph fail to generalize to this superclass as witnessed by the following problems that admit polynomial-time algorithms on interval graphs but are \NP-complete on chordal graphs:
{\sc Dominating Set}~\cite{DBLP:journals/siamcomp/Chang98,DBLP:journals/ipl/Bertossi84}, 
{\sc Connected Dominating Set}~\cite{DBLP:conf/wads/BalakrishnanRR93,DBLP:journals/networks/WhiteFP85},
{\sc Steiner Tree}~\cite{DBLP:conf/wads/BalakrishnanRR93,DBLP:journals/networks/WhiteFP85}, 
{\sc Multicut with Undeletable Terminals}~\cite{DBLP:journals/eor/GuoHKNU08,DBLP:journals/dam/Papadopoulos12}, 
{\sc Subset Feedback Vertex Set (Subset FVS)}~\cite{DBLP:journals/dam/PapadopoulosT19,DBLP:journals/algorithmica/FominHKPV14},
{\sc Longest Cycle}~\cite{DBLP:journals/ipl/Keil85, golumbic2004algorithmic}\footnote{See Exercise $2$ in Chapter~$6$ in \cite{golumbic2004algorithmic}.},
{\sc Longest Path}~\cite{DBLP:journals/algorithmica/IoannidouMN11}, 
{\sc Component Order Connectivity}~\cite{DBLP:journals/algorithmica/DrangeDH16}, 
{\sc $s$-Club Contraction}~\cite{DBLP:journals/siamdm/GolovachHHP15}, 
{\sc Independent Set Reconfiguration}~\cite{DBLP:journals/mst/BelmonteKLMOS21}, 
{\sc Bandwidth}~\cite{DBLP:journals/iandc/Kratsch87}, 
{\sc Cluster Vertex Deletion}~\cite{DBLP:journals/algorithmica/KonstantinidisP21}.
Also, \textsc{Graph Isomorphism} on chordal graphs is polynomial-time equivalent to the problem on general graphs whereas it admits a linear-time algorithm on interval graphs~\cite{DBLP:journals/jacm/LuekerB79}.

The problems above remain hard even on \emph{split graphs}, another well-studied subclass of chordal graphs.
A graph is a split graph if its vertex set can be partitioned into a clique and an independent set.
The collection of split graphs is a (proper) subset of the class of intersection graphs where $\calF$ is a collection of substars of a star. 
As interval graphs are intersection graphs of subpaths of a path (a tree with two leaves) and split graphs are intersection graphs of substars of a star (a tree with arbitrary number of leaves), a natural question to consider is what happens to these problems on subclasses of chordal graphs that are intersection graphs of subtrees of a tree with a bounded number of leaves.
Motivated by such questions, we consider the notion of \emph{leafage} introduced by Lin et al.~\cite{DBLP:journals/dmgt/LinMW98}:
the leafage of a chordal graph $G$ is the minimum integer $\ell$ such that $G$ can be realized as an intersection graph of a collection $\calF$ of subtrees of a tree that has $\ell$ leaves.
Note that the leafage of interval graphs is at most $2$ while split graphs have unbounded leafage.
Thus the leafage measures, in some sense, how close a chordal graph is to an interval graph.
Alternately, an \FPT\ or \XP\ algorithm parameterized by the leafage can be seen as a generalization of the algorithm on interval graphs.

\paragraph{Related Work.} Habib and Stacho~\cite{DBLP:conf/esa/HabibS09} showed that we can compute the leafage of a connected chordal graph in polynomial time. 
Their algorithm also constructs a corresponding \emph{representation tree}\footnote{We present formal definitions of the terms used in this section in Section~\ref{sec:prelim}.} $T$ with the minimum number of leaves.
In recent years, researchers have studied the structural parameterization of various graph problems on chordal graphs parameterized by the leafage.
Fomin et al.~\cite{DBLP:journals/algorithmica/FominGR20} and Arvind et al.~\cite{DBLP:journals/corr/ArvindNPZ22} proved, respectively, that the \textsc{Dominating Set} and \textsc{Graph Isomorphism} problems on chordal graphs are \FPT\ parameterized by the leafage.
Barnetson et al.~\cite{DBLP:journals/networks/BarnetsonBEHPR21} and Papadopoulos and Tzimas \cite{DBLP:conf/iwoca/PapadopoulosT22} presented \XP-algorithms running in time $n^{\calO(\ell)}$ for \textsc{Fire Break} and \textsc{Subset FVS} on chordal graphs, respectively.
Papadopoulos and Tzimas \cite{DBLP:conf/iwoca/PapadopoulosT22} also proved that \textsc{Subset FVS} is \W[1]-\hard\ when parameterized by the leafage.
Hochst{\"{a}}ttler et al~\cite{DBLP:journals/dam/HochstattlerHMP21} showed that we can compute the neighborhood polynomial of a chordal graph in $n^{\calO(\ell)}$-time.

It is known that the size of  \emph{asteroidal set} in a chordal graph is upper bounded by its leafage~\cite{DBLP:journals/dmgt/LinMW98}. 
See~\cite{DBLP:journals/ejc/HabibS12,DBLP:journals/dam/Alcon14} for the relationship between leafage and other structural properties of chordal graphs.
Kratsch and Stewart~\cite{DBLP:journals/siamdm/KratschS02} proved that we can effectively $2\ell$-approximate bandwidth of chordal graphs of leafage $\ell$.
Chaplick and Stacho~\cite{DBLP:journals/dam/ChaplickS14} generalized the notion of leafage to \emph{vertex leafage} and proved that, unlike leafage, it is hard to determine the optimal vertex leafage of a given chordal graph.
Figueiredo et al.~\cite{DBLP:conf/walcom/FigueiredoLMS22} proved that \textsc{Dominating Set}, \textsc{Connected Dominating Set} and \textsc{Steiner Tree} are \FPT\ on chordal graphs when parameterized by the size of the solution plus the vertex leafage, provided that a tree representation with optimal vertex leafage is given as part of the input.

\paragraph{Our Results.}
We consider well-studied domination and cut problems on chordal graphs.
As our first result, we prove that \textsc{Dominating Set} on chordal graphs of leafage at most $\ell$ admits an algorithm running in time $2^{\OO(\ell)} \cdot \polyn$.
This improves upon the existing algorithm by Fomin et al.~\cite[Theorem~$9$]{DBLP:journals/algorithmica/FominGR20} which runs in time $2^{\calO(\ell^2)} \cdot \polyn$.
Despite being significantly simpler than the algorithm in \cite{DBLP:journals/algorithmica/FominGR20}, our algorithm in fact solves the \textsc{Red-Blue Dominating Set} problem, a well-known generalization of \textsc{Dominating Set}.
In this generalized version, an input is a graph $G$ with a partition $(R, B)$ of its vertex set and an integer $k$, and the objective is to find a subset $D$ of $R$ that dominates every vertex in $B$, i.e., $B \subseteq N(D)$.
We further use this algorithm to solve other related domination problems.
\begin{restatable}{theorem}{thm-ds-fpt}
\label{thm:ds-fpt}
{\sc Dominating Set}, {\sc Connected Dominating Set}, and {\sc Steiner Tree} can be solved in $\timefptds$ on chordal graphs of leafage at most $\ell$.
\end{restatable}
The reductions in \cite{DBLP:journals/ipl/Bertossi84} and \cite{DBLP:journals/networks/WhiteFP85} used to prove that these problems are \NP-complete on chordal graphs imply that these problems do not admit $2^{o(n)}$, and hence $2^{o(\ell)} \cdot \polyn$, algorithms unless the \ETH\ fails.

Arguably, the two most studied cut problems are {\sc MultiCut} and {\sc Multiway Cut}.
In the {\sc MultiCut} problem, an input is graph $G$, a set of terminal pairs $P \subseteq V(G) \times V(G)$ and an integer $k$, and the objective is to find a subset $S \subseteq V(G)$ of size at most $k$ such that no pair of vertices in $P$ is connected in $G - S$.
In the {\sc Multiway Cut} problem, instead of terminal pairs, we are given a terminal set $P$ and the objective is to find a subset $S \subseteq V(G)$ of size at most $k$ such that no two vertices in $P$ are connected in $G - S$.
These problems and variations of them have received a considerable attention which lead to the development of new techniques~\cite{DBLP:journals/tcs/Marx06, DBLP:journals/siamcomp/MarxR14, DBLP:journals/siamcomp/BousquetDT18, DBLP:journals/siamcomp/ChitnisHM13, DBLP:journals/talg/ChitnisCHM15}.
Misra et al.~\cite{DBLP:conf/mfcs/MisraP00S20} studied the parameterized complexity of these problems on chordal graphs.
Guo et al.~\cite{DBLP:journals/eor/GuoHKNU08} proved that \textsc{MultiCut with Deletable Terminals} is \NP-complete on interval graphs, thereby implying that this problem is \paraNP-hard when parameterized by the leafage.
We consider the \textsc{MultiCut with Undeletable Terminals} problem and prove the following result. 

\begin{restatable}{theorem}{thmmcleafagewhard}
\label{thm:mc-w-hardness}
{\sc MultiCut with Undeletable Terminals} on chordal graphs
is {\em \W[1]-hard} when parameterized by the leafage $\ell$
and assuming the {\em \ETH}, does not admit an algorithm
running in time $f(\ell)\cdot n^{o(\ell)}$ for any computable function $f$.
However, it admits an \XP-algorithm running in time $n^{\calO(\ell)}$.
\end{restatable}

Next, we focus on the \textsc{Multiway Cut with Undeletable Terminals} problem.
We find it somewhat surprising that the classical complexity of this problem on chordal graphs was not known.
Bergougnoux et al.~\cite{DBLP:journals/algorithmica/BergougnouxPT22}, using the result in \cite{DBLP:journals/algorithmica/FominGR20}, proved that the problem admits an \XP-algorithm when parameterized by the leafage\footnote{See the discussion after Corollary~$2$ on page 1388 in \cite{DBLP:journals/algorithmica/BergougnouxPT22}.}.
Our next result significantly improves upon this and \cite[Theorem~$2$]{DBLP:conf/mfcs/MisraP00S20} which states that the problem admits a polynomial kernel when parameterized by the solution size.

\begin{restatable}{theorem}{thm:mwc-poly}
\label{thm:mwc-poly}
\textsc{\sc Multiway Cut with Undeletable Terminals} can be solved in $\polyn$-time on chordal graphs.
\end{restatable}

A well-known trick to convert an instance of \textsc{Multiway Cut with Deletable Terminals} into an instance of \textsc{Multiway Cut with Undeletable Terminals} is to add a pendant vertex to each terminal, remove that vertex from the set of terminals, and make the newly added vertex a terminal.
As this reduction converts a chordal graph into another chordal graph, \Cref{thm:mwc-poly} implies that \textsc{Multiway Cut with Deletable Terminals} is also polynomial-time solvable on chordal graphs.
Another closely related problem is \textsc{Subset FVS}
which is \NP-complete on split graphs~\cite{DBLP:journals/dam/PapadopoulosT19}.
To the best of our knowledge, this is the first graph class on which the classical complexity of these two problems differ. 

Next, we revisit the problems on chordal graphs with bounded leafage and examine how far we can generalize this class.
An \emph{asteroidal triple} of a graph $G$ is a set of three vertices such that each pair is connected by some path 
that avoids the closed neighborhood of  the third vertex.
Lekkerkerker and Boland~\cite{lekkeikerker1962} showed that a graph is an interval graph if and only if it is chordal and does not contain an asteroidal triple.
They also listed all minimal chordal graphs that contain an asteroidal triple (see, for instance, \cite[Figure~1]{DBLP:conf/soda/Cao16}).
Among this list, we found the \emph{net graph} to be the most natural to generalize.
For a positive integer $\ell \ge 3$, we define $H_{\ell}$ as a split graph on $2\ell$ vertices with split partition $(C, I)$ such that the only edges across $C, I$ are a perfect matching.
Note that $H_{3}$ is the net graph.
As interval graphs are a proper subset of the collection of chordal graphs that do not contain a net graph as an induced subgraph,
the collection of the chordal graph of leafage $\ell$ is a proper subset of the collection of chordal graphs that do not contain $H_{\ell + 1}$ as an induced subgraph (see Observation~\ref{obs:subclasses-chordal}). 
We show that, although the considered domination problems are polynomial-time solvable for constant $\ell$, the fixed-parameter tractability results are unlikely to extend to this larger class.

\begin{restatable}{theorem}{thmdomsetwhard}
\label{thm:dom-set-ind-sub-w-hard}
{\sc Dominating Set}, {\sc Connected Dominating Set} and {\sc Steiner Tree} on $H_{\ell}$-induced-subgraph-free chordal graphs are {\em \WoneH} when parameterized by $\ell$ and assuming the {\em \ETH}, do not admit an algorithm running in time $f(\ell)\cdot n^{o(\ell)}$ for any computable function $f$.
However, they all admit \XP-algorithms running in time $n^{\calO(\ell)}$.
\end{restatable}

We observe a similar trend with respect to \textsc{MultiCut with Undeletable Terminals} as its parameterized complexity jumps from \W[1]-\hard\ on chordal graph of leafage $\ell$ to {\paraNP-\hard} on $H_{\ell}$-induced-subgraph-free chordal graphs when parameterized by $\ell$.

\begin{restatable}{theorem}{thm:mc-para-NP-hard}
\label{thm:mc-para-NP-hard}
{\sc MultiCut with Undeletable Terminals} is {\em \NP-hard}
even when restricted to $H_{3}$-induced-subgraph-free chordal graphs.
\end{restatable}
\noindent
Table~\ref{fig:overview-results} summarises our results.

\begin{table}[t]
\def\arraystretch{1.5}
\begin{tabular}{|p{0.2\textwidth}|p{0.25\textwidth}|p{0.25\textwidth}|p{0.2\textwidth}|}
\hline
Input graph & \textsc{Dom Set}, \textsc{ConnDom Set}, \textsc{Steiner Tree} & \textsc{MultiCut with UnDel Term} & \textsc{MultiwayCut} \\
\hline
Interval Graphs & Poly-time~\cite{DBLP:journals/siamcomp/Chang98,DBLP:conf/wads/BalakrishnanRR93} & Poly-time~\cite{DBLP:journals/eor/GuoHKNU08} & Poly-time~\cite{DBLP:journals/algorithmica/BergougnouxPT22}\\
\hline
Chordal graphs of leafage $\ell$ & 
$2^{\calO(\ell^2)} \cdot n^{\calO(1)}$ algo \cite{{DBLP:journals/algorithmica/FominGR20}} \newline
$2^{\calO(\ell)} \cdot n^{\calO(1)}$ algo (Thm~\ref{thm:ds-fpt})
 & \W[1]-hard \newline
 $n^{\calO(\ell)}$ algo (Thm~\ref{thm:mc-w-hardness})
 & $n^{\calO(\ell)}$
 algo~\cite{DBLP:journals/algorithmica/BergougnouxPT22} \newline
 Poly-time (Thm~\ref{thm:mwc-poly}) \\
\hline
$H_{\ell}$-induced subgraph-free chordal
  & \W[1]-hard; \newline $n^{\calO(\ell)}$ algo (Thm~\ref{thm:dom-set-ind-sub-w-hard})
  & \NP-hard\ for $\ell \ge 3$ \newline
(Thm~\ref{thm:mc-para-NP-hard})
  & Poly-time (Thm~\ref{thm:mwc-poly})\\
\hline
Chordal graphs
  & \NP-complete~\cite{DBLP:journals/ipl/Bertossi84}
  & \NP-complete~\cite{DBLP:journals/networks/WhiteFP85}
  & Poly-time (Thm~\ref{thm:mwc-poly}) \\
\hline
\end{tabular}
\medskip
\caption{Overview of the known results and our contributions.
Every graph class mentioned in the first column is a proper subset of the graph class mentioned below.
\label{fig:overview-results}}
\end{table}

\paragraph*{Our Methods.} We briefly discuss the methods used in our two main algorithms, 
namely the algorithm for {\sc Dominating Set} and the one for {\sc Multiway Cut}.

\subparagraph*{Red-Blue Dominating Set in Chordal Graphs.}
As mentioned earlier, the linear ordering of cliques in interval graphs is particularly useful for the design of polynomial-time algorithms.
Such an ordering is not possible even if $G$ is a chordal graph whose representation tree $T$ is a star.
Consider the case where the model of every red vertex in $G$ includes the center of the star $T$ (and possibly some leaves) and the model of every blue vertex is (only) a leaf.
We can solve this instance by converting it to an instance of \SetCov and solving it using the \FPT\ algorithm parameterized by the size of the universe. 
In this case, the size of the universe is at most the number of leaves which is upper bounded by the leafage. 
In the other case where the properties of red vertices and blue vertices are reversed, we obtain a similar result by creating an equivalent instance of \HitSet.

These ideas can be used in a more general setting as long as the following two properties are satisfied:
(1) the model of each vertex is \emph{local}, that is, it contains at most one branching node,
and (2) each branching node is contained only in models of either red vertices or blue vertices.
Based on this observation, we introduce a restricted version of the problem in which the input graph is required to satisfy these two conditions.
We then show that the general case reduces to this restricted version:
indeed, we prove that there is a branching algorithm that constructs $2^{\calO(\ell)}$ many instances (where $\ell$ is the leafage of the input graph) of the restricted version of the problem
such that the input instance is a \yes-instance
if and only if one of these newly created instances is a \yes-instance.
These two properties ensure that the graph induced by the red and blue vertices whose model intersect the subtree rooted at a farthest branching node (from some fixed root) satisfies the premise of at least one of the cases mentioned in the previous paragraph.
We then present a greedy procedure, based on solving the \SetCov and \HitSet problems, that identifies some part of an optimum solution.
Apart from this greedy selection procedure, all other steps of the algorithm run in polynomial time.

\subparagraph*{Multiway Cut in Chordal Graphs.}
We give a polynomial-time algorithm for {\sc Multiway Cut} on chordal graphs by solving several instances of the \stcut\ problem (not necessarily with unit capacities). Our strategy is based on a bottom-up dynamic programming (DP) on a tree representation of a chordal graph.
An interesting aspect of our DP is that we need to look-up {\em all} DP table values that are already computed to compute a new entry.
This is in contrast to typical DP-based algorithms that do computations only based on {\em local} entries.

We remark that we do not expect to design an algorithm for {\sc Multiway Cut} on chordal graphs using much simpler arguments (like a simple dynamic programming procedure etc.) as the problem generalizes some well-studied cut-flow based problems.
As an example,
recall the {\sc Vertex Cover} problem on bipartite graphs where given a bipartite graph $G$ with bipartition $(A,B)$, the goal is to find $A' \subseteq A$ and $B' \subseteq B$ such that $|A' \cup B'|$ is minimum and $N(A \setminus A') \subseteq B'$.
The set $A'\cup B'$ is called a vertex cover of $G$.  The {\sc Vertex Cover} problem on bipartite graphs reduces to the {\sc Multiway Cut} problem on chordal graphs: indeed, let $G'$ be the graph obtained from $G$ by making $B$ a clique, adding new pendant vertex $t_a$ to each vertex $a \in A$, and further adding another new vertex $t$ that is adjacent to all vertices of $B$.
Then $G'$ is a chordal graph and letting $T=t \cup \{t_a~|~ a \in A\}$,
it is easy to see that $S \subseteq V(G)$ is a vertex cover of $G$ if and only if $S$ is a $T$-multiway-cut in $G'$. 
As mentioned earlier, our algorithm solves several instances of the \stcut\ problem, which also sits at the heart of some algorithms for {\sc Vertex Cover} on bipartite graphs.
The above reduction suggests that an algorithm for {\sc Multiway Cut} on chordal graphs using much simpler techniques, would imply an algorithm for {\sc Vertex Cover} on bipartite graphs that uses much simpler techniques as well.  

Note that a similar reduction would work from the weighted variant of the {\sc Vertex Cover} problem on bipartite graphs. 
This can be achieved by further replacing each vertex of the graph $G$ by a clique of size proportional to the weight of this vertex and making each vertex of the clique adjacent to all the neighbors of this vertex. 
This reduction still preserves the chordality of the resulting graph. 

\paragraph*{Organization of the Paper.}
In \Cref{sec:prelim}, we define the notations and terminology used throughout the paper.
In \Cref{sec:dominating-set}, we present the \FPT\ algorithm for the generalized {\sc Red-Blue Dominating Set} problem parameterized by the leafage.
In \Cref{sec:multicut-bounded-leafage}, we consider the {\sc Multicut} problem and provide the proof of \Cref{thm:mc-w-hardness}.
We present the polynomial-time algorithm for {\sc Multiway Cut} on chordal graphs in \Cref{sec:multiway-cut}.
In \Cref{sec:Hlchordal}, we revisit the aforementioned problems by restricting the input to $H_\ell$-induced-subgraph-free chordal graphs and prove \Cref{thm:dom-set-ind-sub-w-hard} for {\sc Dominating Set} as well as \Cref{thm:mc-para-NP-hard}.
Finally in \Cref{sec:otherdompb}, we consider the {\sc Connected Dominating Set} and the {\sc Steiner Tree} problems and complete the proofs of \Cref{thm:ds-fpt} and \Cref{thm:dom-set-ind-sub-w-hard}.
\section{Preliminaries}
\label{sec:prelim}

For a positive integer $q$, we denote the set $\{1, 2, \dots, q\}$ by $[q]$
and for any $0 \leq p \leq q$, we denote the set $\{p,\ldots,q\}$ by $[p,q]$.
We use $\mathbb{N}$ to denote the set of all non-negative integers.
Given a function $f : X \to Z$ and $Y \subseteq X$,
$f|_{Y}$ denotes the function $f$ restricted to $Y$.

\paragraph*{Graph Theory.}
For a graph $G$, we denote by $V(G)$ and $E(G)$
the set of vertices and edges of $G$, respectively.
Unless specified otherwise, we use $n$ to denote the number of vertices in $G$.
{We denote the edge with endpoints $u, v$ by $uv$}.
For any $v \in V(G)$,
$N_G(v) = \{u \mid uv \in E(G)\}$ denotes the (open) neighbourhood of $v$,
and $N_G[v]=N_G (v) \cup \{v\}$ denotes the closed neighbourhood of $v$.
When the graph $G$ is clear from the context, we omit the subscript $G$.
For any $S \subseteq V(G)$, $G - S$ denotes the graph obtained from $G$ by deleting vertices in $S$.
We denote the subgraph of $G$ induced by $S$, i.e., the graph $G - (V(G) \setminus S)$, by $G[S]$.
We say graph $G$ contains graph $H$ as in \emph{induced subgraph}
if $H$ can be obtained from $G$ by series of vertex-deletions.
Recall that for a directed graph $H$,
we denote by $N_H^+(v)$ the out-neighbors of $v \in V(H)$
and by $N_H^-(v)$ the in-neighbors of $v \in V(H)$.
If $H$ is clear from the context, we omit the subscript $H$.
Given a (directed) path $P$ in a graph $G$ and two vertices
$u,v \in V(P)$, we denote by $P[u,v]$ the subpath of $P$ from $u$ to $v$.
For any further notation from basic graph theory, we refer the reader to~\cite{DBLP:books/daglib/Diestel12}.

\paragraph*{Trees.}
A {\em tree} $T$ is a connected acyclic graph.
Consider a tree $T$ rooted at $\rr$.
We define function $\parent(t, T) : V(T) \setminus \{\rr\} \mapsto V(T)$ 
to specify unique parent of the nodes in rooted tree $T$.
For any node $t \in T$, we denote by $T_t$ the subtree rooted at $t$.
A {\em subdivided star} is a tree with at most one vertex of degree at least $3$
(in other words, it is a tree obtained by repeatedly subdividing the edges of a star graph).
The sets $V_{\geq 3}(T)$ and $V_{=1}(T)$
denote the set of vertices of degree at least $3$,
and of degree equal to $1$, respectively.
The set $V_{\geq 3}(T)$ is also called the set of {\em branching vertices} of $T$
and the set $V_{=1}(T)$ is called the set of {\em leaves} of $T$.
Note that $|V_{\geq 3}(T)| \le |V_{=1}(T)|-1$.
Any node of $T$ which is not a leaf is called \emph{internal}.

\paragraph*{Chordal graphs and Tree representations.}
A graph is called a chordal graph if it contains no induced cycle of length at least four.
It is well-known that chordal graphs can be represented as intersection graphs of subtrees in a tree, that is, 
for every chordal graph $G$, there exists a tree $T$ and a collection $\calM$ of subtrees of $T$ 
in one-to-one correspondence with $V(G)$ 
such that two vertices in $G$ are adjacent if and only if their corresponding subtrees intersect.
The pair $(T, \calM)$ is called a \emph{tree representation} of $G$.
For every $v \in V(G)$, we denote by $\mld(v)$ the subtree corresponding to $v$ and
refer to $\mld(v)$ as the {\em model} of $v$ in $T$.
Throughout this article, we use \emph{nodes} to refer to the vertices of the tree $T$
to avoid confusion with the vertices of the graph $G$.
Furthermore, we use the greek alphabet to denote nodes of $T$ and the latin alphabet to denote vertices of $G$.
For notational convenience, for any node $\alpha \in V(T)$ and edge $e \in E(T)$, we may abuse notation and write $\alpha \in \mld(v)$ in place of $\alpha \in V(\mld(v))$ as well as $e \in \mld(v)$ in place of $e \in E(\mld(v))$.

For every node $\alpha \in V(T)$,
we let $\ver(\alpha) = \{v \in V(G) \mid \alpha \in \mld(v)\}$,
that is, $\ver(\alpha)$ is the set of vertices in $G$ that contain the node $\alpha$ is their model.
A vertex $v \in V(G)$ whose model contains $\alpha$ may also be referred to as an \emph{$\alpha$-vertex}.
Similarly, for every edge $e \in E(T)$,
we let $\ver(e)=\{v \in V(G) \mid e \subseteq \mld(v)\}$,
that is, $\ver(e)$ is the set of vertices of $G$ that contain the edge $e$ in their model.
Given a subtree $T'$ of $T$, we denote by $G_{|T'}$ the subgraph of $G$ 
induced by those vertices $x \in V(G)$ such that $V(\mld(x)) \subseteq V(T')$.
If $T$ is rooted, then for each vertex $v \in V(G)$, we call the node in $\mld(v)$ that is closest to the root of $T$, 
the {\em topmost} node of $\mld(v)$ and denoted it by $\topnode_{\mld}(v)$. 
 
The \emph{leafage} of chordal graph $G$, denoted by $\lf(G)$, is defined as 
the minimum number of leaves in the tree of a tree representation of $G$.
A tree representation $(T,\calM)$ for $G$ such that the number of leaves in $T$ is $\lf(G)$,
can be computed in time $O(|V(G)|^3)$ \cite{DBLP:conf/esa/HabibS09}. 
Furthermore, the number of nodes in $T$ is at most $O(|V(G)|)$.

\paragraph*{Parameterized Complexity.}
The input of a parameterized problem comprises an instance $I$,
which is an input of the classical instance of the problem,
and an integer $k$, which is called the parameter.
A parameterized problem $\Pi$ is said to be \emph{fixed-parameter tractable} (\FPT\ for short) %
if given an instance $(I,k)$ of $\Pi$,
we can decide whether $(I,k)$ is a \yes-instance of $\Pi$
in time $f(k)\cdot |I|^{\OO(1)}$
for some computable function $f$ depending only on $k$.
We say that an instance $(I, k)$ of a parameterized problem $\Pi$
and an instance $(I', k')$ of a parameterized problem $\Pi'$ (possibly $\Pi = \Pi'$)
are \emph{equivalent} 
if $(I, k) \in \Pi$ if and only if $(I', k') \in \Pi'$.
 A \emph{reduction rule}, for a parameterized problem $\Pi$,
 is a polynomial-time algorithm that takes as input an instance $(I, k)$ of $\Pi$
 and outputs an instance $(I', k')$ of $\Pi$.
 If $(I, k)$ and $(I', k')$ are equivalent
 then we say that the reduction rule is \emph{safe}.
For more details on parameterized algorithms, and in particular parameterized branching algorithms,
we refer the reader to the book by Cygan et al.~\cite{DBLP:books/sp/CyganFKLMPPS15}.
\section{Dominating Set}
\label{sec:dominating-set}

For a graph $G$, a set $X \subseteq V(G)$ is a \emph{dominating set} if every vertex in $V(G)\setminus X$ has at least one neighbor in $X$, that is, $V(G) = N[X]$.
In the \textsc{Dominating Set} problem (\DS for short),
the input is a graph $G$ and an integer $k$, and
the objective is to decide whether $G$ has a dominating set of size at most $k$.
We assume that the leafage of the input graph is given as part of the input.
If not, recall that it can be computed in polynomial time \cite{DBLP:conf/esa/HabibS09}.
We consider a generalized version of this problem as defined below.\\

\defproblem{
\textsc{Red-Blue Dominating Set (\RBDS)}
}{
A graph $G$, a partition $(R, B)$ of $V(G)$, and an integer $k$.
}{
Does there exist a set $X \subseteq R$ of size at most $k$
such that $B \subseteq N(X)$?
}

\bigskip

We first prove that to solve \DS, 
it is sufficient to solve \RBDS\ even when the input is restricted to chordal graphs of leafage $\ell$.
There is indeed a simple reduction from the former problem to the latter that preserves the properties in which we are interested.

\begin{lemma}
\label{lemma:ds-to-rbds-leafage}
There is a polynomial-time algorithm that given an instance $(G, k)$ of \DS\ constructs an equivalent instance $(G', (R',B'), k)$ of \RBDS\ such that if $G$ has leafage at most $\ell$, then so does $G'$.
\end{lemma}
\begin{proof}
  We construct $G'$ from $G$ as follows.
  For every vertex $v \in V(G)$,
  add two copies $v_R$ and $v_B$ to $V(G')$
  and add an edge $v_R v_B$ to $E(G')$.
  For every edge $uv \in E(G)$,
  add edges $v_R u_R$, $v_R u_B$, $v_B u_R$, and $v_B u_B$ to $E(G')$.
  This completes the construction of $G'$.
  Let $R' = \{v_R \mid v\in V(G)\}$ and $B' = \{v_B \mid v\in V(G)\}$.

  Suppose that the \DS instance has a solution $S \subseteq V(G)$.
  Then the set $S_R = \{ v_R \mid v\in S \}$,
  i.e.,\ $S_R$ contains the red version of each vertex in $S$,
  is a solution for the \RBDS instance:
  indeed, the blue vertices $v_B$ such that $v \notin S$
  are dominated since $S$ is a solution,
  and if $v \in S$, then $v_B$ is dominated because of the newly added edges.
  Conversely,
  if $S_R\subseteq R'$ is a solution for the \RBDS instance,
  then  it is easy to see that $S = \{ v \mid v_R \in S_R \}$ is a solution for the \DS instance.

  Finally note that a tree representation for $G'$ can be obtained from a tree representation 
  for $G$ by duplicating the model of each vertex, 
  and making the original model a model for the blue version of the vertex, and the copy a model for its red version. 
  In particular, the leafage of $G'$ is at most that of $G$. 
 \end{proof}
  
In the remainder of this section, we present an \FPT\ algorithm for \RBDS
when parameterized by the leafage $\ell$ of the input graph.
The algorithm consists of two parts.
In the first part, the algorithm constructs $2^{\calO(\ell)}$ many instances of a ``restricted version'' of the problem
such that the input instance is a \yes-instance
if and only if one of these newly created instances is a \yes-instance.
Moreover, the graphs in the newly created instances satisfy certain properties that allow us to design a fast algorithm.
See \cref{lemma:reduce-normalized-graph} for the formal statement.
In the second part (cf.\ \cref{lem:domSet:solvingNormalized}),
the algorithm solves the restricted version of \RBDS which is defined as follows.\\

  \vspace{1mm}
\noindent\fbox{
  \begin{minipage}{0.96\textwidth}
  \begin{tabular*}{\textwidth}{@{\extracolsep{\fill}}lr} \textsc{Restricted-Red-Blue Dominating Set (\RestRBDS)} \\ \end{tabular*}
  {\bf{Input:}} A chordal graph $G$, a partition $(R, B)$ of $V(G)$, an integer $k$ and tree representation $(T, \calM)$ of $G$ such that
\begin{itemize}
\item for every vertex in $G$, its model contains at most one branching node of $T$, and
\item  for all branching nodes $\gamma \in V(T)$,
    there are either only red $\gamma$-vertices
    or only blue $\gamma$-vertices.
\end{itemize}
  {\bf{Question:}} Does there exist a set $D \subseteq R$ of size at most $k$
such that $B \subseteq N(D)$?
  \end{minipage}
  }
  \vspace{1mm}

\subsection{Constructing \RestRBDS Instances}
In this section, we prove the following result.

\begin{lemma}
  \label{lemma:reduce-normalized-graph}
  Let $\calI=(G, (R, B), k)$ be an instance of \RBDS
  where $G$ is a chordal graph of leafage at most $\ell$.
We can construct, in time $2^{\calO(\ell)} \cdot n^{\Oh(1)}$, a collection $\{\calI_{i} = (G_{i}, (R_{i}, B_{i}), k) \mid i  \in [ 2^{\calO(\ell)}]\}$ of \RestRBDS instances such that 
  \begin{itemize}
    \item
     for every $i \in [2^{\calO(\ell)}]$, $G_{i}$ is a chordal graph of leafage at most $2\ell$, and 
    \item $\calI$ is a \yes-instance of \RBDS if and only if at least one of the instances in the collection is a \yes-instance of \RestRBDS. 
    \end{itemize}
\end{lemma}

\begin{proof}
Let $G$ be a chordal graph and let $(T,\calM)$ be a tree representation of $G$. 
We define the following functions.
\begin{itemize}
\item Let $f_T(G)$ denote the number of branching nodes $\gamma \in V(T)$ such that there exist both a red vertex and a blue vertex whose models contain $\gamma$.
\item Let $f_r(G)$ denote the number of pairs of consecutive branching nodes $\alpha, \beta$ in $T$ (that is, no node on the unique path in $T$ from $\alpha$ to $\beta$ is a branching node) such that there is red vertex whose model contains both $\alpha$ and $\beta$.
\item Similarly, let $f_b(G)$ denote the number of pairs of consecutive branching nodes $\alpha, \beta$ in $T$ such that there is blue vertex whose model contains both $\alpha$ and $\beta$.
\end{itemize}
We further define $\mu(G) := \lf(G) + 2 \cdot (f_T(G) + f_r(G) + f_b(G))$.
Note that, by definition, $\mu(G) \geq \lf(G)$.
We design a polynomial-time branching algorithm whose measure $\mu$ decreases in each branch.
We first show that if $\mu(G) = \lf(G)$ then $(G,(R,B),k)$ is in fact an instance of \RestRBDS
and then show how the branching algorithm proceeds.

Assume therefore that $\mu(G) = \lf(G)$.
Then $f_T(G) = f_r(G) = f_b(G) = 0$ by definition.
However, when $f_T(G) = 0$, then, by definition, for every branching node $\gamma \in V(T)$, all the vertices containing $\gamma$ in their model are either red or blue;
and when $f_r(G) = f_b(G) = 0$ then, considering the fact that every model is a subtree in $T$, for every vertex in $G$, its model contains at most one branching node in $T$.
Therefore if $\mu(G) = \lf(G)$, then $(G, (R, B), k)$ is also an instance of \RestRBDS. 

Now assume that $\mu(G) > \lf(G)$.
Then $f_T(G) + f_r(G) + f_b(G)  > 0$.
We consider the following three exhaustive cases.

\paragraph*{Case-I.} \emph{$f_T(G) > 0$.}
Let $\gamma$ be a branching node in $T$ such that there is both a red-vertex and a blue-vertex whose models contain $\gamma$.
Suppose that $\calI$ is a \yes-instance of \RBDS\ and let $D$ be a solution.
Consider first the case where $D$ includes a red vertex whose model contains $\gamma$.
In this case, we return the instance $\calI_1 = (G_1, (R_1, B_1), k)$ which is obtained as follows.
\begin{itemize}
\item Initialize $V(G_1) = V(G)$, $R_1 = R$,  $B_1 = B$.
\item Let $T_1$ be the tree obtained from $T$ by adding a node $\delta$ and making it adjacent to $\gamma$ only.
Note that $V(T_1) \setminus \{\delta\} \subseteq V(T)$.
\item For every red vertex $v \in V(G_1)$ such that $\gamma \in \mld(v)$, add $\delta$ to its model, i.e.,
$\mld_1(v) = \mld(v) \cup \{\delta\}$. 
\item For every blue vertex $v \in V(G_1)$ such that $\gamma \in \mld(v)$, delete $v$ from $V(G_1)$.
\item Add a new blue vertex $x$ to $V(G_1)$ and to $B_1$ with $\mld_1(x) = \{\delta\}$.
\item For every (red or blue) vertex $v \in V(G)$ such that $\gamma \not\in \mld(v)$, 
define $\mld_1(v_1) =  \mld(v)$.
\end{itemize}
It is easy to verify that $(T_1, \mld_1)$ is a tree representation of $G_1$ and that $T_1$ has exactly one more leaf than $T$, i.e., $\lf(G_1) \le \lf(G) + 1$.
However, since we have deleted all the blue vertices whose models contained $\gamma$, $f_T(G_1) = f_T(G) - 1$.
As the other parts of the measure do not change, $\mu(G_1) < \mu(G)$.

In the second case where no vertex in $D$ contains $\gamma$ in its model,
we return the instance $\calI_2 = (G_2, (R_2, B), k)$
where $G_2, R_2$ are obtained from $G, R$, respectively,
by deleting red vertices whose model contains $\gamma$.
It is easy to verify that $\mu(G_2) < \mu(G)$.

If $\calI$ is a \yes-instance,
then at least one of $\calI_1$ or $\calI_2$ is a \yes-instance as these two branches are exhaustive.
If $\calI_1$ is a \yes-instance, then any optimum solution must include a red $\gamma$-vertex because of the newly added vertex $x$.
As $R_2 \subseteq R$, if $\calI_2$ is a \yes-instance,
then $\calI$ is a \yes-instance.
Hence, this branching step is correct.

\paragraph*{Case-II.} \emph{$f_T(G) = 0$ and $f_r(G) > 0$.}
Let $\alpha, \beta$ be two consecutive branching nodes in $T$ such that there is a red vertex whose model contains both $\alpha$ and $\beta$.
Suppose that $\calI$ is a \yes-instance of \RBDS\ and let $D$ be a solution.
Consider the case where $D$ includes a red vertex whose model contains both $\alpha$ and  $\beta$.
In this case, we return the instance $\calI_1 = (G_1, (R_1, B_1), k)$ which is obtained as follows.
\begin{itemize}
\item Initialize $V(G_1) = R_1 = B_1 = \emptyset$.
\item Let $T_1$ be the tree obtained from $T$ by contracting the unique path $P_{\alpha\beta}$ from $\alpha$ to $\beta$ in $T$ and let $\gamma_{\alpha\beta}$ be the node resulting from this contraction.
Add a node $\delta$ to $T_1$ and make it adjacent to $\gamma_{\alpha\beta}$ only.
Note that $V(T_1) \setminus \{\gamma_{\alpha\beta}, \delta\} \subseteq V(T)$.
\item For every red vertex $v \in V(G)$ such that $\mld(v) \cap V(P_{\alpha\beta}) \neq \emptyset$, add a red vertex $v_1$ to $V(G_1)$ (and to $R_1$) with $\mld_1(v_1) = (\mld(v) \setminus V(P_{\alpha\beta})) \cup \{\gamma_{\alpha\beta}, \delta\}$.
\item Add a new blue vertex $x$ to $V(G_1)$ with $\mld_1(x) = \{\delta\}$.
\item For every (red or blue) vertex $v \in V(G)$ such that $\mld(v) \cap V(P_{\alpha\beta}) = \emptyset$, add $v_1$ to $G_1$ (and to, respectively, either $R_1$ or $B_1$) with $\mld_1(v_1) =  \mld(v)$. 
\end{itemize}
Note that for every blue vertex $v \in V(G)$ such that $\mld(G) \cap V(P_{\alpha\beta}) \neq \emptyset$, there is no corresponding blue vertex in $G_1$.
It is easy to verify that $(T_1, \mld_1)$ is a tree representation of $G_1$ and that $T_1$ has one more leaf than $T$ which implies $\lf(G_1) \le \lf(G)  + 1$.
Since we have contracted the path $P_{\alpha\beta}$ to obtain the node $\gamma_{\alpha\beta}$, 
$f_r(G_1) < f_r(G)$.
As the other parts of the measure do not change, $\mu(G_1) < \mu(G)$.

In the second case where no vertex in $D$ contains both $\alpha$ and $\beta$ in its model,
we return an instance $\calI_2 = (G_2, (R_2, B), k)$ where $G_2, R_2$ are obtained from $G,R$, respectively,
by deleting red vertices whose model contains both $\alpha$ and $\beta$.
It is easy to verify that $\mu(G_2) < \mu(G)$.
We argue as in the previous case for the correctness of this branching steps.

\paragraph*{Case-III.} \emph{$f_T(G) = 0$ and $f_b(G) > 0$.}
Let $\alpha, \beta$ be two consecutive branching nodes in $T$ such that there is a blue vertex whose model contains both $\alpha$ and $\beta$.
Note that since $f_T(G) = 0$, for every red vertex $v \in V(G)$ such that $\mld(v) \cap V(P_{\alpha\beta}) \neq \emptyset$, in fact $\mld(v) \subseteq V(P_{\alpha\beta}) \setminus \{\alpha, \beta\}$.
Suppose that $\calI$ is a \yes-instance of \RBDS and let $D$ be a solution.
Consider first the case where $D$ includes a red vertex whose model is in $V(P_{\alpha\beta})\setminus \{\alpha, \beta\}$.
In this case, we return the instance $\calI_1 = (G_1, (R, B_1), k)$ where $G_1$, $B_1$ are obtained from $G$ and $B$ as follows.
\begin{itemize}
\item Delete all the blue vertices whose model contains both $\alpha$ and $\beta$.
\item Add a blue vertex $x$ to $V(G_1)$ (and to $B_1$) with $\mld(x) = V(P_{\alpha\beta})\setminus \{\alpha, \beta\}$. 
\end{itemize}
It is easy to verify that $(T, \mld)$ is a tree representation of $G_1$ and $f_b(G_1) < f_b(G)$.
As the other parts of the measure do not change, $\mu(G_1) < \mu(G)$.

In the second case where there is no vertex in $D$ whose model is in $V(P_{\alpha\beta}) \setminus \{\alpha, \beta\}$, we consider the following two subcases.
If there is a blue vertex $v$ such that $\mld(v) \subseteq V(P_{\alpha\beta})$, then we return a trivial \no-instance.
Otherwise, we return the instance $\calI_2 = (G_2, (R_2, B_2), k)$ which is constructed as follows.
\begin{itemize}
\item Initialize $V(G_2) = R_2 = B_2 = \emptyset $.
\item Let $T_2$ be the tree obtained from $T$ by contracting the path $P_{\alpha\beta}$ from $\alpha$ to $\beta$ in $T$ and let $\gamma_{\alpha\beta}$ be the node resulting from this contraction.
Note that $V(T_2) \setminus\{\gamma_{\alpha\beta} \} \subseteq V(T)$.
\item For every (red or blue) vertex $v \in G$ such that $\mld(v) \cap V(P_{\alpha\beta}) = \emptyset$, add a vertex $v_2$ to $G_2$ (and to, respectively, either $R_2$ or $B_2$) with $\mld_2(v_2) = \mld(v)$.
\item For every blue vertex $v \in V(G)$
such that $\mld(v) \cap V(P_{\alpha\beta}) \neq \emptyset$,
add a blue vertex $v_2$ to $V(G_2)$ (and to $B_2$) with
$\mld_2(v_2) = (\mld(v_2) \setminus V(P_{\alpha\beta})) \cup \{\gamma_{\alpha\beta}\}$.
\end{itemize}
Note that for any red vertex $v \in V(G)$ such that $\mld(v) \subseteq V(P_{\alpha\beta}) \setminus \{\alpha, \beta\}$, there is no corresponding red vertex in $G_2$.
It is easy to verify that $(T_2, \mld_2)$ is a tree representation of $G_2$. 
Furthermore, the number of leaves of $T_2$ is the same as $T$ and $f_b(G_2) < f_b(G)$.
As the other parts in the measure do not change, $\mu(G_2) < \mu(G)$.

The correctness of this branching step follows from the same arguments
as in the previous cases and the fact that in the second case,
since there is no red vertex whose model intersects $V(P_{\alpha\beta})$,
it is safe to contract that path.

\paragraph*{Finishing the Proof.}
The correctness of the overall algorithm follows from the correctness of branching steps in the above three cases.
To bound its running time and the number of instances it outputs, note that $f_T(G) + f_r(G) + f_b(G) \le 3\cdot \lf(G)$ as these functions either count the number of branching nodes or the unique paths containing exactly two (consecutive) branching nodes.
\end{proof}

\subsection{Solving an instance of \RestRBDS}
\label{sec:domSet:leafage:noramlized}

In this section, we present an algorithm to solve \RestRBDS.
Formally, we prove the following lemma.

\begin{lemma}
  \label{lem:domSet:solvingNormalized}
{\sc \RestRBDS} admits an algorithm running in time 
$2^{\Oh(\ell)} \cdot n^{\Oh(1)}$.
\end{lemma}

We first state some easy reduction rules before we handle two cases based on whether the farthest branching node\footnote{We assume that the tree in the tree representation is rooted and thus, by farthest branching node, we mean farthest from the root.} is contained only in the models of red vertices or blue vertices.
We present \Cref{greedy-select:red} and \Cref{greedy-select:blue} to handle these cases.
The proof of the lemma follows from \Cref{lemma:greedy-select:red}, \Cref{lemma:greedy-select:blue} and the fact that each application of the greedy selection procedure deletes some vertices in the graph.

We first introduce some notations.
Recall that an instance of \RestRBDS
contains a chordal graph $G$, a partition $(R, B)$ of $V(G)$, an integer $k$ and tree representation $(T, \calM)$ of $G$ such that
for every vertex in $G$, its model contains at most one branching node of $T$, and
for all branching nodes $\gamma \in V(T)$,
there are either only red $\gamma$-vertices
or only blue $\gamma$-vertices.
We assume, without loss of generality, that the tree $T$ is rooted at node $\rr$.
Unless mentioned otherwise, $\alpha$ denotes the farthest branching node in $T$ from the root,
that is, each proper subtree of $T_\alpha$ is a path.
If there are more than one branching node that satisfy the property, we arbitrarily select one of them.
Let $\beta$ be the closest branching ancestor of $\alpha$, that is,
no internal node in the unique path from $\alpha$ to $\beta$ is a branching node in $T$.%
\footnote{
If $\alpha$ is the root of the tree,
then we can add an artificial new root $\beta$
which is not contained in the model of any vertex.
}
Recall that for a vertex $v \in V(G)$, we define $\topnode_\mld(v)$
as the node $\eta \in \mld(v)$ that is closest to the root.
Likewise if a leaf $\lambda$ is fixed, we define $\botnode^{\lambda}_\mld(v)$
as the node $\eta \in \mld(v)$ that is closest to $\lambda$.
For ease of notation, we omit $\lambda$ as it is always clear from the context.

\begin{definition}
  Let $\gamma$ be a node of the tree $T$.
  We define the following sets of vertices in $G$.
  \begin{itemize}
    \item $\interB{\gamma}, \interR{\gamma}, \interV{\gamma}$
    are the sets of, respectively, blue, red, all vertices $v \in V(G)$
    whose models intersect the tree rooted at $\gamma$,
    i.e.,\ $\mld(v) \cap V(T_\gamma) \neq \emptyset$.

    \item $\belowB{\gamma}, \belowR{\gamma}, \belowV{\gamma}$
    are the sets of, respectively, blue, red, all vertices $v \in V(G)$
    whose models are completely contained inside the tree rooted at $\gamma$,
    i.e.,\ $\mld(v) \subseteq V(T_\gamma)$.

    \item $\underB{\gamma}, \underR{\gamma}, \underV{\gamma}$
    are the sets of blue, red, all vertices $v \in V(G)$
    where the model is completely contained inside the tree rooted at $\gamma$ but does not contain $\gamma$,
    respectively,
    i.e.\ $\mld(v) \subseteq V(T_\gamma^\dag) = V(T_\gamma) \setminus \{\gamma\}$.

    \item $\containB{\gamma}, \containR{\gamma}, \containV{\gamma}$
    are the sets of, respectively, blue, red, all vertices $v \in V(G)$
    whose models contains $\gamma$,
    i.e.,\ $\gamma \in \mld(v)$.
  \end{itemize}
\end{definition}

\paragraph*{Simplifications.}
We first apply the following easy reduction rules
whose correctness readily follows from the definition of the problem.
It is also easy to see that the reduction rules can be applied in polynomial time and the reduced instance is also a valid instance of \RestRBDS.

\begin{reduction rule}
  \label{rr:domSet:blueNeedNeighbors}
  \label{rr:domSet:negativeBudget}
  If there is a blue vertex, which is not adjacent to a red vertex,
  or if $k < 0$,
  then return a trivial \no-instance.
\end{reduction rule}

\begin{reduction rule}~
  \label{rr:domSet:noTwins}
\begin{itemize}
\item If there are two blue vertices $u, v$ such that $\mld(u) \subseteq \mld(v)$, then delete $v$.
\item If there are two red vertices $u, v$ such that $\mld(u) \subseteq \mld(v)$, then delete $u$.
 \end{itemize}
\end{reduction rule}

Consider a blue vertex $v$ in $G$ whose model is contained in the subtree rooted at $\alpha$. 
Moreover, let $v$ be such a vertex for which $\topnode_\mld(v)$ is farthest from the root
and $v$ is not adjacent to a red vertex whose model contains $\alpha$.
Hence, there is a natural ordering amongst the red neighbors of $v$.
Note that such an ordering is not possible if some of its neighbors contain $\alpha$ in their models.
As any solution contains a red neighbor of $v$, it is safe to include its neighbor $v_r$ for which $\topnode_\mld(v_r)$ is closest to $\alpha$.

\begin{reduction rule}
  \label{rr:domSet:greedilySelectReds}
  Suppose that there is a blue vertex $v \in \underB{\alpha}$ 
  such that $\topnode_\mld(v)$ is farthest from the root and $v$ is not adjacent to any red $\alpha$-vertices.
  Moreover, amongst all the red neighbors of $v$, let $v_r$ be the node such that $\topnode_\mld(v_r)$ is closest to $\alpha$.
  Then, remove $v_r$ and all of its blue neighbors and 
  decrease $k$ by $1$.
\end{reduction rule}

We remark that the above reduction rule is applicable irrespective of the fact
whether either all $\alpha$-vertices are red or all $\alpha$-vertices are blue.

  \newcommand{\OPTlocal}{\OPT(\calI_\alpha)}
  \newcommand{\OPTglobal}{\OPT}

\paragraph*{Case-1: All the vertices that contain $\alpha$ in their models are red vertices.}
Let $\beta$ be the closest branching ancestor of $\alpha$.
Consider the blue vertices whose model intersect the path from $\alpha$ to $\beta$.
Note that there may not be any such blue vertex; however, we find it convenient to present an uniform argument. 
With a slight abuse of notation, let $b_1,\dots,b_d$ be these blue vertices
  ordered according to their endpoint in the direction of $\alpha$,
  that is, for $i<j$ we have either  $\botnode_\mld(b_i) =\botnode_\mld(b_j)$ or $\botnode_\mld(b_i)$ is closer to $\alpha$
  than $\botnode_\mld(b_j)$.
For each $i \in [d]$,
  we compute an optimal solution for dominating
  the vertices whose model is in the tree rooted at $\alpha$ (i.e., the vertices of $\underB{\alpha}$) and the vertex $b_i$
  while only using red $\alpha$-vertices.
  Formally, we want to compute an optimal solution for the following instance: $\calI_i \deff G[\interR{\alpha} \cup \belowB{\alpha} \cup \{b_i\}]$.
We also define instance 
$\calI_0 \deff G[\interR{\alpha} \cup \belowB{\alpha}]$ to handle the cases when there are no blue vertices whose model intersects the path from $\alpha$ to $\beta$ or 
when $b_1$ (and hence, the other blue vertices mentioned above) are not dominated by red $\alpha$-vertices in an optimum solution.
To simplify notation we set $\OPT_i \deff \OPT(\calI_i)$ in the following.
If $\calI_i$ is not defined, then we set $\OPT_i = \infty$.
Note that the solution $\OPT_i$ also dominates the blue vertices
  $b_1,\dots,b_{i-1}$ due to the ordering of the $b_i$s.
Hence, for any $i, j \in [0,d]$ such that $i < j$, we have $\abs{\OPT_i} \le \abs{\OPT_j}$.
We use this monotonicity to prove the following structural lemma.

\begin{lemma}
\label{lemma:red-alpha-greedy}
Let $q \in [0,d]$ be the largest value such that
$\abs{\OPT_q} = \abs{\OPT_0}$.
If there is a solution,
then there is an optimum solution containing $\OPT_q$.
\end{lemma}

  \begin{proof}
Let $\OPT$ be an optimum solution of $(G, (R, B), k)$.
Let $S$ denote the collection of vertices in $\OPT$ whose model contains nodes in the subtree rooted at $\alpha$, 
i.e., $S \deff \OPT \cap \interR{\alpha}$.
We claim that  we can replace $S$ by a super-set $S'$ of $\OPT_q$
    of equal size to obtain another solution.

    Let $j\in[0,d]$ be the largest integer such that $b_j$ is dominated by some vertex in $S$.
If $j \le q$,  then by our choice of $q$, $\abs{S} = \abs{\OPT_q}$.
    By the definition of the $\calI_i$s,
    we get that $\OPT_q$ is also a solution for $\calI_i$.
    Hence, we can replace $S$ by $\OPT_q$
    to get another optimal solution.
Suppose therefore that  $j > q$.
By our choice of $q$, we have $\abs{S} > \abs{\OPT_q}$.
    Let $r_j$ be the red $\alpha$-vertex
    with $\topnode_\mld(r_j)$ closest to $\beta$
    such that $b_j$ is a neighbor of $r_j$.
    Such a vertex exists, as by assumption, $S$ contains one of these vertices which dominates $b_j$.
    Then we replace $S$ by $S' = \OPT_q \cup \{ r_j \}$.
    As $\abs{S} > \abs{\OPT_q}$, we have $\abs{S'} \le \abs{S}$.
   	Moreover, observe that $S' \cup \OPT \setminus S$ is still a solution as 
    all vertices in $\underB{\alpha}$ and the vertices $b_1,\dots,b_q$
    are dominated by some vertex in $\OPT_q$, 
    vertex $r_j$ dominates the vertices $b_{q+1},\dots,b_{j}$
    and, by the choice of $j$, the vertices $b_{j+1},\dots,b_d$
    are dominated by some vertex not contained in~$S$.
  \end{proof}   
  
 We devise a greedy selection step based on the above lemma.
 
  \begin{greedy select}
  \label{greedy-select:red}
   Let $q \in [0,d]$ be the largest value such that
    $\abs{\OPT_q} = \abs{\OPT_0}$.
Include the vertices of $\OPT_q$ in the solution, i.e., delete the red vertices in $\OPT_q$, the blue vertices that are adjacent to vertices in $\OPT_q$, and
    decrease $k$ by $\abs{\OPT_q}$.
  \end{greedy select}

 \begin{lemma}
 \label{lemma:greedy-select:red}
 \Cref{greedy-select:red} step is correct and can be completed it time 
 $2^{\calO(\ell)} \cdot \polyn$.
 \end{lemma} 
  
  \begin{proof}
The correctness of the step follows directly from Lemma~\ref{lemma:red-alpha-greedy}.
In the remaining proof, we show how to compute, for every $i \in [0,d]$, $\OPT_i$ in time
 $\Oh(\ell \cdot \abs{R} \cdot 2^\ell \cdot n)$ by constructing an instance of \SetCov.
Before constructing  such an instance, we justify that only one
blue vertex (which is farthest from $\alpha$) is critical while constructing this \SetCov instance.

Let $\alpha'$ be a child of $\alpha$.
As $\alpha$ is a farthest branching node of $T$ from the root, the tree rooted at $\alpha'$ is a path.
Let $\lambda$ be the another endpoint of this path.
Consider a blue vertex $v_{\alpha'}$ whose model is contained in $T_{\alpha'}$, i.e., $v_{\alpha'} \in \belowB{\alpha'}$.
Moreover, suppose that $v_{\alpha'}$ is the vertex for which $\topnode_\mld(v_r)$ is farthest from $\alpha'$.
As Reduction Rule~\ref{rr:domSet:greedilySelectReds} is not applicable, there exists at least one red neighbor of $v_{\alpha'}$ which is an $\alpha$-vertex. 
Hence, an optimum solution can always include a red neighbor of $v_{\alpha'}$ which is also an $\alpha$-vertex.
This red $\alpha$-vertex also dominates all the blue vertices in $\belowB{\alpha'}$.

We now explain how to construct an instance $(U, \calF)$ of \SetCov.
For every child $\alpha'$ of $\alpha$, if the vertex $v_{\alpha'}$ mentioned in the previous paragraph exists, then add an element $u_{\alpha'}$ corresponding to it to $U$. 
When $i \neq 0$, add another element $u_i$  corresponding to $b_i$ to $U$. 
For every red $\alpha$-vertex $v$, we define set $S_v \subseteq U$ as the collection of elements corresponding to the blue vertices in $\calI_i$ that are adjacent to $v$.
This completes the construction of the instance.

It is easy to see the one-to-one correspondence between the optimum solutions of these two instances.
The running time of the algorithm follows from the known algorithms for \SetCov (see, for instance, \cite{DBLP:conf/wg/FominKW04}) and the fact that $\alpha$ has at most $\ell$ many children.
\end{proof}

  \renewcommand{\OPTlocal}{\OPT(\calI_\gamma)}
  \newcommand{\local}{\calI_\gamma}

\paragraph*{Case-2: All the vertices that contain $\alpha$ in their models are blue vertices.}
  Let $\beta$ be the closest branching ancestor of $\alpha$.
We consider two cases depending on whether there is a red vertex whose model intersects the path from $\alpha$ to $\beta$.
If there is no such red vertex, then we consider the graph induced by all the red vertices whose model is (properly) contained in the subtree rooted at $\alpha$ and the blue vertices whose model intersects the subtree rooted at $\alpha$.
Formally, we define $\calI_0 = G[\belowR{\alpha} \cup \interB{\alpha}]$.
  
Consider the other case and suppose that there are $d \ge 1$ many red vertices whose model intersects the path from $\alpha$ to $\beta$.
Let $r_1,\dots,r_d$ be these vertices ordered according to their endpoints in the direction of $\alpha$, that is, for $i < j$, we have either $\botnode_\mld(r_i) = \botnode_\mld(r_j)$ or 
$\botnode_\mld(r_i)$ is closer to $\alpha$ than $\botnode_\mld(r_j)$.
For each such red vertex $v_i$, we compute the optimal solution
  to dominate the vertices in $\interB{\alpha}$
  by vertices in $\belowR{\alpha}$ assuming that $v_i$ is already selected.
  Note that we only have to focus on the blue vertices in $\interB{\alpha}$
  which are not adjacent to $v_i$.
  Formally we define $
    \calI_i
      = G[\belowR{\alpha} \cup (\interB{\alpha}\setminus N[v_i])]$.
It is possible that the optimum solution does not include any of the vertices in $\{r_1, r_2, \dots, r_d\}$.
To handle this case, we define $
    \calI_{d + 1}
      = G[\belowR{\alpha} \cup \interB{\alpha}]$.
To simplify notation, we set $\OPT_i \deff \OPT(\calI_i)$ in the following.
Note that for the instance defined above, $R_i$ is same for every instance whereas $B_{i} \subseteq B_{i + 1}$ because of the ordering.
Hence, for any $i , j \in [d + 1]$ such that $i < j$, we have $|\OPT_i| \le |\OPT_j|$.
We use this monotonicity to prove the following structural lemma.

  \begin{lemma}
  \label{lemma:blue-alpha-greedy}
  If there is a red vertex whose model intersects the path from $\alpha$ to $\beta$, let $q \in [d + 1]$ be the largest value such that
    $\abs{\OPT_q} = \abs{\OPT_1}$.
    Otherwise, define $\OPT_q = \OPT_0$.   
    If there is a solution for the instance,
    then there is an optimum solution $\OPT$ such that
    $\OPT \cap \belowR{\alpha} = \OPT_q$.
  \end{lemma}
  
  \begin{proof}
If there is no red vertices whose model intersects the path from $\alpha$ to $\beta$, then all the red vertices in $G$ that are adjacent to blue vertices in $\calI_0$ are the red vertices in $\calI_0$.
Hence, the statement of the lemma follows.

We now consider the case where there are red vertices whose model intersects the path from $\alpha$ to $\beta$.
Let $\OPT$ be an optimum solution of $(G, (R, B), k)$.
Let $S$ denote the collection of vertices in $\OPT$ whose model is (properly) contained in the subtree rooted at $\alpha$, i.e., $S \deff \OPT \cap \underR{\alpha}$.
We claim that we can replace $S$ by a super-set $S'$ of $\OPT_q$ of equal size to obtain another optimum solution.

Let $j \in [d]$ be the smallest index
    such that $v_j$ is contained in $\OPT$.
Note that, by definition, $j \neq d + 1$ as there are only $d$ red vertices with the said property.
If $j \le q$, then by our choice of $q$, $\abs{S} \ge \abs{\OPT_j}$.
By the definition of $\calI_j$ and the fact blue vertices in $\calI_j$ are subset of blue vertices in $\calI_q$,
    $\OPT_q$ is also a solution for $\calI_j$.
    Hence, we can replace $S$ by $\OPT_q$ to get another optimal solution.
    Suppose therefore that $j > q$.
    By our choice of $q$, 
    we have $\abs{\OPT_j} > \abs{\OPT_q}$.
    As $\OPT$ is a solution,
    all vertices in $\interB{\alpha}$ must be covered by $\OPT$.
    Hence, we can replace $S$ by $S'=\OPT_q \cup \{ r_q \}$
    and get a solution of not larger size
    which still dominates all vertices in $\interB{\alpha}$.
    Indeed, the vertices which are not dominated by $\OPT_q$
    are dominated by $r_q$.
    \end{proof}
    
  We devise a greedy selection step based on the above lemma.
  
  \begin{greedy select}
  \label{greedy-select:blue}
    If there is a red vertex whose model intersects the path from $\alpha$ to $\beta$, let $q \in [d + 1]$ be the largest value such that
    $\abs{\OPT_q} = \abs{\OPT_1}$.
    Otherwise, define $\OPT_q = \OPT_0$. 
    Include $\OPT_q$ in the solution, i.e., delete the red vertices in $\OPT_q$, the blue vertices that are adjacent to vertices in $\OPT_q$, and decrease $k$ by $\abs{\OPT_\ell}$.
  \end{greedy select}
	
\begin{lemma}
\label{lemma:greedy-select:blue}
\Cref{greedy-select:blue} step is correct and can be completed in time $2^{\calO(\ell)} \cdot \polyn$.
\end{lemma}  

\begin{proof}
The correctness of the step follows directly from Lemma~\ref{lemma:blue-alpha-greedy}.
In the remaining proof, we show how to compute $\OPT_i$ for every $i \in [0, d + 1]$, by constructing an instance of \HitSet.
As in Lemma~\ref{lemma:greedy-select:red}, we first argue that only one red vertex (which is closest to $\alpha$) is critical while constructing a \HitSet instance.
 
Recall that, by assumption, none of the previous reduction rules are applicable. 
As in the previous case, let $\alpha'$ be a child of $\alpha$.
We first argue that there are no blue vertices
whose path is completely contained in the path rooted at $\alpha'$.
Assume, for the sake of contradiction, that there exists such a blue vertex $v$.
As \cref{rr:domSet:blueNeedNeighbors} is not applicable, $v$ is adjacent to at least one red vertices.
However, since all $\alpha$-vertices are blue, by the property of the instance, there are no red $\alpha$-vertices.
In particular, $v$ is not adjacent to any red $\alpha$-vertex.
This contradicts the fact that \cref{rr:domSet:greedilySelectReds} is not applicable.
Hence, there is no blue vertex whose model is contained in the path rooted at $\alpha'$.
Since this is true for any child of $\alpha$, there are no blue vertices in $\underB{\alpha}$, i.e., $\underB{\alpha}=\emptyset$ and $\interB{\alpha}=\containB{\alpha}$.
Now, for a child $\alpha'$ of $\alpha$, let $v_{\alpha'} \in \belowR{\alpha'}$ be a red vertex such that
    $\topnode_\mld(v_r)$ is closest to $\alpha$.
Since all blue vertices contain $\alpha$ in their model, the only critical red vertex in this leg is $v_{\alpha'}$.
 
We now explain how to construct an instance $(U, \calF)$ of \HitSet.
For every child $\alpha'$ of $\alpha$, let $v_{\alpha'}$ be the vertex as mentioned above.
Add an element $u_{\alpha'}$ corresponding to $v_{\alpha'}$ in $U$.
For every blue $\alpha$-vertex $v$, we define $S_v \subseteq U$ as the collection of elements corresponding to the red vertices in $\calI_i$ that are adjacent to $v$.
This completes the construction of the instance.

It is easy to see the one-to-one correspondence between the optimum solutions of these two instances.
The running time of the algorithm follows from the simple brute-force algorithm for \HitSet parameterized by the size of the universe and the fact that $\alpha$ has at most $\ell$ many children. 
\end{proof}

\section{Multicut with Undeletable Terminals}
\label{sec:multicut-bounded-leafage}

\tikzset{
  circ/.style = {circle,draw,fill,inner sep=1.3pt},
  circb/.style = {circle,draw=teal,fill=teal,text=teal,inner sep=.8pt},
  circr/.style = {circle,draw=purple,fill=purple,inner sep=.8pt},
  circg/.style = {circle,draw=olive,fill=olive,inner sep=.8pt},
  circlg/.style = {circle,draw=gray,fill=gray,inner sep=.8pt},
  scirc/.style = {circle,draw,fill,inner sep=.8pt},
}

This section considers the \textsc{MultiCut with Undeletable Terminals} problem formally defined as follows.\\

\defproblem{ \textsc{MultiCut with Undeletable Terminals} (\MCundel)}{%
An undirected graph $G$,
a set $\Pairs \subseteq V(G) \times V(G)$,
and an integer $k$.
}{%
Is there a set $S \subseteq V(G) \setminus V(\Pairs)$ such that
$\abs{S} \le k$ and
for all $(p,p') \in \Pairs$,
there is no path between $p$ and $p'$ in $G-S$?
}

\bigskip

In the following, a set $S \subseteq V(G) \setminus V(\Pairs)$ such that for all $(p,p') \in \Pairs$,
there is no path between $p$ and $p'$ in $G-S$ is called a \emph{$P$-multicut} in $G$.
We first prove that when the input is restricted to chordal graphs,
the problem is unlikely to admit an \FPT~algorithm
when parameterized by the leafage.
We then complement this result with an \XP-algorithm parameterized by the leafage.
We restate the theorem with the precise statement for the reader's convenience.

\thmmcleafagewhard*

To prove that the problem is \W[1]-hard,
we present a parameter preserving reduction from \textsc{Multicolored Clique}.
An instance of this problem consists of 
a simple graph $G$, an integer $q$, and a partition $(V_1, V_2, \dots, V_q)$ of $V(G)$.
The objective is to determine whether there is a {clique} in $G$ that contains exactly one vertex from each part $V_i$.
Such a clique is called a \emph{multicolored clique}.
We assume, without loss of generality, that each $V_i$ is an independent set
and that $|V_1| = \ldots = |V_q| = n$.%
\footnote{
Unlike in the rest of the article,
we \emph{do not} use $n$ to denote the total number of vertices in $G$
to keep notation simple while presenting the reduction.
}
This implies, in particular, that $|E(G)| < n^2 \cdot q^2$.
For every $i \in [q]$, we denote by $v^i_1,\ldots,v^i_n$ the vertex set of $V_i$
and for every $i \neq j \in [q]$,
we denote by $E_{i,j} \subseteq E(G)$ the set of edges between $V_i$ and $V_j$.
We define $M \deff (n + 1)^2 \cdot q^2$.

\begin{figure}[t]
  \centering
  \begin{tikzpicture}
    \node[circ,label={[label distance=.15cm]below: $p_{1}$}] (pn1) at (-0.1,.75) {};

    \node[draw=none] at (1.35,.75) { $K_1$};
    \draw (0.9,0.5) rectangle (1.8,1.0);

    \draw[very thick] (pn1) -- (0.9,0.5)
    (pn1) -- (0.9,1.0);

    \node[circ,label={[label distance=.15cm]below: $p_2$}] (pn2) at (2.8,.75) {};
    \draw[very thick] (pn2) -- (1.8,1.0)
    (pn2) -- (1.8,0.5);

    \node[draw=none] at (4.25,.75) { $K_{2}$};
    \draw (3.8,0.25) rectangle (4.7,1.25);

    \draw[very thick] (pn2) -- (3.8,1.25)
    (pn2) -- (3.8,0.25);

  \node[circ,label={[label distance=.15cm]below: $p_3$}] (pn2) at (5.8,.75) {};

  \draw[very thick] (pn2) -- (4.7,1.25)
    (pn2) -- (4.7,0.25);

  \node[draw=none] at (6.75,.75) {$\cdots$};

  \node[circ,label={[label distance=.15cm]below: $p_{n}$}] (pn) at (7.5,.75) {};

  \draw[very thick] (pn) -- (8.7,0)
    (pn) -- (8.7,1.5);

  \node[draw=none] at (9.15,.75) { $K_{n}$};
    \draw (8.7,0) rectangle (9.6,1.5);

    \node[circ,label={[label distance=.15cm]below: $p_{n+1}$}] (pnp1) at (10.6,.75) {};
    \draw[very thick] (pnp1) -- (9.6,1.5)
    (pnp1) -- (9.6,0);
  \end{tikzpicture}
  \caption{The auxiliary graph $B$. Rectangles represent cliques and thick edges indicate that the corresponding vertex is complete to the corresponding cliques.}
  \label{fig:auxB}
\end{figure}
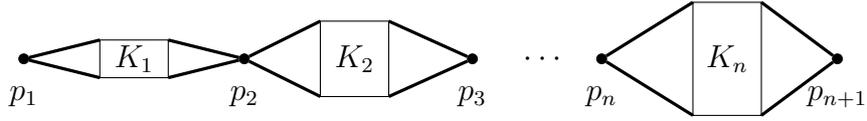

\paragraph*{Reduction.}
The reduction takes as input an instance $(G, q, (V_1, \dots, V_q))$ of \textsc{Multicolored Clique}
and outputs an instance $(H, \Pairs, k)$ of \MCundel which is constructed as follows.
\begin{itemize}
  \item The reduction starts by constructing an auxiliary graph $B$.
  The vertex set of $B$ consists of $n+1$ vertices $p_1,\ldots,p_{n+1}$
  and $n$ vertex-disjoint cliques $K_1,\ldots,K_{n}$
  such that $|K_a| = a \cdot M$ for every $a \in [n]$.
  Then, it adds edges so that $p_1$ is complete to $K_1$,
  $p_{n+1}$ complete to $K_{n}$, and $p_a$ complete to $K_{a-1} \cup K_a$
  for every $a \in [n] \setminus \{1\}$.
  This completes the construction of $B$ (see \Cref{fig:auxB}).

  \item
  For each $i \in [q]$, the reduction introduces two vertex-disjoint copies $B^{i,\alpha}$ and $B^{i,\beta}$ of $B$.
  For every $i \in [q]$,
  let $p_1^{i, \alpha},\dots,p_{n+1}^{i,\alpha}$ denote the copies of
  $p_1,\dots,p_{n+1}$ in $B^{i, \alpha}$ and
  $K_1^{i,\alpha},\dots,K_{n}^{i,\alpha}$ denote the copies of
  $K_1,\ldots,K_{n}$ in $B^{i,\alpha}$.
  Moreover, for every $1 \le a_1 \le a_2 \le n + 1$,
  we define, for notational convenience,
  \[
    p^{i, \alpha}[a_1, a_2]
      \deff \{p^{i, \alpha}_{a} \mid a_1 \le a \le a_2\}
  \]
  and
  \[
    K^{i, \alpha}[a_1, a_2]
      \deff \bigcup_{a_1 \le a \le a_2} K^{i, \alpha}_a
      .
  \]
  We define $p^{i, \beta}_a$, $K^{i, \beta}_a$,
  $p^{i, \beta}[a_1, a_2]$, and $K^{i, \beta}[a_1, a_2]$
  in a similar way.

  \item
  For $i \in [q]$ and $a \in [n]$,
  the reduction uses $p^{i, \alpha}_a$, $p^{i, \beta}_{n + 1 - a}$,
  $K^{i, \alpha}_a$, and $K^{i, \beta}_{n + 1 - a}$
  to encode vertex $v^i_a$.

  \item
  For every edge $e = v^i_{a_i}v^j_{a_j} \in E(G)$, the reduction introduces an \emph{edge-vertex} $v_e$ and
  adds edges so that $v_e$ is complete to the following sets.
  \begin{itemize}
    \item
    $ p^{i,\alpha}[a_i+1, n+1]$ and
    $K^{i, \alpha}[a_i, n+1]$ in $V(B^{i, \alpha})$.
    \item
    $ p^{j,\alpha}[a_j+1, n+1]$ and
    $K^{i, \alpha}[a_j, n+1]$ in $V(B^{j, \alpha})$.
    \item
    $ p^{i,\beta}[n + 1 - a_i + 1, n+1]$ and
    $K^{i, \beta}[n + 1 - a_i + 1, n+1]$ in $V(B^{i, \beta})$.
    \item
    $ p^{j,\beta}[n + 1 - a_j + 1, n+1]$ and
    $K^{i, \beta}[n + 1 - a_j + 1, n+1]$ in $V(B^{j, \beta})$.
  \end{itemize}
  Note that $v_e$ is adjacent to vertices in
  $K^{i, \alpha}[a_i] \cup K^{j, \alpha}[a_j]$
  but not to any vertex in
  $K^{i, \beta}[n + 1 - a_i ] \cup K^{j, \beta}[n + 1 - a_j ]$.

  \item
  The reduction introduces a central clique $K$ of size $2M^2$
  and makes it complete to
  $\{p^{i, \alpha}_{n + 1}, p^{i, \beta}_{n + 1} \mid i \in [q]\}$ and $V_E$
  where $V_E = \{v_e \mid e \in E(G)\}$ is the set of edge-vertices.

  This completes the construction of $H$.

  \item
 The reduction further defines
  \begin{align*}
    \Pairs &\deff
    \{(p_a^{i,\alpha},p_{n + 2 -a}^{i,\beta}) \mid a \in [n] \text{ and } i \in [q]\},
    \text{ and } \\
    k &\deff q(n + 1)M + |E(G)| - q(q-1)/2.
\end{align*}
\end{itemize}
The reduction returns $(H, \Pairs, k)$ as the instance of \MCundel.
This completes the reduction. It is easy to see that $H$ is chordal and has leafage at most $2q$.
See \Cref{fig:redH} for a tree representation of~$H$.

\begin{figure}[t]
  \centering
  \begin{tikzpicture}
	 \foreach \j in {0,1,2,3}
	 {
	 \ifthenelse{\j=0 \OR \j=1}{\pgfmathtruncatemacro{\p}{2}}{\pgfmathtruncatemacro{\p}{1}}
    \ifthenelse{\j=0 \OR \j=2}{\def\q{\beta}}{\def\q{\alpha}}
    \pgfmathsetmacro{\yb}{\j*1.5}
    \pgfmathsetmacro{\yu}{\yb+.15}

	 \foreach \i in {2,3,4,5}
    {\pgfmathsetmacro{\xl}{2.4*\i-.4}
    \pgfmathsetmacro{\xr}{\xl+1.2}
    \pgfmathtruncatemacro{\l}{\i - 1}
    \node[circr] (a\j\i) at (\xl,\yu) {};
    \node[circr] (b\j\i) at (\xr,\yu) {};
    \draw[very thick,purple] (a\j\i) -- (b\j\i) node[midway,above,purple] { $K_\l^{\p,\q}$};
    }

     \foreach \i in {1,2,3,4,5}
    {\pgfmathsetmacro{\xl}{2.4*\i+.8}
    \pgfmathsetmacro{\xr}{\xl+1.2}
    \pgfmathtruncatemacro{\l}{\i}
    \node[circb] (c\j\i) at (\xl,\yb) {};
    \node[circb] (d\j\i) at (\xr,\yb) {};
    \draw[very thick,teal] (c\j\i) -- (d\j\i) node[midway,above,teal] { $p_\l^{\p,\q}$};
    }
	}

	\node[scirc] (k1) at (16,2.5) {};
    \node[scirc] (k2) at (14,0.25) {};
    \node[scirc] (k3) at (14,1.75) {};
    \node[scirc] (k4) at (14,3.25) {};
    \node[scirc] (k5) at (14,4.75) {};

   	\draw[very thick] (k1) -- (k2)
    (k1) -- (k3)
    (k1) -- (k4)
    (k1) -- (k5) node[pos=.05,above right] { $K$};

  	\node[circlg] (e21) at (16,2.2) {};
    \node[circlg] (e22) at (14,-.3) {};
    \node[circlg] (e23) at (14,1.2) {};
    \node[circlg] (e24) at (14,2.7) {};
    \node[circlg] (e25) at (14,4.2) {};
    \node[circlg] (e26) at (12.8,4.2) {};
    \node[circlg] (e27) at (4.4,2.7) {};
    \node[circlg] (e28) at (8,1.2) {};
    \node[circlg] (e29) at (9.2,-.3) {};
    \draw[very thick,gray] (e21) -- (e22)
    (e21) -- (e23)
    (e21) -- (e24)
    (e21) -- (e25)
    (e24) -- (e27)
    (e23) -- (e28)
    (e22) -- (e29)
    (e25) -- (e26) node[pos=.95,below,gray] { $v_{e'}$};

    \node[circg] (e11) at (16,2.35) {};
    \node[circg] (e12) at (14,-.15) {};
    \node[circg] (e13) at (14,1.35) {};
    \node[circg] (e14) at (14,2.85) {};
    \node[circg] (e15) at (14,4.35) {};
    \node[circg] (e19) at (10.4,4.35) {};
    \node[circg] (e16) at (6.8,2.85) {};
    \node[circg] (e17) at (5.6,1.35) {};
    \node[circg] (e18) at (11.6,-.15) {};
    \draw[very thick,olive] (e11) -- (e12)
    (e11) -- (e13)
    (e11) -- (e14)
    (e11) -- (e15)
    (e13) -- (e17)
    (e12) -- (e18)
    (e14) -- (e16)
    (e15) -- (e19) node[pos=.95,below,olive] { $v_{e}$};

  \end{tikzpicture}
  \caption{A tree representation of the graph $H$ restricted to the gadgets representing $V_1$, $V_2$ and $E_{1,2}$ where $n= 4$ and $E_{1,2} = \{e = v_3^1v_1^2, e'=v_4^1v_2^2\}$.}
  \label{fig:redH}
\end{figure}
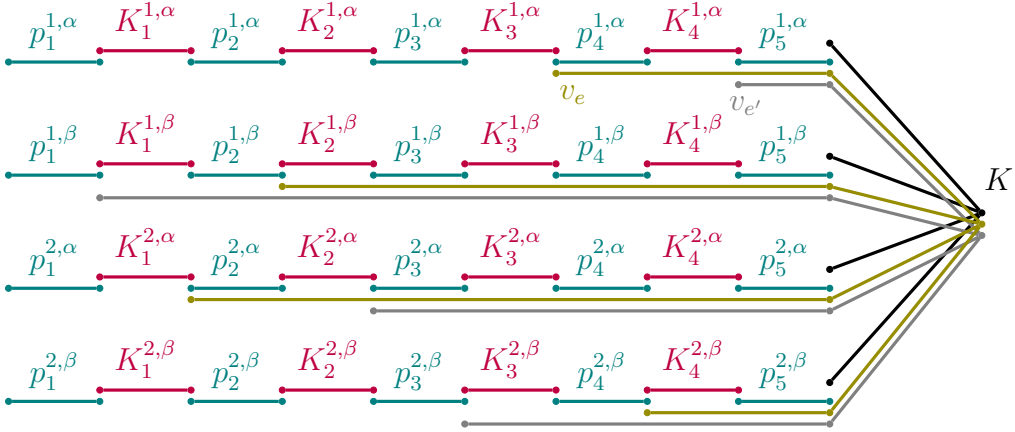

\paragraph*{Intuition.}
We first provide the intuition behind the reduction.
Recall that the reduction uses $p^{i, \alpha}_{a}$, $p^{i, \beta}_{n + 1 - a}$,
$K^{i, \alpha}_{a}$, and $K^{i, \beta}_{n + 1 - a}$ to encode vertex $v^i_a$
where $i \in [q]$ and $a \in [n]$.
Hence, for $a, b \in [n]$, if $a + b = n + 1$,
then $p^{i, \alpha}_a$ and $p^{i, \beta}_b$ correspond to the same vertex.
Note that the pairs in $\Pairs$ \emph{do not} correspond to the vertices associated with $v^i_{a}$.
Rather, $p^{i, \alpha}_{a + 1}$ is paired with $p^{i, \beta}_{n + 1 - a}$.
Conversely, for $a, b \in [n]$, if $a + b = n + 2$,
then $(p^{i, \alpha}_a, p^{i, \beta}_b) \in \Pairs$.
By the construction of $H$ and $\Pairs$, for a $\Pairs$-multicut $S$ of $H$,
if there is a path from $p^{i, \alpha}_a$ to $p^{i, \beta}_b$ in $H - S$,
then $a + b \ge n + 3$.

Now, consider the terminal pairs $(p^{i, \alpha}_1, p^{i, \beta}_{n + 1})$ in $\Pairs$ for some $i \in [q]$.
Because of the size constraints,
$S$ cannot contain all the vertices of the central clique $K$.
Since $S$ cannot contain a terminal,
it needs to include one clique from $B^{i, \alpha}$.
Let $a_i \in [n]$ be the largest index such that $K^{i, \alpha}_{a_i} \subseteq S$.
Using similar arguments, there must also exist $b_i \in [n]$ such that $K^{i, \beta}_{b_i} \subseteq S$ and $b_i$ is largest such index.
By definition of $a_i, b_i$ and construction of $H$,
there is a path from $p^{i, \alpha}_{a_i + 1}$ to $p^{i, \beta}_{b_i + 1}$ in $H - S$.
The discussion in the previous paragraph implies that
$a_i + 1 + b_i + 1 \ge n + 3$, i.e.,\ $a_i + b_i \ge n + 1$.
However, by definition of the solution size $k$ and the size of the cliques,
we have $a_i + b_i \le n + 1$.
Hence, the structure of the auxiliary graphs and the terminal pairs ensure that
the selected cliques in $S \cap V(B^{i, \alpha})$ and $S \cap V(B^{i, \beta})$
encode selecting a vertex in $V_i$ in $G$.

Suppose that $\{v_{a_1}^1,v_{a_2}^2,\allowbreak \ldots, v_{a_q}^q \}$
are the vertices in $G$ that are selected by $S$.
Recall that $V_E$ is the collection of edge-vertices in $H$.
Considering the remaining budget,
a solution $S$ can include at most $|E(G)| - q(q-1)/2$ many vertices in $V_E$.
We argue that $q(q-1)/2$ edges in $G$ corresponding to vertices in $V_E \setminus S$
should have their endpoints in $\{v_{a_1}^1,v_{a_2}^2,\allowbreak \ldots, v_{a_q}^q \}$
as otherwise some terminal pair is connected in $H - S$.
Hence, a $\Pairs$-multicut $S$ of $H$ corresponds to a multicolored clique in $G$.
We formalize this intuition in the following two lemmas.

\begin{lemma}
  \label{lemma:multicut-w-hard-forward}
  If $(G, q, (V_1, V_2, \dots, V_q))$ is a \yes-instance of {\sc Multicolored Clique},
  then $(H, \Pairs, k)$ is a \yes-instance of {\sc \MCundel}.
\end{lemma}
\begin{proof}
Assume that $(G, q, (V_1, V_2, \dots, V_q))$ is a \yes-instance of {\sc Multicolored Clique}
 and let $\{v_{a_1}^1,v_{a_2}^2,\allowbreak \ldots, v_{a_q}^q \}$
  be a clique in $G$ such that
  $v^{i}_{a_i} \in V_i$ for every $i \in [q]$.
  We construct a $\Pairs$-multicut $S$ of $H$ as follows.
  First, we add
  $V_E \setminus \{v_e \mid e \in \{v_{a_i}^i v_{a_j}^j \mid i,j \in [q]\}\}$ to $S$.
  For every $i \in [q]$,
  we further add $K_{a_i}^{i,\alpha}$ and $K_{n + 1 - a_i}^{i,\beta}$ to $S$.
  It is easy to verify that $|S| = q(n + 1)M + |E(G)| - q(q-1)/2 = k$.

  Let us show that $S$ is indeed a $\Pairs$-multicut.
  Fix indices $i \in [q]$ and $a \in [n]$,
  and consider the terminal pair $(p^{i, \alpha}_{a}, p^{i, \beta}_{n + 2 - a})$ in $\Pairs$.
  Suppose first that $a \le a_i$.
  By construction of $H$,
  any path from $p^{i,\alpha}_a$ to $p^{i,\beta}_{n+2-a}$ in $H$
  contains a vertex of $K^{i,\alpha}_{a_i}$ or of
  $N(p^{i,\alpha}_{a_i}) \cap V_E$.
  Recall that if edge $e \in E(G)$ is incident on $v^i_{a_i}$,
  then the edge-vertex $v_e$ in $H$ is adjacent only to vertices in
  $p^{i, \alpha}[a_i + 1, \cdots, n + 1]$ and
  $K^{i, \alpha}[a_i, \cdots, n + 1]$ in $V(B^{i, \alpha})$;
  in particular, it is \emph{not} adjacent to $p^{i, \alpha}_{a_i}$.
  As $S$ only excludes edge-vertices in $V_E$ that encode edges incident on $v^i_{a_i}$,
  it contains every vertex in $N(p^{i,\alpha}_{a_i}) \cap V_E$.
  Since $S$ also contains every vertex in $K^{i,\alpha}_{a_i}$,
  we conclude that there is no path from $p^{i,\alpha}_a$ to $p^{i,\beta}_{n+2-a}$ in $H - S$.

  Now, consider the case where $a_i < a$,
  i.e.,\ $n + 2 - a < n + 2 - a_i$.
  In this case,
  it is convenient to consider a path from $p^{i, \beta}_{n + 2 - a}$ to $p^{i, \alpha}_a$.
  Once again, by construction of $H$,
  any path from $p^{i, \beta}_{n + 2 - a}$ to $p^{i, \alpha}_a$ in $H$
  contains a vertex of
  $K^{i, \beta}_{n + 2 - (a_i + 1)} = K^{i, \beta}_{n + 1 - a_i}$ or of
  $N(p^{i, \beta}_{n + 2 - (a_i + 1)}) \cap V_E = N(p^{i, \beta}_{n + 1 - a_i}) \cap V_E$.
  Recall that if edge $e \in E(G)$ is incident on $v^{i}_{a_i}$,
  then the corresponding edge-vertex $v_e$ in $H$
  is adjacent only to vertices in $p^{i, \beta}[n - a_i + 2, \cdots n + 1]$
  and $K^{i, \beta}[n - a_i + 2, \cdots, n + 1]$ in $V(B^{i, \beta})$;
  in particular, it is \emph{not} adjacent to $p^{i, \beta}_{n + 1 - a_i}$.
  As $S$ only excludes edge-vertices in $V_E$ that encode edges incident on $v^i_{a_i}$,
  it contains every vertex in $N(p^{i, \beta}_{n + 1 - a_i}) \cap V_E$.
  Since $S$ also contains every vertex in $K^{i, \beta}_{n + 1 - a_i}$,
  we conclude that there is no path from $p^{i, \beta}_{n + 2 - a}$ to $p^{i, \alpha}_a$.
  This implies that no terminal pair in $\Pairs$ is connected in $H - S$
  which concludes the proof.
\end{proof}
\begin{lemma}
  \label{lemma:multicut-w-hard-backward}
  If $(H, \Pairs, k)$ is a \yes-instance of {\sc \MCundel},
  then $(G, q, (V_1, V_2, \dots, V_q))$ is a \yes-instance of {\sc Multicolored Clique}.
\end{lemma}

\begin{proof}
Assume that $(H, \Pairs, k)$ is a \yes-instance of \MCundel
 and let $S$ be a $\Pairs$-multicut of $H$ of size at most $k$.
  Recall that, by definition of the problem, $S \cap V(\Pairs) = \emptyset$.
  Also, recall that the reduction adds the clique $K$ of size $2M^2$
  and makes it complete to
  $\{p^{i, \alpha}_{n + 1}, p^{i, \beta}_{n + 1} \mid i \in [q]\}$ and $V_E$.
  Note that $K\setminus S \neq \emptyset$ as $k < 2M^2$.

  Consider an index $i \in [q]$.
  It is easy to see that there exists $a \in [n]$ such that
  $K^{i,\alpha}_a \subseteq S$ as otherwise,
  there is a path from $p^{i,\alpha}_1$ to $p^{i,\beta}_{n+1}$ in $H - S$.
  Let $a_i \in [n]$ be the largest index such that $K^{i,\alpha}_{a_i} \subseteq S$.
  Similarly, there must exist $b \in [n]$ such that $K^{i,\beta}_{b} \subseteq S$:
  let $b_i \in [n]$ be the largest index such that $K^{i,\beta}_{b_i} \subseteq S$.
  Note that by definition of $a_i, b_i$ and the fact that $K \setminus S \neq \emptyset$,
  there is path from $p^{i, \alpha}_{a_i + 1}$ to $p^{i, \beta}_{b_i + 1}$ in $H-S$.
  Now suppose for a contradiction that $a_i + 1 + b_i + 1 \le n + 2$.
  Then there exists $a'_i \in [n]$ such that $a'_i \ge a_i$ and $a'_i + 1 + b_i + 1 = n + 2$
  and so, by definition of $\Pairs$,
  $(p^{i, \alpha}_{a'_i + 1}, p^{i, \beta}_{b_i + 1}) \in \Pairs$.
  Moreover, by construction of $H$,
  the existence of a path from $p^{i, \alpha}_{a_i + 1}$ to $p^{i, \beta}_{b_i + 1}$ in $H - S$
  implies that there is path from $p^{i, \alpha}_{a'_i + 1}$ to $p^{i, \beta}_{b_i + 1}$ in $H - S$;
  this however, contradicts the fact that $S$ is a $\Pairs$-multicut of $H$.
  Therefore $a_i + 1 + b_i + 1 \ge n + 3$,
  i.e.,\ $a_i + b_i \ge n + 1$.
  Since this holds for any $i \in [q]$, we have that
  \[
    \left|S \cap \left(\bigcup_{i \in [q]}V(B^{i, \alpha}) \cup V(B^{i, \beta})\right)\right|
    \ge \sum_{i \in [q]} (a_i M + b_iM) \ge q (n + 1) M
    .
  \]
  Since $|E(G)| - q (q - 1)/2 < M$
  and $S$ has size at most $k = q(n + 1) M + |E(G)| - q (q - 1)/2$,
  it follows that, in fact, $a_i + b_i = n + 1$ for all $i \in [q]$.
  Hence, $S \cap (\bigcup_{i \in [q]}V(B^{i, \alpha}) \cup V(B^{i, \beta}))$
  corresponds to a collection of vertices
  $\{v_{a_1}^1,v_{a_2}^2,\allowbreak \ldots, v_{a_q}^q \}$ in $G$
  such that $v^{i}_{a_i} \in V_i$ for every $i \in [q]$.

  In the remaining proof, we argue that there are at least $q(q-1)/2$ edges with endpoints in $\{v_{a_1}^1,v_{a_2}^2,\allowbreak \ldots, v_{a_q}^q \}$.
  Since $|E(G)| - q (q - 1)/2 < M$,
  and every clique in $B^{i, \alpha}$ is of size at least $M$,
  for any $a \in [n]$ such that $a < a_i$,
  we have $K^{i, \alpha}_a \setminus S\neq \emptyset$.
  In other words, there is at least one vertex in $H-S$ from each clique
  $K^{i, \alpha}_a$ where $a < a_i$.
  Since $a_i$ is the largest index such that $K^{i, \alpha}_{a_i} \subseteq S$, this also holds for every $a > a_i$.
  As $S$ intersects every path from $p^{i, \alpha}_1$
  to $p^{i, \beta}_{n + 1}$, it contains every vertex in
  $N(p^{i, \alpha}_{a_i}) \cap V_E$.
  Using similar arguments, we conclude that $S$ also contains every vertex in
  $N(p^{i, \beta}_{b_i}) \cap V_E = N(p^{i, \beta}_{n + 1 - a_i}) \cap V_E$.
  Now recall that if edge $e \in E(G)$ is incident on $v^{i}_{a_i}$, 
  then the corresponding edge-vertex $v_e$ in $H$ is adjacent to
  \begin{itemize}
    \item
    terminals in $V(B^{i, \alpha})$
    which are in $p^{i, \alpha}[a_i + 1, n + 1]$, and
    \item
    terminals in $V(B^{i, \beta})$
    which are in $p^{i, \beta}[n - a_i + 2, n + 1]$.
  \end{itemize}
  In particular, $v_e$ is \emph{not} adjacent to
  $p^{i, \alpha}_{a_i}$ and $p^{i, \beta}_{n + 1 - a_i}$.
  This implies that only edges-vertices that
  correspond to edges incident on $v^{i}_{a_i}$ can be excluded from $S$.
  As this holds for any $i \in [q]$,
  every vertex in $V_E \setminus S$ has its endpoints in
  $\{v_{a_1}^1,v_{a_2}^2,\allowbreak \ldots, v_{a_q}^q \}$.
  As $|S \cap (\bigcup_{i \in [q]}V(B^{i, \alpha}) \cup V(B^{i, \beta}))|
  =  q (n + 1) M$
  and $k = q (n + 1) M + |E(G)| - q(q-1)/2$,
  we have $|S \cap V_E| \le |E(G)| - q (q - 1)/2$
  which implies that $|V_E \setminus S| \ge q(q - 1)/2$.
  Since $G$ is a simple graph, it follows that
  $\{v_{a_1}^1,v_{a_2}^2,\allowbreak \ldots, v_{a_q}^q \}$
  is a multicolored clique in $G$.
  This concludes the proof of the lemma.
\end{proof}

  Finally, it is known that, assuming the \ETH, there is no algorithm
  that can solve \textsc{Multicolored Clique}
  on instance $(G, q, (V_1, V_2, \dots V_{q}))$
  in time $f(q) \cdot |V(G)|^{o(q)}$ for any computable function $f$
  (see, e.g., \cite[Corollary~14.23]{DBLP:books/sp/CyganFKLMPPS15}).
Thus, together with the fact that the reduction takes polynomial time in the size of the input,
\cref{lemma:multicut-w-hard-forward,lemma:multicut-w-hard-backward},
and arguments that are standard for parameter preserving reductions, we conclude that the following holds.

\begin{lemma}
  \label{lemma:multicut-w-hard-leafage}
  {\sc MultiCut with Undeletable Terminals} on chordal graphs
is {\em \W[1]-hard} when parameterized by leafage $\ell$
and assuming the {\em \ETH}, does not admit an algorithm
running in time $f(\ell)\cdot n^{o(\ell)}$ for any computable function $f$.
\end{lemma}

The remainder of this section is devoted to the proof of the following lemma,
which together with \Cref{lemma:multicut-w-hard-leafage} proves \Cref{thm:mc-w-hardness}.

\begin{lemma}
{\sc MultiCut with Undeletable Terminals} on chordal graph of leafage at most $\ell$ admits an \XP-algorithm running in time $n^{\calO(\ell)}$.
\end{lemma}

\begin{proof}
  Let $(G,\Pairs)$ be an instance of \MCundel where $G$ is a chordal graph of leafage at most $\ell$.
  Let $(T, \mld)$ be a tree representation of $G$ of leafage at most $\ell$.
  We say that a path in $T$ is a \emph{maximal degree-$2$ path}
  if it contains no branching nodes,
  except for possibly the first and last node of the path, 
  and it cannot be extended without violating this property (that is, it is maximal).
  A $\Pairs$-multicut $S$ of $G$ is said to \emph{destroy} an edge $e \in E(T)$ if $\ver(e) \subseteq S$.
  
  Let us root $T$ at an arbitrary node $\rr \in V(T)$.
  Since the number of leaves of $T$ is at most $\ell$,
  $T$ has at most $2\ell -2$ maximal degree-$2$ paths,
  one starting at each each leaf or branching node (except the root)
  and ending at the first ancestor in $T$ which is a branching node.

  Now for each maximal degree-$2$ path $Q$ from $\alpha$ to $\beta$ in $T$,
  guess the first (i.e.,\ closest to $\alpha$)
  and last (i.e.,\ closest to $\beta$) edge of $Q$,
  say $e_1^Q$ and $e_2^Q$, respectively, such that
  $S$ destroys $e_1^Q$ and $e_2^Q$.
  Note that, it might be the case that an optimal solution does not destroy an edge of $Q$
  or only destroys one edge of $Q$ (i.e.,\ $e_1^Q=e_2^Q$).
  Since the length of any maximal degree-$2$ path is $\OO(n)$,
  this creates at most $(n+1)^{2\ell}$ branches.

  In each such branch,
  let $D \subseteq E(T)$ be the set of guessed edges of $T$.
  Pick $V_D = \{\ver(e) \mid e \in D\}$ in the solution and
  delete $V_D$ from $G$:
  let $(G',\Pairs')$ be the resulting instances
  and further let $T'$ be obtained from $T$ by deleting the edges in $D$
  and set $\mld' = \mld|_{V(G')}$.
  Observe that the tree representation of each connected component of $G'$
  is given by some tree of the forest $T'$ together with $\mld'$
  restricted to the vertices of the corresponding connected component.
  Note that it is enough to solve the problem independently on each connected component of $G'$.

  Thus, without loss of generality, assume that $G'$ is connected
  and let $(T',\mld')$ be a tree representation of $G'$ as defined above.
  Suppose that $G'$ has at least one terminal pair in $\Pairs'$,
  say $(s,t) \in \Pairs' \subseteq \Pairs$.
  If $T'$ is a path, i.e.,\ $G'$ is an interval graph,
  then the problem can be solved in polynomial time~\cite[Theorem~5]{DBLP:journals/eor/GuoHKNU08}.
  Otherwise, we ignore this branch.

  The algorithm outputs a solution if there is at least one branch
  where a solution was computed.
  Otherwise, there is no solution.

  It is not difficult to see that the above algorithm indeed solves the problem,
  as it considers all the possible ways a solution could intersect every maximal degree-$2$ path.
\end{proof}

\newcommand{\Tzero}{T_0}
\newcommand{\AAA}{\texttt{A}}
\newcommand{\trunc}{\textsf{trunc}}

\section{\mwcfull\ on Chordal Graphs}
\label{sec:multiway-cut}

In this section, we consider the \mwcfull\ problem formally defined below.
Given a graph $G$ and a set $\Pterm \subseteq V(G)$,
a set $S \subseteq V(G) \setminus \Pterm$ is a called a \emph{$\Pterm$-\mwcset} in $G$ 
if $G-S$ has no $(p,p')$-path for any two distinct $p,p' \in \Pterm$.\\

\defproblemques{\mwcfull\ (\mwc)}{
An undirected graph $G$ and a set $\Pterm \subseteq V(G)$ of terminals.
}{
Find the size of a minimum $\Pterm$-\mwcset\ in $G$.
}

\bigskip
\noindent
The aim of this section is to prove Theorem~\ref{thm:mwc-poly}
which states that \mwcfull\ can be solved in $\polyn$-time on chordal graphs.
Before turning to the proof,
we first start with a few definitions.
Let $(T,\calM)$ a tree representation of a chordal graph $G$
where $T$ is rooted at an arbitrary node $\rr \in V(T)$.
Given a subtree $T'$ of $T$ and a set $Q \subseteq V(G)$, 
we let $Q_{|T'} \subseteq Q$ be the set of vertices $x \in Q$ such that $\calM(x) \subseteq V(T')$.
Now let $Q \subseteq V(G)$ be an independent set of $G$ 
such that for every leaf $\eta$ of $T$, $\ver(\eta) \cap Q \neq \emptyset$. 
Then the \emph{\truncated\ tree w.r.t. $Q$} is the tree $T^{\trunc}_Q$ obtained from $T$ as follows.
Let $\{\eta_1,\ldots,\eta_q\}$ be the set of leaves of $T$.
For each $i \in [q]$, 
let $Q_i \subseteq Q \setminus \ver(\rr)$ be the set of vertices $p \in Q \setminus \ver(\rr)$ such that
$\topnode_{\mld}(p)$ is on the
$(\eta_i,\rr)$-path in $T$,
and let $p_i \in Q_i$ be the vertex of $Q_i$ such that
$\topnode_{\mld}(p_i)$ is closest to $\rr$.
Then $T^{\trunc}_Q$ is obtained from $T$ by deleting the subtrees rooted at the children of the nodes 
in $\{\topnode_{\mld}(p_i)~|~i \in [q]\}$.
Note that, by construction, the set of leaves of $T^\trunc_Q$ is $\{\topnode_{\mld}(p_i)~|~i \in [q]\}$
and that, apart from the vertices in $\{p_i ~|~ i \in [q]\}$, there is at most one other vertex in $Q$ 
whose model intersects $V(T^\trunc_Q)$, namely the potential vertex in $Q \cap \ver(\rr)$
(note that if such a vertex exists, its model is in fact fully contained in $T^\trunc_Q$).
Finally, given a set $P \subseteq V(G)$,
a $P$-\mwcset\ $X$ in $G$ is said to \emph{destroy} an edge $e \in E(T)$ if $\ver(e) \subseteq X$.

\newcommand{\Ttilde}{\widetilde{T}}
\newcommand{\Gtilde}{\widetilde{G}}
\newcommand{\Ptilde}{\widetilde{P}}
\newcommand{\rrtilde}{\tilde{\textsf{r}}}

We now turn to the proof of \Cref{thm:mwc-poly}.
Throughout the remaining of this section, we let $(G,\Pterm)$ be an instance of \mwc, 
where $G$ is a $n$-vertex chordal graph,
and further let $(T,\mld)$ be a tree representation of $G$. 
First, we may assume that $\Pterm$ is an independent set:
indeed, if there exist $p,p' \in \Pterm$ such that $pp' \in E(G)$,
then $(G,\Pterm)$ is a \no-instance.
Furthermore, if a vertex $v \in V(G)$ does not belong to any $(p,p')$-path in $G$,
where $p,p' \in \Pterm$,
then it can be safely deleted
as no minimal $\Pterm$-\mwcset\ in $G$ may contain $v$.
Hence, we assume that every vertex in $G$ participates in some $(p,p')$-path where $p,p' \in \Pterm$;
in particular, we may assume that
for every leaf $\eta$ of $T$, $\ver(\eta) \cap \Pterm \neq \emptyset$.
Note that, consequently, for every internal node $\alpha \in V(T)$,
the truncation of $T_\alpha$ w.r.t. $P_{|T_\alpha}$ exists.

Now let $\Tzero$ be the tree obtained by adding a new node $\rr_0$ 
and connecting it to an arbitrary node $\rr \in V(T)$.
Observe that $(\Tzero,\mld)$ is also a tree representation of $G$.
In the following, we root $T_0$ at $\rr_0$.
To prove \Cref{thm:mwc-poly}, we design a dynamic program that computes, 
in a bottom-up traversal of $T_0$, the entries of a table $\AAA$ whose content is defined as follows.
The table $\AAA$ is indexed over the edges of $E(\Tzero)$.
For each node $\alpha \in V(T)$,
$\AAA[\alpha\parent_{\Tzero}(\alpha)]$ stores the size of a minimum $\Pterm_{|T_\alpha}$-\mwcset\ in $G_{|T_\alpha}$.
The size of a minimum $\Pterm$-\mwcset\ in $G$ may then be found in $\AAA[\rr\rr_0]$.
We describe below how to compute the entries of $\AAA$.

\paragraph*{Update Procedure.}
For every leaf $\eta$ of $T$, we set $\AAA[\eta\parent_{\Tzero}(\eta)] = 0$.
Consider now an internal node $\alpha$ of $T$. 
We show how to compute $\AAA[\alpha\parent_{\Tzero}(\alpha)]$
assuming that for every edge $e \in E(T_\alpha)$, the entry $\AAA[e]$ is correctly filled.

\newcommand{\sss}{\texttt{s}}
\newcommand{\ttt}{\texttt{t}}

Let $\Ttilde$ be the truncation of $T_\alpha$ w.r.t. $\Pterm_{|T_\alpha}$
and let $\Gtilde = G_{|\Ttilde}$.
Denote by $\eta_1,\ldots,\eta_q$ the leaves of $\Ttilde$.
Recall that, by construction, for every $i \in [q]$,
there exists $p_i \in \Pterm_{|T_\alpha}$ such that $\eta_i = \topnode_{\mld}(p_i)$:
we let $\Ptilde = \{p_i~|~i \in [q]\}$.
Furthermore, it may be that $\Pterm_{|T_\alpha} \cap \ver(\rr)$ is nonempty:
we let $\Ptilde_\rr = \Pterm_{|T_\alpha} \cap \ver(\rr)$.
Note that $|\Ptilde_\rr| \leq 1$:
if $\Ptilde_\rr \neq \emptyset$ then 
we refer to the terminal in $\Ptilde_\rr$ as the \emph{root terminal}.
Observe that $V(\Gtilde) \cap \Pterm_{|T_\alpha} = V(\Gtilde) \cap \Pterm = \Ptilde \cup \Ptilde_\rr$ by construction.
To compute $\AAA[\alpha\parent_{\Tzero}(\alpha)]$, we distinguish two cases:
\begin{itemize}
\item[(1)]
if $\Ptilde_\rr \neq \emptyset$ then we construct a unique instance $(H_0,\sss,\ttt,\wt_0)$ of \stcut;
\item[(2)]
otherwise, for every $i \in [0,q]$, we construct an instance $(H_i,\sss,\ttt,\wt_i)$ of \stcut.
\end{itemize}
We describe below how such instances are constructed.
First, recall that an instance of the \stcut\ problem consists of a digraph $D$,
vertices $s,t \in V(D)$, a weight function $\wt : E(D) \to \mathbb{N} \cup \{\infty\}$,
and the goal is to find a set $X \subseteq E(D)$ such that
$D-X$ has no $(s,t)$-path and $\wt(X)$ is minimum with this property,
where $\wt(X) = \sum_{u \in X} \wt(u)$. 

\newcommand{\sor}{\texttt{source}}
\newcommand{\snk}{\texttt{sink}}
\newcommand{\conn}{\texttt{conn}}
\newcommand{\rterm}{\texttt{rterm}}

\paragraph{Construction of the \stcut\ Instances.} 
For every $i \in [q]$, let us denote by $\Ptilde_i = \Ptilde \setminus \{p_i\}$
and let $\Ptilde_0 = \Ptilde$.
Consider $i \in [0,q]$.
Before turning to the formal construction of the instance $(H_i,\sss,\ttt,\wt_i)$, 
let us first give an intuitive idea of the construction. 
The digraph $H_i$ is obtained from $\Ttilde$
by orienting all edges of $\Ttilde$ towards its root $\rrtilde = \alpha$
and further adding vertices and weighted arcs to encode the graph $G_{|T_\alpha}$.
The arcs in $H_i$ corresponding to the edges of $\Ttilde$ are called the {\em tree arcs}
and the nodes in $H_i$ corresponding to the nodes of $\Ttilde$ are called the {\em tree nodes}.
The idea is that we separate, for each terminal $p \in \Ptilde_i$,
the node $\topnode_{\mld}(p)$ from the root $\rrtilde$.
To achieve this, we add a \emph{source} node $\sss$ and \emph{source} arcs from $\sss$ to $\topnode_{\mld}(p)$
(of infinite weight) and look for an $(\sss,\rrtilde)$-cut in $H_i$.
Since the edges of $T$ can presumably not be independently destroyed in a $\Pterm$-\mwcset,
we need some additional vertices to encode these dependencies.
For each vertex $v \in V(\Gtilde) \setminus \Ptilde_i$, we introduce a node $\gamma(v)$ in $H_i$
which is reachable via \emph{\connection} arcs (with infinite weight) 
from all the tree nodes that are contained in the model of $v$.
This node $\gamma(v)$ is further connected via a \emph{\sink} arc (of weight one) to $\topnode_{\mld}(v)$
which ensures that if we want to cut a tree arc,
we also have to cut all the sink arcs associated to vertices containing the corresponding edge in their model.
The index $i$ is then used to specify which root-to-leaf path of $\Ttilde$ is uncut: 
if $i = 0$ then every such path is cut, otherwise the $(\eta_i,\rrtilde)$-path is uncut.
To encode the rest of the solution,
we associate with each tree arc $(\beta,\delta)$ a weight $\wt_i((\beta,\delta))$ 
corresponding to the size of a minimum $\Pterm_{|\beta}$-\mwcset\ in $G_{|\beta}$.

We proceed with the formal construction of $H_i$.
The vertex set of $H_i$ is $V(H_i)=V(\Ttilde) \uplus \{\sss\} \uplus \{\Gamma\}$
where $\Gamma = \{\gamma(v) \mid v \in V(\Gtilde) \setminus \Ptilde\}$,
that is, $\Gamma$ contains a node of every non-terminal vertex in $\Gtilde$.
For every $z \in \Gamma$, we denote by $\gamma^{-1}(z)$ the corresponding vertex in $V(\Gtilde) \setminus \Ptilde$.
The arc set of $H_i$ is partitioned into four sets:
\begin{itemize}
\item the set $E_{\Ttilde}$ of \emph{tree arcs} containing all the edges of $\Ttilde$ oriented towards the root $\rrtilde$,
\item the set $E^i_{\sor}=\{(\sss,\topnode_{\mld}(p)) \mid p \in \Ptilde_i\}$ of \emph{source} arcs,  
\item the set  $E_{\conn} = \{(\alpha,\gamma(v)) \mid \gamma(v) \in \Gamma, \alpha \in \mld(v) \cap V(\Ttilde)\}$ of \emph{\connection} arcs and
\item the set $E_{\snk}=\{(\gamma(v),\topnode_{\mld}(v)) \mid v \in V(\Gtilde) \setminus \Ptilde\}$ of \emph{\sink} arcs.
\end{itemize}
Furthermore, if $\Ptilde_\rr \neq \emptyset$, 
then we let $E_{\rterm} \subseteq E_{\Ttilde}$ be the set of tree arcs $(\beta,\delta) \in E_{\Ttilde}$ 
such that the edge $\beta\delta$ is contained in the model of the root terminal;
otherwise, we let $E_{\rterm} = \emptyset$. 
The weight function $\wt_i : E(H_i) \to \mathbb{N} \cup \{\infty\}$ is defined as follows.
For every $j \in [q]$, let $\rho_j$ be the path in $\Ttilde$ from $\eta_j$ to $\rrtilde$ 
and let $\overrightarrow{\rho_j}$ be the corresponding directed path in $H_i$
(that is, $\overrightarrow{\rho_j}$ is the path in $H_i$ from $\eta_j$ to $\rrtilde$ consisting only of tree arcs).
Then for every arc $e$ of $H_i$, 
\[
  \wt_i(e) = 
  \begin{cases}
  \AAA[e]  &\text{if } i=0 \text{ and } e \in E_{\Ttilde} \setminus E_{\rterm} \\
  \AAA[e] &\text{if } i\neq 0\text{, } e \in E_{\Ttilde} 
  \text{ and } e \text{ does not belong to the path } \overrightarrow{\rho_i} \\
  1 &\text{if } e \in E_{\snk} \\
  \infty &\text{otherwise.}
    \end{cases}
\]
Note, in particular, that every arc in $E_{\rterm}$ (if any) has infinite weight.
Similarly, if $i \neq 0$, then every arc of the path $\overrightarrow{\rho_i}$ has infinite weight.
This completes the construction of the instance $(H_i,\sss,\ttt=\rrtilde,\wt_i)$ (see \Cref{fig:Hi}).
It is easy to see that such an instance can be constructed in $\calO(n^2)$-time.\\

\tikzset{
  circ/.style = {circle,draw,fill,inner sep=1.3pt},
  circb/.style = {circle,draw=teal,fill=teal,text=teal,inner sep=.8pt},
  circr/.style = {circle,draw=red,fill=red,inner sep=1.3pt},
  circg/.style = {circle,draw=olive,fill=olive,inner sep=.8pt},
  circlg/.style = {circle,draw=gray,fill=gray,inner sep=.8pt},
  scirc/.style = {circle,draw,fill,inner sep=.8pt},
}

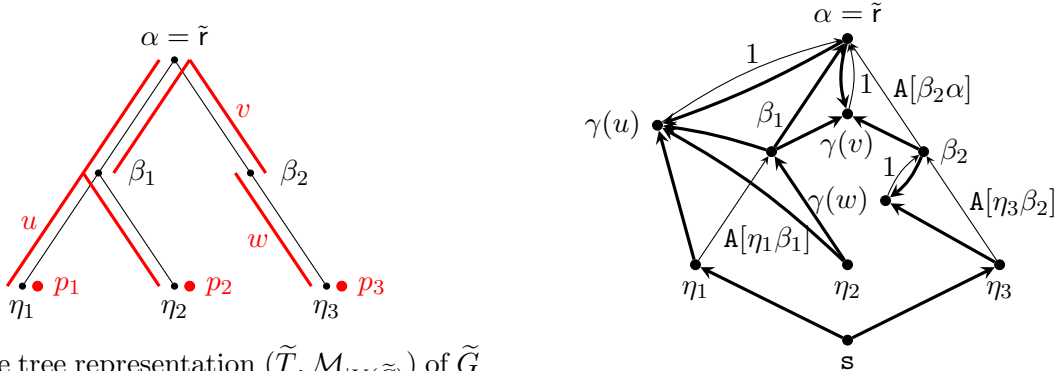
\begin{figure}
\begin{subfigure}[b]{.45\textwidth}
\centering
\begin{tikzpicture}
\node[scirc,label=below:{\small $\eta_1$}] (eta1) at (0,0) {};
\node[scirc,label=below:{\small $\eta_2$}] (eta2) at (2,0) {};
\node[scirc,label=below:{\small $\eta_3$}] (eta3) at (4,0) {};
\node[scirc,label={[label distance=.2cm]0:\small $\beta_1$}] (beta1) at (1,1.5) {};
\node[scirc,label={[label distance=.2cm]0:\small $\beta_2$}] (beta2) at (3,1.5) {};
\node[scirc,label=above:{\small $\alpha = \rrtilde$}] (alpha) at (2,3) {};

\draw (eta1) -- (alpha)
(eta2) -- (beta1)
(eta3) -- (alpha);

\node[circr,label=right:{{\color{red} \small $p_1$}}] at (.2,0) {};
\node[circr,label=right:{{\color{red} \small $p_2$}}] at (2.2,0) {};
\node[circr,label=right:{{\color{red} \small $p_3$}}] at (4.2,0) {};

\draw[very thick,red] (-.2,0) -- (1.8,3) node[left,pos=0.28] {\color{red} \small $u$};
\draw[very thick,red] (1.8,0) -- (.8,1.5);

\draw[very thick,red] (3.8,0) -- (2.8,1.5) node[left,pos=.4] {\color{red} \small $w$};

\draw[very thick,red] (3.2,1.5) -- (2.2,3) node[right,pos=.55] {\color{red} \small $v$};
\draw[very thick, red] (2.2,3) -- (1.2,1.5);
\end{tikzpicture}
\caption{The tree representation $(\Ttilde,\calM_{|V(\Gtilde)})$ of $\Gtilde$ where $V(\Gtilde) = \{p_1,p_2,p_3,u,v,w\}$ 
and $\Ptilde_\rr = \emptyset$.}
\end{subfigure}
\hspace*{.5cm}
\begin{subfigure}[b]{.45\textwidth}
\centering
\begin{tikzpicture}
\node[circ,label=below:{\small $\eta_1$}] (eta1) at (0,0) {};
\node[circ,label=below:{\small $\eta_2$}] (eta2) at (2,0) {};
\node[circ,label=below:{\small $\eta_3$}] (eta3) at (4,0) {};
\node[circ,label={[label distance=.1cm]92:\small $\beta_1$}] (beta1) at (1,1.5) {};
\node[circ,label=right:{\small $\beta_2$}] (beta2) at (3,1.5) {};
\node[circ,label=above:{\small $\alpha = \rrtilde$}] (alpha) at (2,3) {};
\node[circ,label=left:{\small $\gamma(u)$}] (u) at (-.5,1.85) {};
\node[circ,label=below:{\small $\gamma(v)$}] (v) at (2,2) {};
\node[circ,label=left:{\small $\gamma(w)$}] (w) at (2.5,.85) {};
\node[circ,label=below:{\small $\sss$}] (s) at (2,-1) {};

\draw[very thick,->,>=stealth] (s) -- (eta1);
\draw[very thick,->,>=stealth] (s) -- (eta3);

\draw[very thick,->,>=stealth] (beta1) edge[bend right=5] (u);
\draw[very thick,->,>=stealth] (eta1) -- (u);
\draw[very thick,->,>=stealth] (eta2) edge[bend right=7] (u);
\draw[very thick,->,>=stealth] (alpha) edge[bend left=5] (u);
\draw[->,>=stealth] (u) edge[bend left=10] (alpha);
\node[draw=none] at (.75,2.8) {\small $1$};

\draw[very thick,->,>=stealth] (beta1) -- (v);
\draw[very thick,->,>=stealth] (beta2) -- (v);
\draw[very thick,->,>=stealth] (alpha) edge[bend right=15] (v);
\draw[->,>=stealth] (v) edge[bend right=15] (alpha);
\node[draw=none] at (2.25,2.4) {\small $1$};

\draw[very thick,->,>=stealth] (eta3) -- (w);
\draw[very thick,->,>=stealth] (beta2) edge[bend left=15] (w);
\draw[->,>=stealth] (w) edge[bend left=15] (beta2);
\node[draw=none] at (2.55,1.3) {\small $1$};

\draw[->,>=stealth] (eta1) -- (beta1) node[right,pos=.2] {\small $\AAA[\eta_1\beta_1]$};
\draw[very thick,->,>=stealth] (eta2) -- (beta1);
\draw[very thick,->,>=stealth] (beta1) -- (alpha);
\draw[->,>=stealth] (eta3) -- (beta2) node[right,pos=.55] {\small $\AAA[\eta_3\beta_2]$};
\draw[->,>=stealth] (beta2) -- (alpha) node[right,pos=.55] {\small $\AAA[\beta_2\alpha]$};
\end{tikzpicture}
\caption{The instance $(H_2,\sss,\rrtilde,\wt_2)$ (thick arcs have infinite weight).}
\end{subfigure}
\caption{An illustration of the construction of the \stcut\ instances.}
\label{fig:Hi}
\end{figure}

\newcommand{\cost}{\textnormal{\texttt{cost}}}

Now let $X_0$ be an $(\sss,\rrtilde)$-cut in $H_0$ such that $\wt_0(X_0)$ is minimum;
and if $\Ptilde_\rr = \emptyset$, 
then for every $i \in [q]$, further let $X_i$ be an $(\sss,\rrtilde)$-cut in $H_i$ such that $\wt_i(X_i)$ is minimum.
For each $i \in [q]$, let us denote by $\cost_i =  \AAA[\eta_i\parent_{\Tzero}(\eta_i)]$ 
and let $\cost_0=0$.
Then we set 
\[
\AAA[\alpha\parent_{\Tzero}(\alpha)] = 
\begin{cases}
|X_0| & \text{if } \Ptilde_\rr \neq \emptyset \\
\min_{i \in [0,q]}\{|X_i| + \cost_i\} & \text{otherwise}
\end{cases}
\]
In the following, for convenience, we let $I = [0,q]$ if $\Ptilde_\rr = \emptyset$, 
and $I = \{0\}$ otherwise.
We next show that the entry $\AAA[\alpha\parent_{\Tzero}(\alpha)]$ is updated correctly. 
To this end, we show that $G_{|T_\alpha}$ has a $P_{|T_\alpha}$-\mwcset\ of size at most $k$
if and only if there exists $i \in I$ such that
$H_i$ has an $(\sss,\rrtilde)$-cut of weight at most $k -\cost_i$ w.r.t. $\wt_i$.

\begin{lemma}\label{lem:backward}
  For any $i \in I$, if $H_i$ has an $(\sss,\rrtilde)$-cut $Y$ such that $\wt_i(Y) \leq k -\cost_i$,
  then $G_{|T_\alpha}$ has a $P_{|T_\alpha}$-\mwcset\ of size at most $k$.
\end{lemma}

\begin{proof}
Assume that there exists $i \in I$ such that 
$H_i$ has an $(\sss,\rrtilde)$-cut $Y$ where $\wt_i(Y) \leq k -\cost_i$.
For every $j \in [q] \setminus \{i\}$, 
let $A_j$ be the set of tree arcs on the path $\overrightarrow{\rho_j}$ belonging to $Y$ 
(recall that $\overrightarrow{\rho_j}$ is the path in $H_i$ from $\eta_j$ to $\rrtilde$ consisting only of tree arcs).
Note that since $Y$ is an $(\sss,\rrtilde)$-cut, $A_j \neq \emptyset$ for every $j \in [q] \setminus \{i\}$.

  \begin{claim}\label{claim:arcs-and-gamma}
    For every terminal $j \in [q] \setminus \{i\}$, there exists an arc $(x,y) \in A_j$
    such that for every $z \in N_{H_i}^+(x) \setminus (N_{H_i}^-(x) \cup \{y\})$, the sink arc with tail $z$ belongs to $Y$.
  \end{claim}
  
  \begin{claimproof}
    Suppose for a contradiction that this does not hold for some index $j \in [q] \setminus \{i\}$, that is,
    for every arc $(x,y) \in A_j$, there exists $z \in N_{H_i}^+(x) \setminus (N_{H_i}^-(x) \cup \{y\})$
    such that the sink arc with tail $z$ does not belong to $Y$.
    Let $(x_1,y_1),\ldots,(x_a,y_a)$ be the arcs of $A_j$
    ordered according to their order of appearance when traversing the path $\overrightarrow{\rho_j}$.
    We show that, in this case, there is a path from $\sss$ to $\rrtilde$ in $H-Y$.
    For every $b \in [a]$,
    denote by $Z_b \subseteq N_{H_i}^+(x_b) \setminus (N_{H_i}^-(x_b) \cup \{y_b\})$
    the set of vertices $z$ such that the sink arc with tail $z$ does not belong to $Y$.
    Let $b_1,\ldots,b_w \in [a]$ be the longest sequence defined as follows:
    \begin{itemize}
      \item
      $b_1 \in [a]$ is the largest index such that $Z_1 \cap Z_{b_1} \neq \emptyset$ and
      \item
      for every $l > 1$, $b_{l} \in [a]$ is the largest index such that $Z_{b_{l-1}+1} \cap Z_{b_{l}} \neq \emptyset$.
    \end{itemize}
    For every $l \in [w]$,
    consider a vertex $z_{b_{l}} \in Z_{j_{l}}$
    and let $h_{b_{l}} \in N_{H_i}^+(z_{b_{l}})$ be the head of the sink arc with tail $z_{b_{l}}$.
    Then for every $l \in [w-1]$, $h_{b_{l}}$ lies on the path $\overrightarrow{\rho_j}[y_{b_{l}},x_{b_{l}+1}]$:
    indeed, since $z_{b_{l}} \notin Z_{b_{l} + 1}$ by the choice of $b_{l}$,
    either $z_{b_{l}} \notin N_{H_i}^+(x_{b_{l}+1})$
    or $z_{b_{l}} \in N_{H_i}^+(x_{b_{l}+1}) \cap N_H^-(x_{b_{l}+1})$;
    but $z_{b_{l}} \in N_{H_i}^+(x_{b_{l}}) \setminus N_{H_i}^-(x_{b_{l}})$ by construction,
    and so, $h_{b_{l}}$ necessarily lies on $\overrightarrow{\rho_j}[y_{b_{l}},x_{b_{l}+1}]$.

    Now observe that, by maximality of the sequence, $b_w = a$:
    indeed, if $b_w < a$ then the sequence could be extended as $Z_{b_w + 1} \neq \emptyset$ by assumption.
    Since $z_{b_w} \notin N_{H_i}^-(x_{b_w})$, this implies, in particular, that 
    $h_{b_w}$ lies on the path $\overrightarrow{\rho_j}[y_{b_w},\rrtilde]$. 
    It follows that
    \[
      \sss \overrightarrow{\rho_j}[\eta_j,x_1] z_{b_1} 
      \overrightarrow{\rho_j}[h_{b_1},x_{b_1+1}]
      z_{b_2} \ldots z_{b_{l}} \overrightarrow{\rho_j}[h_{b_{l}},x_{b_{l}+1}]
      z_{b_{{l}+1}} \ldots \overrightarrow{\rho_j}[h_{b_{w-1}},x_{b_{w-1}+1}]
      z_{b_w} L[h_{b_w},\rrtilde]
    \]
    is a path from $\sss$ to $\rrtilde$ in $H-Y$, a contradiction which proves our claim.
  \end{claimproof}

  For every $j \in [q] \setminus \{i\}$,
  let $e_j = (x_j,y_j) \in A_j$ be the arc closest to $\rrtilde$ such that
  for every $z \in N_{H_i}^+(x_j) \setminus (N_{H_i}^-(x_j) \cup \{y_j\})$, the sink arc with tail $z$ belongs to $Y$
  (note that we may have $e_j = e_{j'}$ for two distinct $j,j' \in [q] \setminus \{i\}$).
  Denote by $E = \{e_j \mid j \in [q] \setminus \{i\}\} \cup \{e^*\}$ where $e^* = (\eta_i, \parent(\eta_i))$.  
  For every $e=(x,y) \in E$,
  let $\Ptilde_{e} \subseteq \Ptilde_i$ be the set of terminals in $\Ptilde_i$
  which are also terminals in the instance restricted to $T_x$.
  Note that $\{\Pterm_{e} \mid e \in E \setminus \{e^*\}\}$ is a partition of $\Ptilde_i$:
  indeed, by construction, every $p \in \Ptilde_i$ belongs to at least one such set
  and if there exist $e,e' \in E \setminus \{e^*\}$ such that $\Ptilde_{e} \cap \Ptilde_{e'} \neq \emptyset$,
  then for any $j \in [q] \setminus \{i\}$ such that $p_j \in \Pterm_{e} \cap \Pterm_{e'}$, $e,e' \in A_j$;
  in particular, both $e$ and $e'$ lie on the path $\overrightarrow{\rho_j}$,
  a contradiction to the choice of the arc in $A_j$.

  Now for every $e = (x,y) \in E$, let $S_{e}$ be a minimum $P_{|T_{x}}$-multiway-cut in $G_{|T_{x}}$
  and denote by $N_{e} = N_{H_i}^+(x) \setminus (N_{H_i}^-(x) \cup \{y\})$.
  We define
  \[
    S = S_{e^*} \cup \bigcup_{e \in E \setminus \{e^*\}} S_{e} \cup \{\gamma^{-1}(z) \mid z \in N_{e}\}
    .
  \]
  
  \begin{claim}
    $S$ is a $P_{|T_\alpha}$-\mwcset\ in $G_{|T_\alpha}$.
  \end{claim}
  
  \begin{claimproof}
    Since for every $e = (x,y) \in E$, $S_{e}$ is a $P_{|T_{x}}$-multiway-cut in $G_{|T_{x}}$,
    it is in fact enough to show that for every $e,e' \in E$, $p \in \Ptilde_{e}$ and $p' \in \Ptilde_{e'}$,
    there is no path from $p$ to $p'$ in $G_{|T_\alpha} - S$.

    Consider therefore $j,j' \in [q] \setminus \{i\}$ 
    such that $p_j \in \Ptilde_{e}$ and $p_{j'} \in \Ptilde_{e'}$ for two distinct $e,e' \in E$.
    Since, as shown above, $\{ \Ptilde_f \mid f \in E \setminus \{e^*\}\}$ is a partition of $\Ptilde_i$,
    $p_{j'} \notin \Ptilde_{e}$ and $p_j  \notin \Ptilde_{e'}$;
    in particular, $e'$ does not lie on the path $\overrightarrow{\rho_j}$
    and $e$ does not lie on the path $\overrightarrow{\rho_{j'}}$.
    It follows that any path in $G_{|T_\alpha}$ from $p_j$ to $p_{j'}$ 
    contains at least one vertex $x$ whose model contains the edge corresponding to $e$;
    but then, $\gamma(x) \in N_{e}$ and so, $x \in S$ by construction.
    Thus, there is no path from $p_j$ to $p_{j'}$ in $G_{|T_\alpha} - S$. 
      \end{claimproof}
      
  Finally, note that, by construction,
  \begin{equation*}
  \begin{split}
    |S| &= |S_{e^*}| + \sum_{e \in E \setminus \{e^*\}} |S_{e}| 
    + \abs*{\bigcup_{e \in E \setminus \{e^*\}}  \{\gamma^{-1}(z) \mid z \in N_{e}\}}\\
    &= |S_{e^*}|  + \sum_{e \in E \setminus \{e^*\}} \wt_i(e) 
    + \sum_{z \in \bigcup_{e \in E \setminus \{e^*\}} N_{e}} \wt_i((z,\topnode_{\mld}(\gamma^{-1}(z))))\\
    &\leq \cost_i + \wt_i(Y) \leq k
  \end{split}
  \end{equation*}
  which concludes the proof.
\end{proof}

\begin{lemma}\label{lem:forward}
If $G_{|T_\alpha}$ has a $P_{|T_\alpha}$-\mwcset\ $X$ of size at most $k$,
then there exists $i \in I$ such that $H_i$ has an $(\sss,\rrtilde)$-cut $Y$ where $\wt_i(Y) \leq k -\cost_i$.
\end{lemma}

\begin{proof}
Recall that for every $j \in [q]$, $\rho_j$ is the unique $(\eta_j,\rrtilde)$-path in $\Ttilde$.
To prove the lemma, we first show the following.

\begin{claim}
\label{clm:GtoHi}
If there exists $i \in [q]$ such that $G_{|T_\alpha}$ has a $P_{|T_\alpha}$-\mwcset\ $X$ of size at most $k$ where
\begin{itemize}
\item[(1)]
$X$ does not destroy any edge of $\rho_i$ and
\item[(2)]
for every $j \in [q] \setminus \{i\}$, $X$ destroys an edge of $\rho_j$,
\end{itemize}
then $H_i$ has an $(\sss,\rrtilde)$-cut $Y$ such that $\wt_i(Y) \leq k -\cost_i$.
\end{claim}
 
 \begin{claimproof}
Assume that such an index $i \in [q]$ exists 
and let $X$ be a $P_{|T_\alpha}$-\mwcset\ $X$ of size at most $k$ satisfying item (1) and (2).
Note that since $X$ does not destroy any edge of $\rho_i$, $\Ptilde_\rr = \emptyset$
for, otherwise, $p_i$ and the root terminal would be in the same connected component of $G_{|T_\alpha} - X$
thereby contradicting the fact that $X$ is a $P_{|T_\alpha}$-\mwcset. 
For every $j \in [q] \setminus\{i\}$, 
  let $e_j \in E(\Ttilde)$ be the closest edge to $\eta_j$ on $\rho_j$ such that $\ver(e_j) \subseteq X$
  (note that the edges $e_1, \ldots, e_q$ are not necessarily pairwise distinct).
Denote by $E = \{e_j~|~j \in [q] \setminus \{i\}\}$.
  We construct an $(\sss,\rrtilde)$-cut $Y$ in $H_i$ as follows:
  $Y$ contains the tree arcs of $H_i$ corresponding to the edges in $E$ and
  for each $v \in X$ such that $\mld(v)$ contains at least one edge of $E$
  (that is, $v \in \ver(e)$ for some edge $e \in E$),
  we include in $Y$ the \sink\ arc $(\gamma(v),\topnode_{\mld}(v))$ of $E(H_i)$.
  Let us show that $Y$ is indeed an $(\sss,\rrtilde)$-cut in $H_i$.
  
  For every $j \in [q] \setminus \{i\}$,
  let $V_-^j \subseteq V(\Ttilde)$ ($V_+^j \subseteq V(\Ttilde)$, respectively) 
  be the set of nodes of the subpath of $\rho_j$
  from $\eta_j$ to the tail of $e_j$ (the head of $e_j$ to $\rrtilde$, respectively).
  We contend that for every $j \in [q] \setminus \{i\}$, there is no $(V_-^j,V_+^j)$-path in $H_i-Y$.
  Note that if true, this would prove that $Y$ is indeed an $(\sss,\rrtilde)$-cut in $H_i$.
    For the sake of contradiction,
    suppose that, for some $j \in [q] \setminus \{i\}$,
    there is a path $L$ in $H_i -Y$ 
    from a vertex $x \in V_-^j$ to a vertex $y \in V_+^j$.
    Since the tree arc in $H_i$ corresponding $e_j$ belongs to $Y$,
    there must exist a vertex $z \in V(L)$ such that
    $N^-_{H_i}(z) \cap V_-^j \cap V(L) \neq \emptyset$ and 
    $N^+_{H_i}(z) \cap V_+^j \cap V(L) \neq \emptyset$; 
    in particular, the sink arc $e$ with tail $z$ must belong to $L$.
    By construction of $H_i$, it must then be that $\mld(\gamma^{-1}(z))$ contains the edge $e_j$,
    that is, $\gamma^{-1}(z) \in \ver(e_j)$; 
    but then, $\gamma^{-1}(z) \in X$ and so, $e \in Y$ by construction,
    a contradiction which proves our claim.
    
    Let us finally show that $\wt_i(Y) \leq k - \cost_i$.
    To this end, for every $e \in E$, 
    let $X_e \subseteq X$ be the restriction of $X$ to $T_{t_e}$
    where $t_e$ is the endpoint of $e$ the furthest from $\rrtilde$
    (note that for any two distinct $e,e' \in E$, $X_e \cap X_{e'} = \emptyset$).
    Then, for every $e \in E$, $X_e$ is a $\Pterm_{|T_{t_e}}$-\mwcset\ in $G_{|T_{t_e}}$
    and so, $\wt_i(e) \leq |X_e|$.
    Similarly, the restriction $X_i$ of $X$ to $T_{\eta_i}$ is a $\Pterm_{|T_{\eta_i}}$-\mwcset\ in $G_{|T_{\eta_i}}$
    and so, $|X_i| \geq \cost_i$
    (note that, by construction, $X_i \cap X_e = \emptyset$ for every $e \in E$).
    Letting $X' = \bigcup_{e \in E} \ver(e)$,
    it then follows from the definition of $Y$ that 
    \[
    \wt_i(Y) = |X'| + \sum_{e \in E} \wt_i(e) \leq |X'| + \sum_{e\in E} |X_e| \leq |X| - |X_i| \leq k - \cost_i
    \]
    as $X' \cap X_i = \emptyset$ and for every $e \in E$, $X' \cap X_e = \emptyset$. 
  \end{claimproof}
  
  Using similar arguments, we can also prove the following.
  
  \begin{claim}
  \label{clm:GtoH0}
  If $G_{|T_\alpha}$ has a $P_{|T_\alpha}$-\mwcset\ $X$ of size at most $k$ 
  such that for every $i \in [q]$, $X$ destroys an edge of $\rho_i$,
  then $H_0$ has an $(\sss,\rrtilde)$-cut $Y$ such that $\wt_i(Y) \leq k$.
  \end{claim}

To conclude the proof of \Cref{lem:forward}, let us show that 
for any $P_{|T_\alpha}$-\mwcset\ $S$ in $G_{|T_\alpha}$,
$S$ destroys an edge of every root-to-leaf path of $\Ttilde$, except for at most one when $\Ptilde_\rr = \emptyset$.
Note that if the claim is true, the lemma would then follow from Claims~\ref{clm:GtoHi} and \ref{clm:GtoH0}.

Let $S$ be a $P_{|T_\alpha}$-\mwcset\ in $G_{|T_\alpha}$.
Observe first that if $\Ptilde_\rr \neq \emptyset$ then for every $i \in  [q]$, 
$S$ must destroy an edge of $\rho_i$ for, otherwise, 
$p_i$ and the root terminal are in the same connected component of $G_{|T_\alpha} - S$,
thereby contradicting the fact that $S$ is a $P_{|T_\alpha}$-\mwcset.
Assume therefore that $\Ptilde_\rr = \emptyset$
and suppose, for the sake of contradiction, that there exist two distinct indices $i ,j \in [q]$
such that $S$ destroys no edge of $\rho_i$ and no edge of $\rho_j$.
Then for every edge $e$ of $\rho_i \cup \rho_j$, 
$\ver(e) \setminus S \neq \emptyset$:
for each such edge $e$,
let $\alpha_e \in \ver(e) \setminus S$.
 It is now not difficult to see that there is a path in $G_{|T_\alpha}-S$ from $p_i$ to $p_j$ 
 using only vertices from  $\{\alpha_e \mid e \text{ is an edge of } \rho_i \cup \rho_j\}$,
 a contradiction to the fact that $S$ be a $P_{|T_\alpha}$-\mwcset\ in $G_{|T_\alpha}$.
\end{proof}

  We now conclude by Lemmas~\ref{lem:backward} and~\ref{lem:forward},
  that $\AAA[\alpha\parent_{\Tzero}(\alpha)]$ indeed 
  stores the size of a minimum $\Pterm_{|T_\alpha}$-\mwcset\ in $G_{|T_\alpha}$.
  Since the construction of each $H_i$ takes polynomial-time,
  an $(\sss,\ttt)$-cut in $H_i$ can be computed in polynomial time (see, for instance, \cite{ford_fulkerson})
  and the number of $H_i$s is at most $n$,
  it takes plynomial-time to update $\AAA[\alpha\parent_{\Tzero}(\alpha)]$.
  Finally, since the number of edges of $T$ is linear in $n$,
  the overall running time is polynomial in $n$,
  which proves Theorem~\ref{thm:mwc-poly}.
  We remark that a more careful analysis of the running time of the algorithm 
  leads to an upper bound of $\calO(n^4)$.
\section{Restricting to $H_{\ell}$-induced-subgraph-free chordal graphs}
\label{sec:Hlchordal}

In this section, we consider problems restricted to $H_\ell$-induced-subgraph-free chordal graphs.
Recall that $H_{\ell}$ is the split graph on $2\ell$ vertices such that if $V(H_{\ell}) = C \uplus I$ is a split partition then
$(i)$ $|C| = |I| = \ell$,
$(ii)$ every vertex in $C$ is adjacent to exactly one vertex in $I$, and
$(iii)$ every vertex in $I$ is adjacent to exactly one vertex in $C$.
As mentioned in the Introduction, the class of $H_\ell$-induced-subgraph-free chordal graphs
is a natural generalization of the class of chordal graphs of leafage at most $\ell$.
In fact, denoting by $\calC_{\ell}$ the collection of all chordal graphs 
that have leafage at most $\ell$
and by $\calC^{is}_{\ell}$ the collection of all chordal graphs that do not contain $H_{\ell}$ as a induced subgraph
(that is, the collection of $H_{\ell}$-induced-subgraph-free chordal graphs),
the following holds.

\begin{observation}
\label{obs:subclasses-chordal}
$\calC_{\ell} \subsetneq \calC^{is}_{\ell + 1}$.
\end{observation}

Let us briefly explain why \Cref{obs:subclasses-chordal} holds true.
Walter generalized the concept of asteroidal triple in order to characterize other subclasses of chordal graphs \cite{walter1972representations} as follows.
A subset of nonadjacent vertices of $G$ is an \emph{asteroidal set} if the removal of the closed neighborhood of any one of its elements does not disconnect the remaining ones.
Formally, a set of vertices $A$ of a graph $G$ is \emph{asteroidal} if for each $a \in A$, the vertices in $A \setminus \{a\}$ belong to a common connected component of $G - N[a]$.
The asteroidal number of $G$, denoted by $\at(G)$, is then the size of a largest asteroidal set of $G$.
Note that in the graph $H_{\ell + 1}$, $I$ is an asteroidal set of size $\ell + 1$ and thus, $\at(H_{\ell + 1}) \ge \ell + 1$. 
By definition, if $H$ is a subgraph of $G$ and $H$ is connected, then $\at(H) \le \at(G)$.
Lin et al.~\cite[Therorem~$1$]{DBLP:journals/dmgt/LinMW98} proved that for a connected chordal graph $G$, $\at(G) \le \lf(G)$.
Hence, if $\lf(G) \le \ell$, then it cannot contain $H_{\ell + 1}$ as an induced subgraph.
This implies that $\calC_{\ell} \subseteq \calC^{is}_{\ell + 1}$.
To see that $\calC_{\ell}$ is proper subset of $\calC^{is}_{\ell + 1}$, consider a graph obtained from a star by subdividing every edge once.
Then it is easy to see that this graph does not contain $H_3$ as induced subgraph but can have unbounded leafage.\\

The remainder of this section is organized as follows.
In Subsection~\ref{subsec:dominating-set-H-ell-ind-sub-free},
we argue that the \FPT\ algorithms for domination problems
cannot be generalized to this larger graph class.
We complement this with an \XP-algorithm, which is optimal under the \ETH.
In Subsection~\ref{subsec:multicut-H-ell-ind-sub-free},
we present a simple algorithm to prove that \MCundel is \paraNP-\hard\ on this graph class.
This implies that the \XP-algorithm presented in Section~\ref{sec:multicut-bounded-leafage} cannot be generalized for this larger class.

\subsection{Dominating Set and related problems}
\label{subsec:dominating-set-H-ell-ind-sub-free}

In this subsection, we prove Theorem~\ref{thm:dom-set-ind-sub-w-hard}.
We first show the hardness results of the theorem
and provide afterwards the \XP-algorithms for the problems.

\begin{lemma}
  \label{lemma:domset-H-ell-ind-sub-free-w-hard}
  {\sc Dominating Set}, {\sc Connected Dominating Set} and {\sc Steiner Tree} on $H_{\ell}$-induced-subgraph-free chordal graphs are {\em \WoneH} when parameterized by $\ell$ and assuming the {\em \ETH}, do not admit an algorithm running in time $f(\ell)\cdot n^{o(\ell)}$ for any computable function $f$.
\end{lemma}

\begin{proof}
We present a parameter preserving reduction from \textsc{Multicolored Independent Set}.
An instance of this problem consists of a graph $G$, an integer $q$, and a partition $(V_1, \dots, V_{q})$ of $V(G)$.
  The objective is to determine whether $G$ has an independent set which contains exactly one vertex from every part $V_i$.
  We assume, without loss of generality, that each $V_i$ is an independent set.
  We present a slight modification of a known reduction (see \cite[Theorem~$13.9$]{DBLP:books/sp/CyganFKLMPPS15}).

\paragraph*{Reduction.}
The reduction takes as input an instance $(G,q,(V_1,\dots,V_q))$ of \textsc{Multicolored Independent Set} and
constructs a graph $G'$ as follows.
  \begin{itemize}
    \item
    For every vertex $v \in V (G)$, the reduction introduces a vertex $v$ into $G'$:
    we denote by $C$ the set of all these vertices in $G'$.
    Note that the sets $V_i$ carry over directly to $G'$.
    \item
    The reduction turns the set $C$ into a clique in $G'$ by adding edges between any two distinct vertices of $C$.
    \item
    For every $i \in [q]$,
    the reduction introduces two new vertices $x_i, y_i$ into $G'$
    and makes them adjacent to every vertex of $V_i$.
    \item
    For every edge $e=uv \in E(G)$ with endpoints $u \in V_i$ and $v \in V_j$, the reduction
    introduces a vertex $w_e$ into $G'$
    and makes it adjacent to every vertex of $(V_i \cup V_j ) \setminus \{u, v\}$.
  \end{itemize}
  For \DS and \connDS the reduction returns the instance $(G', q)$.
  For \SteinTree, it sets all the vertices in $V(G') \setminus C$ as terminals and returns the instance $(G', V(G') \setminus C, q)$.

  \paragraph*{Correctness.}
  In the following claim, we prove that the reduction produces equivalent instances.
  We only prove the claim for \DS;
  the correctness for the other two problems follows immediately
  from the design of the graph $G'$.

  \begin{claim}
  $(G,q,(V_1,\dots,V_q))$ is a \yes-instance for \textsc{Multicolored Independent Set}
  if and only if $(G',q)$ is a \yes-instance for \DS.
  \end{claim}
  
  \begin{claimproof}
  Assume that $(G,q,(V_1,\dots,V_q))$ is a \yes-instance for \textsc{Multicolored Independent Set}
  and let $I$ be an independent set $I$ of $G$
  containing one vertex from each $V_i$.
  We claim that $I$ is a dominating set in $G'$.
  Since for every $i \in [q]$, $I \cap V_i  \neq \emptyset$,
  the set $I$ dominates every vertex in $V_i \cup  \{x_i,y_i\}$.
  For an edge $e = uv \in E(G)$ where $u \in V_i$ and $v \in V_j$, consider the vertex $w_e$.
  As $u$ and $v$ are adjacent, at least one of them is not in $I$, say $u \notin I$ without loss of generality.
  Since $I \cap V_i \neq \emptyset$, there must then exist $w \in V_i \setminus \{u\}$ such that $w \in I$;
  but $w_e$ is adjacent to $w$ by construction and thus, $I$ dominates $w_e$. 

Conversely, assume that $(G',q)$ is a \yes-instance for \DS and 
  let $D$ be a dominating set of size $q$ in $G'$.
  We claim that $D$ is also an independent set in $G$.
  Since for every $i \in [q]$, $D$ dominates the vertices $x_i$ and $y_i$,
  $D$ has to contain at least one vertex from $V_i \cup \{x_i,y_i\}$;
  and since $x_i$ and $y_i$ are not adjacent, in fact $D$ must contain a vertex from $V_i$.
  As these sets are disjoint for different values of $i$ and $|D| \leq q$,
  it follows that $D$ contains exactly one vertex from each $V_i$:
  let $v_1 \in V_1, \dots, v_q \in V_q$ be the vertices of $D$.
  Now suppose for a contradiction that $v_i$ and $v_j$ are the endpoints of an edge $e$.
  By construction, vertex $w_e$ in $G'$
  is adjacent only to $(V_i \cup V_j ) \setminus \{v_i, v_j \}$
  and hence, $D$ does not dominate $w_e$, a contradiction.
  \end{claimproof}

  The following claim holds for all three problems
  as it only depends on the structure of the graph $G'$.

  \begin{claim}
    $G'$ does not contain $H_{2q + 2}$ as an induced subgraph.
  \end{claim}
  
  \begin{claimproof}
  We first partition the vertex set of $V(G')$.
  For this, let $I = V(G') \setminus C$.
  It is easy to see that $I$ is an independent set
  and since $C$ is a clique in $G'$, 
  $G'$ is in fact a split graph with split partition $(C, I)$.
  For each integer $i \in [q]$,
  the vertices in $I$ can be partitioned into the following three sets
  depending on their adjacency in $V_i$.
  \begin{enumerate}
    \item
    Vertices that are adjacent to \emph{all} vertices in $V_i$:
    these are the vertices $x_i, y_i$.
    \item
    Vertices that are adjacent to \emph{all but one} vertex in $V_i$: these are the vertices of type $w_e$ for edges $e$ with one endpoint in $V_i$.
    \item
    Vertices that are adjacent to \emph{no} vertex in $V_i$:
    these are the vertices of type $w_e$ for edges $e$
    with both endpoints are outside $V_i$,
    and the vertices $x_{i'},y_{i'}$ where $i \neq i'$.
  \end{enumerate}
  Recall that by assumption, $V_i$ is an independent set in $G$ and thus, there is no edge with both endpoints in~$V_i$.

Now suppose, for the sake of contradiction, that $G'$ contains $H_{2q + 2}$ as an induced subgraph.
  Consider the (unique) split partition $(H_C, H_I)$ of $H_{2q + 2}$.
  Let $H_C = \{v_1, v_2, \dots, v_{2q + 2}\}$
  and $H_I = \{u_1, u_2, \dots, u_{2q + 2}\}$.
  Moreover, for every $i \in [2q + 2]$, edge $v_iu_i$ is in $E(H_{2q + 2})$.
  Consider the clique $H_C$ in $G'$.
  As $I$ is an independent set, $|H_C \cap I| \le 1$.
  Hence, $H_C$ contains at least $2q + 1$ vertices of $C$.
  By the Pigeon-Hole principle, there must then exist an integer $i \in [q]$
  such that $|H_C \cap V_i| \ge 3$:
 let  $v_1, v_2, v_3$ be three vertices of $H_C \cap V_i$.

Since by construction, $u_1$ is not adjacent to $v_2$ and $v_3$,
and  $v_2, v_3$ are in $C$, it must be that $u_1 \in I$.
But then, $u_1$ is adjacent to one vertex in $V_i$, namely $v_1$,
and nonadjacent to two vertices in $V_i$, namely $v_2$ and $v_3$,
a contradiction to the fact that vertices in $I$ can be partitioned into the three sets described above.
Therefore, $G$ does not contain $H_{2q + 2}$ as an induced subgraph.
  \end{claimproof}

  It is known that, assuming the \ETH, there is no algorithm
  that can solve \textsc{Multicolored Independent Set}
  on instance $(G, q, (V_1, V_2, \dots V_{q}))$
  in time $f(q) \cdot |V(G)|^{o(q)}$ for any computable function $f$
  (see, e.g., \cite[Corollary~14.23]{DBLP:books/sp/CyganFKLMPPS15}).
  Note finally, that $|V(G')| \in \calO(|V(G)|^2)$
  and $G'$ is an $H_{2q + 2}$ induced-subgraph-free split graph.
  These facts, together with arguments that are standard for parameter preserving reductions,
  concludes the proof of the lemma.
\end{proof}

In the following, we give the \XP-algorithms for the three problems.
Instead of giving the algorithm for \DS,
we give an algorithm for the more general \RBDS.
Recall that, from \cref{lemma:ds-to-rbds-leafage},
there is a reduction from the former to the latter problem.
There remains to argue that this reduction preserves the property
of being $H_\ell$-induced-subgraph-free.

\begin{lemma}
\label{lemma:ds-to-rbds-H-ell-ind-sub-free}
There is a polynomial-time algorithm that given an instance $(G, k)$ of {\sc \DS} constructs an {equivalent} instance $(G', (R',B'), k)$ of {\sc \RBDS} such that if $G$ is a $H_{\ell}$-induced-subgraph-free graph, then so is $G'$.
\end{lemma}

\begin{proof}
  As in \cref{lemma:ds-to-rbds-leafage},  we construct $G'$ from $G$ as follows.
  For every vertex $v \in V(G)$,
  add two copies $v_R$ and $v_B$ to $V(G')$
  and add an edge $v_R v_B$ to $E(G')$.
  For every edge $uv \in E(G)$,
  add edges $v_R u_R$, $v_R u_B$, $v_B u_R$, and $v_B u_B$ to $E(G')$.
  This completes the construction of $G'$.
  By the proof of \cref{lemma:ds-to-rbds-leafage}, it is known that these two instances are equivalent.
  In the following, we let $R' = \{v_R \mid v\in V(G)\}$ and $B' = \{v_B \mid v\in V(G)\}$.

  Now assume that $G$ is $H_{\ell}$-induced-subgraph-free 
  and suppose, for the sake of contradiction, that $G'$
  contains $H_\ell$ as an induced subgraph.
  Let $I$ be the vertices forming the independent set
  and $C$ the vertices forming the clique of $H_\ell$.
  We claim that for no vertex $v\in V(G)$, we have that $v_B, v_R \in C \cup I$.
  Note that if the claim holds, then using the original version of each vertex would give
  an induced $H_\ell$ in $G$ and thus contradict our assumption.

  There remains to prove the claim.
  To this end, consider $v \in V(G)$.
  Since $I$ is an independent set,
  $v_B$ and $v_R$ cannot both be contained in $I$.
  Moreover, it can also not be the case that $v_B\in I$ and $v_R \in C$
  (or vice-versa) as then $v_B$ would also be adjacent to all vertices in $C$.
  Hence, assume that $v_B,v_R \in C$.
  Assume, without loss of generality, that $u_B \in I$ is the unique vertex adjacent to $v_B$ in $C$
  (the case where $u_R \in I$ is the unique adjacent vertex is symmetric).
  Since there is an edge from $v_B$ to $u_B$,
  we know that $u$ and $v$ are adjacent in $G$.
  Hence, by construction, there must also be an edge from $v_R$ to $u_B$
  which contradicts the fact that we have an $H_\ell$ graph.
\end{proof}

\newcommand{\Tone}{\ensuremath{\mathsf{T_1}}\xspace}
\newcommand{\Ttwo}{\ensuremath{\mathsf{T_2}}\xspace}

We are now ready to show that \RBDS on chordal graphs admits an \XP-algorithm
if the input graph does not contain $H_\ell$ as induced subgraph.

\begin{lemma}
  {\sc Red-Blue Dominating Set} restricted to $H_{\ell}$-induced-subgraph-free chordal graphs
  admits an algorithm running in time $n^{\Oh(\ell)}$.
\end{lemma}

\begin{proof}
Let $(G,(R,B),k)$ be an instance of \RBDS where $G$ is an $H_\ell$-induced-subgraph-free chordal graph,
and let $(T,\mld)$ be a tree representation of $G$.
First, we add a node $\rr$ to $T$ by connecting it to an arbitrary node of $T$ and root $T$ at $\rr$
(note that, by construction, no model in $\mld$ contains $\rr$).
We use dynamic programming to compute the entries of two tables $\Tone$ and $\Ttwo$ in a bottom-up traversal of $T$.
The contents of $\Tone$ and $\Ttwo$ are defined as follows.
For every node $\alpha \in V(T)$
and every \emph{nonempty} set $X \subseteq \containR{\alpha}$ of size at most $\ell$,
\begin{gather*}
  \Tone[\alpha, X] \deff
  \min \{ \abs{S} \mid
    S\subseteq \interR{\alpha},
    S \cap \containR{\alpha} = X,
    N[S] \supseteq \interB{\alpha}
    \}
\end{gather*}
Intuitively, this stores the (size of the) smallest set of red vertices containing $X$
such that all blue vertices in $T_\alpha$ are dominated.

For every node $\alpha \in V(T)$
and every set $Y \subseteq \underR{\alpha}$ of size at most $\ell$,
\begin{gather*}
  \Ttwo[\alpha,Y] \deff
  \min \{ \abs{S} \mid
    S\subseteq \underR{\alpha},
    N[S] \supseteq \underB{\alpha}  \cup (N(Y) \cap \containB{\alpha})
    \}
\end{gather*}
Intuitively, this stores the (size of the) smallest set of red vertices
intersecting with $T_\alpha$ but not $\alpha$
which dominate all blue vertices below $\alpha$
and the $\alpha$-blues that are neighbors of the red vertices in $Y$.

Initially, every entry of $\Tone$ and $\Ttwo$ is set to $+\infty$.
The output is \yes\ if and only if $\Ttwo[\rr,\emptyset] \leq k$.
We next show how to update the entries of $\Tone$ and $\Ttwo$.

\paragraph*{Updating the Leaves.}
Let $\alpha\in V(T)$ be a leaf of $T$.
Then set
\[
  \Ttwo[\alpha,\emptyset] = 0
\]
and for every nonempty set $X \subseteq \containR{\alpha}$ of size at most $\ell$, set
\[
  \Tone[\alpha,X] = |X|.
\]

\paragraph*{Updating Internal Nodes.}
Let $\alpha \in V(T)$ be an internal node of $T$
and let $\beta_1,\ldots,\beta_p$ be the children of $\alpha$.
To update the entries of $\Tone[\alpha,\cdot]$, we proceed as follows.
Let $X \subseteq \containR{\alpha}$ be a nonempty set of size at most $\ell$.
Denote by $I \subseteq [p]$ the set of indices $i \in [p]$ such that $X \cap \containR{\beta_i} \neq \emptyset$ and set $\overline{I} = [p] \setminus I$.
For every $i \in I$, further let $X_i = X \cap \containR{\beta_i}$.
We update $\Tone[\alpha,X]$ according to the following procedure.\\

\begin{enumerate}
  \item[1.]
  \label{step:dsXP:tOne:first}
  For every $i \in I$, set
  \[
    m_i = \min_{\substack{Z \subseteq \containR{\beta_i} \setminus \containR{\alpha} \\ \text{s.t. }|Z| + |X_i| \leq \ell}} \Tone[\beta_i,Z \cup X_i].
  \]
  \item[2.]
  \label{step:dsXP:tOne:second}
  For every $i \in \overline{I}$, let
  \[
    \mathcal{Y}_i = \{Z \subseteq \underR{\beta_i}  \mid |Z| \leq \ell \text{ and } \containB{\beta_i} \setminus \containB{\alpha} \subseteq N(Z)\}
  \]
  and set
  \[
    m_i^1 = \min_{\substack{Z \subseteq \containR{\beta_i} \setminus \containR{\alpha} \\ \text{s.t. } 1 \leq |Z| \leq \ell}}
      \Tone[\beta_i,Z]
      \qquad
      \text{ and }
      \qquad
    m_i^2 = \min_{Z \in \mathcal{Y}_i} \Ttwo[\beta_i,Z].
  \]
  \item[3.]
  Set
  \[
    \Tone[\alpha,X] = |X| + \sum_{i \in I} m_i - |X_i| + \sum_{i \in \overline{I}} \min\{m_i^1,m_i^2\}.
  \]
\end{enumerate}

To update the entries of $\Ttwo[\alpha,\cdot]$, we proceed as follows.
Let $Y \subseteq \underR{\alpha}$ be a set of size at most $\ell$.
Denote by $I \subseteq [p]$ the set of indices $i \in [p]$
such that $Y \cap \interR{\beta_i} \neq \emptyset$
and set $\overline{I} = [p] \setminus I$.
We update $\Ttwo[\alpha,Y]$ according to the following procedure.\\

\begin{enumerate}
  \item[1.]
  \label{step:dsXP:tTwo:one}
  Initialise $\mathsf{OPT}_{\overline{I}} = 0$ and $\mathsf{OPT}_I = +\infty$.
  \item[2.]
  \label{step:dsXP:tTwo:two}
  For every $i \in \overline{I}$ do:
  \begin{enumerate}
      \item[2.a.]
    \label{step:dsXP:tTwo:twoOne}
    Let
    \[
      \mathcal{Y}_i = \{Z \subseteq \underR{\beta_i}
        \mid |Z| \leq \ell \text{ and } \containB{\beta_i} \setminus \containB{\alpha} \subseteq N(Z)\}
    \]
    and set
    \[
      m_i^1 = \min_{\substack{Z \subseteq \containR{\beta_i} \setminus \containR{\alpha} \\ \text{s.t. } 1 \leq |Z| \leq \ell}}
        \Tone[\beta_i,Z]
        \qquad
        \text{ and }
        \qquad
      m_i^2 = \min _{Z \in \mathcal{Y}_i} \Ttwo[\beta_i,Z].
    \]
    \item[2.b.]
    Set $\mathsf{OPT}_{\overline{I}} = \mathsf{OPT}_{\overline{I}} + \min \{m_i^1, m_i^2\}$.
  \end{enumerate}
  \item[3.]
  \label{step:dsXP:tTwo:three}
  For every partition $N= \{N_i \mid i \in I\}$ of $N(Y) \cap \containB{\alpha}$
  where for every $i \in I$, $N_i \subseteq N(Y_i) \cap \containB{\alpha}$ do:
  \begin{enumerate}
    \item[3.a.]
    \label{step:dsXP:tTwo:threeOne}
    Initialise $\mathsf{Int}_N = 0$.
    \item[3.b.]
    \label{step:dsXP:tTwo:threeTwo}
    For every $i \in I$ do:
    \begin{enumerate}
      \item[3.b.i.]
      \label{step:dsXP:tTwo:threeTwoOne}
      Let
      \[
        \mathcal{Y}_i^N = \{ Z \subseteq \underR{\beta_i} \mid |Z| \leq \ell \text{ and } N_i \cup (\containB{\beta_i}\setminus \containB{\alpha}) \subseteq N(Z)\}
      \]
      and set
      \[
        m_i^1 = \min_{\substack{Z \subseteq \containR{\beta_i} \setminus \containR{\alpha} \\ \text{s.t. } 1 \leq |Z| \leq \ell}}
          \Tone[\beta_i,Z] \text{ and }
        m_i^2 = \min_{Z \in \mathcal{Y}_i^N} \Ttwo[\beta_i,Z].
      \]
      \item[3.b.ii.]
      Set $\mathsf{Int}_N = \mathsf{Int}_N + \min \{m_i^1,m_i^2\}$.
    \end{enumerate}
    \item[3.c.]
    \label{step:dsXP:tTwo:threeThree}
    Set $\mathsf{OPT}_I = \min \{\mathsf{OPT}_I, \mathsf{Int}_N\}$.
  \end{enumerate}
  \item[4.]
  Set $\Ttwo[\alpha,Y] = \mathsf{OPT}_{\overline{I}} + \mathsf{OPT}_I$.
\end{enumerate}

\bigskip

We next show that the entries of $\Tone[\alpha,\cdot]$ and $\Ttwo[\alpha,\cdot]$ are updated correctly.
To this end, we first introduce some useful notation.
Given a set $X \subseteq B$,
a set $S \subseteq R$ \emph{minimally dominates} $X$
if $X \subseteq N(S)$ and for every $x \in S$,
$X \not\subseteq N(S \setminus \{x\})$.
Additionally, we prove the following.

\begin{claim}
\label{clm:size}
For every node $\alpha \in V(T)$, the following hold.
\begin{enumerate}
\item[(i)]
For every minimum red-blue dominating set $S$ of $G$, $|S \cap \containR{\alpha}| \leq \ell$.
\item[(ii)]
For every set $X \subseteq \containB{\alpha}$ and every set $Y \subseteq R \setminus \containR{\alpha}$ minimally dominating $X$,
$|Y| \leq \ell$.
\end{enumerate}
\end{claim}

\begin{claimproof}
To prove item (i), let $S$ be a minimum red-blue dominating set of $G$.
Since $S$ is minimum, for every $x \in S \cap \containR{\alpha}$, there exists $p_x \in N(x) \cap B$
such that $p_x \notin \bigcup_{y \in S \setminus \{x\}} N(y)$,
i.e.,\ the blue vertex $p_x$ is only dominated by $x$.
Then $\{p_x~|~x \in S \cap \containR{\alpha}\}$ is an independent set:
indeed, if there exist $x,y \in S \cap \containR{\alpha}$ such that $p_xp_y \in E(G)$ then $x,p_x,p_y,y$ induces a $C_4$,
a contradiction as $G$ is chordal.
It follows that $(S \cap \containR{\alpha}) \cup \{p_x~|~x \in S \cap \containR{\alpha}\}$ induces an $H_{|S \cap \containR{\alpha}|}$ and so,
$|S \cap \containR{\alpha}| < \ell$.

To prove item (ii), let $X \subseteq \containB{\alpha}$ and let $Y\subseteq R\setminus \containR{\alpha}$ be a set minimally dominating $X$.
Since $Y$ is minimal, for every $x \in Y$, there exists $p_x \in N(x) \cap X$
such that $p_x \notin \bigcup_{y \in Y \setminus \{x\}} N(y)$,
i.e.,\ the blue vertex $p_x$ is only dominated by $x$.
This implies that $Y$ is an independent set:
indeed, if there exist $x,y \in Y$ such that $xy \in E(G)$ then $x,p_x,p_y,y$ induces a $C_4$,
a contradiction as $G$ is chordal.
It follows that $Y \cup \{p_x~|~x \in Y\}$ induces an $H_{|Y|}$ and so, $|Y| < \ell$.
\end{claimproof}

We now move towards proving the correctness of the update procedure.
We start with the first table.

\begin{claim}
\label{clm:T1}
For every internal node $\alpha \in V(T)$, the entries of $\Tone[\alpha,\cdot]$ are updated correctly.
Furthermore, $\Tone[\alpha,\cdot]$ can be updated in $n^{\Oh(\ell)}$-time.
\end{claim}

\begin{claimproof}
Let $\alpha \in V(T)$ be an internal node of $T$
with children $\beta_1,\ldots,\beta_p$
and assume that for every $i \in [p]$, $\Tone[\beta_i,\cdot]$ and $\Ttwo[\beta_i,\cdot]$ have been correctly filled.
Let us first show that for every nonempty set $X \subseteq \containR{\alpha}$ of size at most $\ell$,
there exists a set $S \subseteq \interR{\alpha}$ of size $\Tone[\alpha,X]$ such that $S \cap \containR{\alpha} = X$ and $S$ dominates every vertex in $\containB{\alpha}$.

Consider a nonempty set $X \subseteq \containR{\alpha}$ of size at most $\ell$.
Let $I \subseteq [p]$ be the set of indices $i \in [p]$ such that $X \cap \containR{\beta_i} \neq \emptyset$ and set $\overline{I} = [p] \setminus I$.
For every $i \in I$, further let $X_i = X \cap \containR{\beta_i}$.
For every $i \in I$,
let $m_i$ be as defined in Step~\ref{step:dsXP:tOne:first} and
let $Z_i \subseteq \containR{\beta_i} \setminus \containR{\alpha}$ be a set such that $|Z_i| + |X_i| \leq \ell$ and $m_i = \Tone[\beta_i,Z_i \cup X_i]$.

Then, since for every $i \in I$, $\Tone[\beta_i,\cdot]$ has been correctly filled,
there exists a set $S_i \subseteq \interR{\beta_i}$ of size $\Tone[\beta_i,Z_i \cup X_i]$
such that $S_i \cap \containR{\beta_i} = Z_i \cup X_i$ and $S_i$ dominates every vertex in $\interB{\beta_i}$.
Similarly, for every $i \in \overline{I}$,
let $m_i^1$ and $m_i^2$ be as defined in Step~\ref{step:dsXP:tOne:second}.
Further let $\overline{I}_1 \subseteq \overline{I}$ be the set of indices $i \in \overline{I}$ such that $\min\{m_i^1,m_i^2\} = m_i^1$
and set $\overline{I}_2 = \overline{I} \setminus \overline{I}_1$.
For every $i \in \overline{I}_1$, let $Z_i \subseteq \containR{\beta_i} \setminus \containR{\alpha}$ be a set
such that $1 \leq |Z_i| \leq \ell$ and $m_i^1 = \Tone[\beta_i,Z_i]$;
and for every $i \in \overline{I}_2$, let $Z_i \subseteq \underR{\beta_i} $ be a set of size at most $\ell$
such that $\containB{\beta_i} \setminus \containB{\alpha} \subseteq N(Z_i)$ and $m_i^2 = \Ttwo[\beta_i,Z_i]$.

Then, since for every $i \in \overline{I}_1$, $\Tone[\beta_i,\cdot]$ has been correctly filled,
there exists a set $S_i \subseteq \interR{\beta_i}$ of size $\Tone[\beta_i,Z_i]$
such that $S_i \cap \containR{\beta_i} = Z_i \cup X_i$ and $S_i$ dominates every vertex in $\interB{\beta_i}$;
similarly, since for every $i \in \overline{I}_2$, $\Ttwo[\beta_i,\cdot]$ has been correctly filled,
there exists a set $S_i \subseteq \interR{\beta_i} \setminus \containR{\beta_i}$ of size $\Ttwo[\beta_i,Z_i]$
such that $S_i$ dominates every vertex in $\underB{\beta_i}  \cup (N(Z_i) \cap \containB{\beta_i})$.

We contend that the set $M = X \cup \bigcup_{i \in [p]} S_i$ is the desired $S$.
Indeed, observe first that, by the update step, $\Tone[\alpha,X] = |X| + \sum_{i \in I} |S_i| - |X_i| + \sum_{i \in \overline{I}} |S_i| = |M|$.
Let us next show that $M \cap \containR{\alpha} = X$.
Since for every $i \in I$, $\Tone[\beta_i,\cdot]$ is correctly filled,
$S_i \cap \containR{\beta_i} = X_i \cup Z_i$ where $Z_i \subseteq \containR{\beta_i} \setminus \containR{\alpha}$ by construction;
similarly, for every $i \in \overline{I}_1$, $S_i \cap \containR{\beta_i} = Z_i$ where $Z_i \cap \containR{\alpha} = \emptyset$
since $Z_i \subseteq \underR{\beta_i} $ by definition.
Now by construction, for every $i \in \overline{I}_2$,
$S_i \cap \containR{\alpha} = \emptyset$ since $S_i \subseteq \underR{\beta_i}$;
thus, $M \cap \containR{\alpha} = \bigcup_{i \in I} X_i = X$ as claimed.

Let us finally show that $M$ dominates every vertex in $\interB{\alpha}$.
First observe that since $X \neq \emptyset$, every vertex in $\containB{\alpha}$ is dominated by $M$.
Consider therefore a vertex $x \in \underB{\alpha}$.
Then there exists $i \in [p]$ such that $x \in \interB{\beta_i}$.
If $i \in I$ then $x$ is dominated by $S_i$ by definition;
similarly, if $i \in \overline{I}_1$ then $x$ is dominated by $S_i$ by definition.
Thus, suppose that $i \in \overline{I}_2$.
Then either $x \in \underB{\beta_i} $ in which case $x$ is dominated by $S_i$ by definition;
or $x \in \containB{\beta_i}$ and since $x \notin \containB{\alpha}$ by assumption, $x \in N(Z_i)$ by construction and thus,
$x$ is dominated by $S_i$ by definition.
Therefore, $M$ dominates every vertex in $\interB{\alpha}$ and so, $M$ is indeed the desired~$S$.

Consider now a minimum red-blue dominating set $S$ of $G$ such that $S \cap \containR{\alpha} \neq \emptyset$.
Then by \Cref{clm:size}(i), $|S \cap \containR{\alpha}| \leq \ell$.
Let us show that $|S \cap \interR{\alpha}| \geq \Tone[\alpha,S \cap \containR{\alpha}]$.
Denote by $X = S \cap \containR{\alpha}$ and for every $i \in [p]$, let $S_i = S \cap \interR{\beta_i}$.
Further let $I \subseteq [p]$ be the set of indices $i \in [p]$ such that $S_i \cap \containR{\beta_i} \neq \emptyset$
and set $\overline{I} = [p] \setminus I$.
By \Cref{clm:size}(i), for every $i \in I$, $|S_i \cap \containR{\beta_i}| \leq \ell$
and since $\Tone[\beta_i,\cdot]$ has been correctly filled, $|S_i| \geq \Tone[\beta_i,S_i \cap \containR{\beta_i}]$.
Now consider $i \in \overline{I}$.
Since $S$ is dominating and $S_i \cap \containR{\beta_i} = \emptyset$,
every vertex in $\containB{\beta_i} \setminus \containB{\alpha}$
must be dominated by some vertex in $S \cap \underR{\beta_i}$:
let $S^*_i \subseteq S \cap \underR{\beta_i} $ be a set minimally dominating $\containB{\beta_i} \setminus \containB{\alpha}$.
Then by \Cref{clm:size}(ii), $|S^*_i| \leq \ell$ and
since $\Ttwo[\beta_i,\cdot]$ has been correctly filled, $|S_i| \geq \Ttwo[\beta_i,S^*_i]$.
Thus, we conclude by the update step and the above that
\begin{equation*}
\begin{split}
\Tone[\alpha,X] &\leq |X| + \sum_{i \in I} \Tone[\beta_i,S_i \cap \containR{\beta_i}] - |S_i \cap X| + \sum_{i \in \overline{I}} \Ttwo[\beta_i,S^*_i]\\
&\leq |X| + \sum_{i \in I} |S_i| - |S_i \cap X| + \sum_{i \in \overline{I}} |S_i| = |S \cap \interR{\alpha}|
\end{split}
\end{equation*}
as claimed.
Now by observing that $S \cap \interR{\alpha}$ is a minimum-sized set
dominating every vertex in $\interB{\alpha}$ and whose intersection with $\containR{\alpha}$ is $X$
($S$ would otherwise not be minimum),
we conclude by the above that $\Tone[\alpha,\cdot]$ is updated correctly.

Finally, it is not difficult to see that it takes $n^{\Oh(\ell)}$-time to update one entry of $\Tone[\alpha,\cdot]$
and since there are $n^{\Oh(\ell)}$ entries, the claim follows.
\end{claimproof}

We next show the correctness of the update procedure for the second table.

\begin{claim}
\label{clm:T2}
For every internal node $\alpha \in V(T)$, the entries of $\Ttwo[\alpha,\cdot]$ are updated correctly.
Furthermore, $\Ttwo[\alpha,\cdot]$ can be updated in $n^{\Oh(\ell)}$-time.
\end{claim}

\begin{claimproof}
Let $\alpha \in V(T)$ be an internal node of $T$
with children $\beta_1,\ldots,\beta_p$
and assume that for every $i \in [p]$, $\Tone[\beta_i,\cdot]$ and $\Ttwo[\beta_i,\cdot]$ have been correctly filled.
Let us first show that for every set $Y \subseteq \underR{\alpha}$ of size at most $\ell$,
there exists a set $S$ of size $\Ttwo[\alpha,Y]$ such that $S$ dominated every vertex in $\underB{\alpha} \cup (N(Y) \cap \containB{\alpha})$.

Consider a set $Y \subseteq \underR{\alpha}$ of size at most $\ell$.
Let $I \subseteq [p]$ be the set of indices $i \in [p]$ such that $Y \cap \interR{\beta_i} \neq \emptyset$ and set $\overline{I} = [p] \setminus I$.
For every $i \in \overline{I}$,
let $\mathcal{Y}_i$, $m_i^1$ and $m_i^2$ be as defined
in Step~\ref{step:dsXP:tTwo:twoOne}.
Further let $\overline{I}_1 \subseteq \overline{I}$ be the set of indices $i \in \overline{I}$ such that $\min\{m_i^1,m_i^2\} = m_i^1$
and set $\overline{I}_2 = \overline{I} \setminus \overline{I}_1$.
Let $N = \{N_i~|~i \in I\}$ be a partition of $N(Y) \cap \containB{\alpha}$
as considered in Step~\ref{step:dsXP:tTwo:three}
such that the final value of $\mathsf{Int}_N$ is minimum among all such final values
taken over every partition of $N(Y) \cap \containB{\alpha}$
as considered in Step~\ref{step:dsXP:tTwo:three}.
For every $i \in I$,
let $\mathcal{Y}_i^N$, $m_i^1$ and $m_i^2$ be as defined
in Step~\ref{step:dsXP:tTwo:threeTwoOne}.
Let $I_1 \subseteq I$ be the set of indices $i \in I$ such that $\min\{m_i^1,m_i^2\} = m_i^1$
and set $I_2 = I \setminus I_1$.

For every $i \in I_1 \cup \overline{I}_1$, let $Z_i \subseteq \containR{\beta_i} \setminus \containR{\alpha}$ be a nonempty set of size at most $\ell$
such that $m_i^1 = \Tone[\beta_i,Z_i]$.
Then, since for every $i\in I_1 \cup \overline{I}_1$, $\Tone[\beta_i,\cdot]$ has been correctly updated,
there exists a set $S_i \subseteq \interR{\beta_i}$ of size $\Tone[\beta_i,Z_i]$ such that
$S_i \cap \containR{\beta_i} = Z_i$ and $S_i$ dominates every vertex in $\interB{\beta_i}$.

Now for every $i \in I_2$, let $Z_i \in \mathcal{Y}_i$ be a set of size at most $\ell$ such that $m_i^2 = \Ttwo[\beta_i,Z_i]$;
similarly, for every $i \in \mathcal{Y}_i^N$, let $Z_i \in \mathcal{Y}_i^N$ be a set of size at most $\ell$ such that $m_i^2 = \Ttwo[\beta_i,Z_i]$.
Then, since for every $i \in I_2 \cup \overline{I}_2$, $\Ttwo[\beta_i,\cdot]$ has been correctly filled,
there exists a set $S_i \subseteq \underR{\beta_i}$ of size $\Ttwo[\beta_i,Z_i]$ such that
$S_i$ dominates every vertex in $\underB{\beta_i} \cup (N(Z_i) \cap \containB{\beta_i})$.

We contend that the set $M = \bigcup_{i \in I} S_i$ is the desired $S$.
Indeed, observe first that, by the update step, $\Ttwo[\alpha,Y] = \sum_{i \in \overline{I}} |S_i| + \sum_{i \in I} |S_i| = |M|$.
Now consider a vertex $x \in \underB{\alpha}  \cup (N(Y) \cap \containB{\alpha})$
and let us show that $x$ is dominated by $M$.

Suppose first that $x \notin \containB{\alpha}$.
Then there exists $i \in I$ such that $x \in \interB{\beta_i}$.
If $i \in I_1 \cup \overline{I}_1$ then $x$ is dominated by $S_i$ by definition.
Suppose therefore that $i \in I_2 \cup \overline{I}_2$.
If $x \in \underB{\beta_i}$ then $x$ is dominated by $S_i$ by definition;
otherwise, $x \in \containB{\beta_i}$ and since $x \notin \containB{\alpha}$ by assumption,
$x \in N(Z_i)$ by construction and so, $x$ is dominated by $S_i$ by definition.

Suppose second that $x \in N(Y) \cap \containB{\alpha}$.
Then there exists $i \in I$ such that $x \in N_i$.
If $i \in I_1$ then $x$ is dominated by $S_i$ by definition.
Suppose therefore that $i \in I_2$.
If $x \in \underB{\beta_i} $ then $x$ is dominated by $S_i$ by definition;
otherwise, $x \in \containB{\beta_i}$ and since $x \in N_i \subseteq N(Z_i)$ by construction,
$x$ is dominated by $S_i$ by definition.
Therefore, $M$ dominates every vertex $\underB{\alpha} \cup (N(Y) \cap \containB{\alpha})$
and so, $M$ is indeed the desired $S$.

Consider now a minimum red-blue dominating set $S$ of $G$ such that $S \cap \containR{\alpha} = \emptyset$
and let $Y \subseteq S \cap \underR{\alpha}$ be a set minimally dominating $N(S \cap \underR{\alpha}) \cap \containB{\alpha}$.
Then by \Cref{clm:size}(ii), $|Y| \leq \ell$.
Let us show that $|S \cap \underR{\alpha}| \geq \Ttwo[\alpha,Y]$.
For every $i \in [p]$, let $S_i = S \cap \interR{\beta_i}$.
Further let $I \subseteq [p]$ be the set of indices $i \in [p]$ such that $Y \cap \interR{\beta_i} \neq \emptyset$
and set $\overline{I} = [p] \setminus I$.
By construction, for every $x \in N(Y) \cap \containB{\alpha}$, there exists $i \in I$ such that $x \in N(Y \cap \interR{\beta_i})$:
let $N= \{N_i~|~i \in I\}$ be a partition of $N(Y) \cap \containB{\alpha}$ where for every $i \in I$,
$N_i \subseteq N(Y \cap \interR{\beta_i})$.

Let $I_1 \subseteq I$ be the set of indices $i \in I$ such that $S_i \cap \containR{\beta_i} \neq \emptyset$
and set $I_2 = I \setminus I_1$.
Then for every $i \in I_1$, $|S_i \cap \containR{\beta_i}| \leq \ell$ by \Cref{clm:size}(i)
and since $\Tone[\beta_i,\cdot]$ has been correctly filled, $|S_i| \geq \Tone[\beta_i,S_i \cap \containR{\beta_i}]$.
Now for every $i \in I_2$, let $Z_i \subseteq \underR{\beta_i}$ be a set minimally dominating $N_i \cup (\containB{\beta_i} \setminus \containB{\alpha})$.
Then for every $i \in I_2$, $|Z_i| \leq \ell$ by \Cref{clm:size}(ii)
(note indeed that $N_i \subseteq \containB{\beta_i}$)
and since $\Ttwo[\beta_i,\cdot]$ has been correctly filled, $|S_i| \geq \Ttwo[\beta_i,Z_i]$.

Similarly, let $\overline{I}_1 \subseteq \overline{I}$ be the set of indices $i \in \overline{I}$ such that $S_i \cap \containR{\beta_i} \neq \emptyset$
and set $\overline{I}_2 = \overline{I} \setminus \overline{I}_2$.
Then for every $i \in \overline{I}_1$, $|S_i \cap \containR{\beta_i}| \leq \ell$ by \Cref{clm:size}(i)
and since $\Tone[\beta_i,\cdot]$ has been correctly filled, $|S_i| \geq \Tone[\beta_i,S_i \cap \containR{\beta_i}]$.
Now for every $i \in \overline{I}_2$, let $Z_i \subseteq \underR{\beta_i}$ be a set minimally dominating $\containB{\beta_i} \setminus \containB{\alpha}$.
Then for every $i \in \overline{I}_2$, $|Z_i| \leq \ell$ by \Cref{clm:size}(ii)
and since $\Ttwo[\beta_i,\cdot]$ has been correctly filled, $|S_i| \geq \Ttwo[\beta_i,Z_i]$.
Thus, we conclude by the update step and the above that
\begin{equation*}
\begin{split}
\Ttwo[\alpha,Y] &\leq \sum_{i \in I_1 \cup \overline{I}_1} \Tone[\beta_i,S_i \cap \containR{\beta_i}]  + \sum_{i \in I_2 \cup \overline{I}_2} \Ttwo[\beta_i,Z_i]\\
&\leq \sum_{i \in I} |S_i| = |S \cap \underR{\alpha}|
\end{split}
\end{equation*}
as claimed.
Now by observing that $S \cap \underR{\alpha}$ is a minimum-sized set dominating
every vertex in $\underB{\alpha} \cup (N(S \cap \underR{\alpha}) \cap \containB{\alpha})$
($S$ would otherwise not be minimum),
we conclude by the above that $\Ttwo[\alpha,\cdot]$ is updated correctly.

Finally, it is not difficult to see that Step~\ref{step:dsXP:tTwo:two}
can be done in $n^{\Oh(\ell)}$-time
and that, similarly, for a fixed partition,
Steps~\ref{step:dsXP:tTwo:threeOne}--\ref{step:dsXP:tTwo:threeThree}
can be done in $n^{\Oh(\ell)}$-time.
Now observe that $|I| \leq \ell$ since $|Y| \leq \ell$
and thus, there are at most $n^{\Oh(\ell)}$ partitions to consider
in Step~\ref{step:dsXP:tTwo:three}.
\end{claimproof}
The lemma now follows from Claims~\ref{clm:T1} and \ref{clm:T2}.
\end{proof}

\subsection{MultiCut with Undeletable Terminals}
\label{subsec:multicut-H-ell-ind-sub-free}

We present a simple reduction from \textsc{Vertex Cover} to \textsc{Multicut with UnDel Term} to prove Theorem~\ref{thm:mc-para-NP-hard}.
Consider an instance $(G, q)$ of \textsc{Vertex Cover} where $G$ has $n$ vertices.
Let $G'$ be a graph obtained from a star with center $r$ and $n + 1$ leaves by subdividing each of its edge once.
Fix an injective mapping $f:V(G) \mapsto V(G')$ such that $f(v)$ is a leaf for every $v \in V(G)$.
Let $w$ be the unique leaf which is not in the range of $f$.
Then, the set of terminal pairs $\calP$ is defined as follows: $\calP= \{(f(u), f(v) \mid uv \in E(G)\} \cup \{(r, w)\}$.
It is easy to see that $(G, q)$ is a yes-instance of \textsc{Vertex Cover} if and only if $(G',\calP, q)$ has a multicut of size at most $q$.
As $G'$ is acyclic, it is $H_3$-induced free.
\section{Other domination-related problems}
\label{sec:otherdompb}

The aim of this section is to complete the proofs of \Cref{thm:ds-fpt} and \Cref{thm:dom-set-ind-sub-w-hard}.
More precisely, we show that \conRBDS\ and \SteinTree\ are \FPT\ parameterized by leafage 
and admit a $n^{\calO(\ell)}$-algorithm on $H_\ell$-induced-subgraph-free chordal graphs.
The two problems are considered in two separate subsections.

\subsection{Connected Red-Blue Dominating Set}

In this subsection, we aim to prove that {\sc Connected Dominating Set} is \FPT\ parameterized by the leafage and admits a $n^{\calO(\ell)}$-algorithm on $H_\ell$-induced-subgraph-free chordal graphs. Formally, we prove the following.

\begin{lemma}
{\sc Connected Red-Blue Dominating Set} is \FPT\ parameterized by the leafage and admits $n^{\calO(\ell)}$-algorithm on $H_\ell$-induced-subgraph-free chordal graphs.
\end{lemma}

To obtain these results, we reduce in both cases to \RBDS and use the algorithms from \Cref{sec:dominating-set} and \Cref{thm:dom-set-ind-sub-w-hard}, respectively. We describe below the reduction and show thereafter that both parameters are preserved. We first start with some useful terminology.\\

A tree representation $(T,\calM)$ of a graph $G$ is \emph{minimal} if for every edge $\alpha\beta \in E(T)$,
the sets $\{x \in V(G)~| \alpha \in \calM(x)\}$ and $\{x \in V(G)~|~\beta \in \calM(x)\}$ are inclusion-wise incomparable.
A tree representation can easily be made minimal by contracting each edge $\alpha\beta \in E(T)$ for which 
the sets $\{x \in V(G)~| \alpha \in \calM(x)\}$ and $\{x \in V(G)~|~\beta \in \calM(x)\}$ are inclusion-wise comparable.
Note that this operation does not increase the number of leaves of the tree representation.

\paragraph{Reduction.} Let $(G,(R_G,B_G),k)$ be an instance of \conRBDS
  and let $(T,\mld)$ be a minimal tree representation of $G$ with $\lf(G)$ leaves.
  We construct an instance $(H,(R_H,B_H),k)$ of \RBDS as follows.
  More precisely, we construct a tree representation $(T_H,\mld_H)$ for $H$ by modifying $(T,\mld)$.

  A node $\alpha \in V(T)$ is called a \emph{red node}
  if $\alpha \notin \bigcup_{x \in B} \mld(x)$,
  that is, $\alpha$ is contained only in models of red vertices.

  Let $T_B$ be the forest obtained by removing every red node in $T$
  and let $\overline{T}_B$ be the smallest connected subtree of $T$ containing $T_B$.
  We further reduce $\overline{T}_B$ according to the following procedure.
  \begin{itemize}
    \item For each leaf $\alpha$ of $\overline{T}_B$ do:
    \begin{itemize}
      \item Let $\beta \in V(\overline{T}_B)$ be the neighbor of $\alpha$.
      \item If $\{x \in B_G \mid \alpha \in \mld(x) \}
      \subseteq \{ x \in B_G \mid \beta \in \mld(x) \}$,
      then set $\overline{T}_B = \overline{T}_B/\alpha\beta$.
    \end{itemize}
  \end{itemize}
  Once the above procedure has been applied to $\overline{T}_B$,
  we subdivide each edge of $\overline{T}_B$ and let $T_H$ be the resulting tree.
  We now define the vertices and edges of $H$ as follows.
  \begin{itemize}
    \item
    For each node $\alpha \in V(T_H)$,
    we add a blue vertex $x$ to $B_H$ with model $\mld_H(x) = \{\alpha\}$.
    \item
    For each red vertex $x \in R_G$ such that $\mld(x) \cap V(\overline{T}_B) \neq \emptyset$,
    we add a red vertex $r_x$ to $R_H$ whose model $\mld_H(r_x)$ in $T_H$
    corresponds to the subdivision of $\mld(x) \cap V(\overline{T}_B)$.
  \end{itemize}
  We next show that these two instances are equivalent.

  \begin{lemma}
    If $(G, (R_G, R_B), k)$ is a \yes-instance for {\sc \conRBDS}
    then $(H, (R_H, B_H),\allowbreak k)$ is a \yes-instance for {\sc \RBDS}.
  \end{lemma}
  
  \begin{proof}
  Let $D_G \subseteq R_G$ be a connected red-blue dominating set of $G$ of size at most $k$
  and let $D_H = \{r_x \in R_H \mid x\in D_G \text{ and } \mld(x) \cap V(\overline{T}_B) \neq \emptyset\}$.
  We contend that $D_H$ is a red-blue dominating set of $H$.

  Indeed, consider a blue vertex $x \in B_H$.
  By construction, there exists a node $\alpha \in V(T_H)$ such that $\mld_H(x) = \{\alpha\}$.
  If $\alpha$ corresponds to a node or an edge of $\overline{T}_B - V(T_B)$ then, since $D_G$ is connected,
  there exists $y \in D_G$ such that $\mld(y)$ contains the node or edge corresponding to $\alpha$;
  but then, $r_y \in D_H$ by construction and so, $x$ is dominated.
  We conclude similarly if $\alpha$ corresponds to an edge between $V(\overline{T}_B) \setminus V(T_B)$ and $V(T_B)$.

  Assume therefore that $\alpha$ corresponds to a node or an edge of $T_B$.
  Suppose first that $\alpha$ is a leaf of $\overline{T}_B$
  and let $\beta$ be the neighbor of $\alpha$ in $\overline{T}_B$.
  Then by construction,
  $\{ x \in B_G \mid \alpha \in \mld(x) \} \setminus
   \{ x \in B_G \mid \beta \in \mld(x) \} \neq \emptyset$
  and so, there exists $y \in D_G$ such that $\alpha \in \mld(y)$ since $D_G$ is dominating;
  but then, $r_y \in D_H$ by construction and so, $x$ is dominated.
  Suppose finally that $\alpha$ corresponds to an edge or an internal node of $T_B$.
  Since $D_G$ is dominating and connected,
  there then exists $y \in D_G$ such that $\alpha \in \mld(y)$;
  but then, $r_y \in D_H$ by construction and so, $x$ is dominated.
  Therefore, $D_H$ is a red-blue dominating set of $H$ and since $|D_H| \leq k$,
  we conclude that $(H,(R_H,B_H),k)$ is a \yes-instance for \RBDS.
  \end{proof}

  \begin{lemma}
    If $(H, (R_H, B_H), k)$ is a \yes-instance for {\sc \RBDS},
    then $(G, (R_G, \allowbreak R_B), k)$ is a \yes-instance {\sc \conRBDS}.
  \end{lemma}
  
  \begin{proof}
  Let $D_H \subseteq R_H$ be a red-blue dominating set of $H$ of size at most $k$
  and let $D_G = \{x \in R_G~|~r_x \in D_H\}$.
  We contend that $D_G$ is a connected red-blue dominating set of $G$.
  Indeed, consider a blue vertex $x \in B_G$.
  Then by construction, there exists a node $\alpha \in V(\overline{T}_B)$ such that $\alpha \in \mld(x)$.
  Since $D_H$ is dominating, there then exists a vertex $r_y \in D_H$
  such that the node in $T_H$ corresponding to $\alpha$ is contained in $\mld_H(r_y)$;
  but then, $x$ is dominated since $y \in D_G$ and $\alpha \in \mld(y)$ by construction.
  Now to see that $D_G$ is connected, observe that if it weren't the case, 
  there would exist an edge $\alpha\beta \in E(\overline{T}_B)$
  such that no model in $\{\mld(y) \mid y \in D_G\}$ contains the edge $\alpha\beta$;
  but then, the vertex in $T_H$ corresponding to the edge $\alpha\beta$ wouldn't be dominated by $D_H$, a contradiction.
  Therefore, $D_G$ is a connected red-blue dominating set of $G$ as claimed
  and $|D_G| \leq k$.
  \end{proof}

Finally, let us show that both parameters are preserved in the above reduction. 
First, it is not difficult to see that the leafage of $H$ is at most that of $G$ 
since the number of leaves of $T_H$ is at most the number of leaves of $T$.
Assume second that $G$ is $H_\ell$-induced-subgraph-free for some $\ell \geq 3$,
and suppose for a contradiction that $H$ contains an induced $H_\ell$.
Let $v_1,u_1,\ldots,v_\ell,u_\ell \in V(H)$ be $2\ell$ such that $H[\{v_iu_i~|~i \in [\ell]\}]$
is isomorphic to $H_\ell$ where $\{v_i~|~i \in [\ell]\}$ is a clique and for every $i \in [\ell]$, $u_iv_i \in E(H)$.
Note that since $B_H$ is an independent set of $H$ and every vertex in $B_H$ is simplicial in $H$, 
$\{v_i~|~i \in [\ell]\} \cap B_H = \emptyset$;
in particular, $\{v_i~|~i \in [\ell]\} \subseteq R_H \subseteq R_G$.
Now for every $i \in [\ell]$, let $\alpha_i$ be a node of $T_H$ defined as follows:
\begin{itemize}
\item if there is a node in $\calM_H(u_i) \cap \calM_H(v_i)$ which correspond to a node in $T$,
then let $\alpha_i$ be any such node.
\item otherwise, $\calM_H(u_i) \cap \calM_H(v_i)$ contains only one node (namely, a node corresponding to an edge of $T$), in which case we let $\alpha_i \in \calM_H(v_i)$ be the neighbor in $\calM_H(v_i)$
of the node in $\calM_H(u_i) \cap \calM_H(v_i)$.
\end{itemize}
Note that, by construction, for every $i \in [\ell]$, $\alpha_i$ corresponds to a node of $T$ 
which is, furthermore, contained in $\calM(v_i)$.
We contend that for every $i \in [\ell]$, there exists $x_i \in V(G) \setminus \{v_j~|~j \in [\ell]\}$ such that
$x_iv_i \in E(G)$ and $x_i$ is nonadjacent to $\{v_j~|~ j \in [\ell] \setminus \{i\}\}$ in $G$,
that is, $\{v_i,x_i~|~i \in [\ell]\}$ induces an $H_{\ell}$ in $G$.
If true, this would contradict the fact that $G$ is $H_\ell$-induced-subgraph-free
and thus conclude the proof.
Let $i \in [\ell]$ and consider a node $\alpha \in \bigcap_{j \in [\ell]} \calM(v_j)$
(note that since subtrees in a tree satisfy the Helly property, this intersection is nonempty).
Further let $\beta \in V(T)$ be the neighbor of $\alpha_i$ on the path in $T$ from $\alpha_i$ to $\alpha$.
Then since $T$ is minimal, 
$I = \{x\in V(G)~|~\alpha_i \in \calM(x)\} \setminus \{x \in V(G)~|~\beta \in \calM(x)\} \neq \emptyset$;
and since $\alpha_i,\beta \in \calM(v_i)$, $v_i \notin I$.
Thus, we may set $x_i = x$ where $x \in I$. 

\subsection{Steiner Tree}

The aim of this section is to prove that \SteinTree is \FPT\ parameterized by the leafage and admits an $n^{\calO(\ell)}$-algorithm on $H_\ell$-induced-subgraph-free chordal graphs. To obtain these results, we give two parameter preserving reductions to \RBDS. We first present a general reduction rule for \SteinTree instances.

\begin{reduction rule}  
\label{rr:independent}
Let $(G,\calT,k)$ be an instance of \SteinTree.
If $G[\calT]$ has a connected component $C$ of size greater than 1,
then return the instance $(G/V(C),(\calT\setminus V(C))  \cup \{v_C\},k-|V(C)|+1)$
where $v_C$ is the vertex resulting from the contraction of $C$ in $G$.
\end{reduction rule}

\begin{lemma}
\Cref{rr:independent} is safe. 
Furthermore, the leafage of $G/V(C)$ is at most that of $G$.
\end{lemma}

\begin{proof}
Suppose that such a connected component $C$ exists.
Assume first that $(G,\calT,k)$ is a \yes-instance for \SteinTree 
and let $S$ be a solution for $(G,\calT,k)$
such that the number of connected component in $S[V(C)]$ is minimum 
amongst all solutions for $(G,\calT,k)$.
We claim that $S[V(C)]$ has only one connected component.
Indeed, suppose to the contrary that $S[V(C)]$ has at least two connected components. 
Since $C$ is connected, 
there exist two connected components $C_1$ and $C_2$ of $S[V(C)]$ such that $C_1$ and $C_2$ are adjacent,
that is, there is an edge $xy \in E(G)$ where $x \in V(C_1)$ and $y \in V(C_2)$.
Let $L = z_1\ldots z_p$ be a shortest path in $S$ from $C_1$ to $C_2$.
Then the tree $S' = S - \{z_1z_2\} + \{xy\}$ is a solution for $(G,\calT,k)$
such that $S'[V(C)]$ contains fewer connected component than $S[V(C)]$, a contradiction to the choice of $S$.
Thus, $S[V(C)]$ has only one connected component
and it is easy to see that $S/V(C)$ is a solution for $(G/V(C),(\calT\setminus V(C))  \cup \{v_C\},k-|V(C)|+1)$.

Conversely, assume that $(G/V(C),(\calT\setminus V(C))  \cup \{v_C\},k-|V(C)|+1)$ is a \yes-instance for \SteinTree
and let $S$ be a solution.
By construction, for every neighbor $y$ of $v_C$ in $S$, there exists $x \in V(C)$ such that $y \in N(x)$:
for every $y \in N(v_C) \cap S$, let $x_y \in V(C)$ be an arbitrary vertex such that $y \in N(x_y)$. 
Set $V = \{x_y~|~y \in N(v_C) \cap S\}$ and for every $x \in V$, denote by $N_x = \{y\in N(v_C) \cap S~|~x_y = x\}$.
Now let $x_1,\ldots,x_p$ be an arbitrary ordering of $V$
and let $y_1,\ldots,y_q$ be an arbitrary ordering of $V(C) \setminus V$.
Then the tree obtained from $S$ by removing the vertex $v_C$ 
to replace it with the path $x_1\ldots x_py_1\ldots y_q$ 
and adding the edges $\{x_iz~|~i \in [p] \text{ and } z \in N_{x_i}\}$
is readily seen to be a solution for $(G,\calT,k)$.

Finally, let us remark that a tree representation for $G/V(C)$ can be obtained from a tree representation $(T,\calM)$ of $G$
by merging the models in $\{\calM(x)~|~x \in V(C)\}$ into a single model representing $v_C$;
in particular, the leafage of $G/V(C)$ is at most that of $G$.
\end{proof}

\begin{lemma}
\label{lem:induceHl}
Let $(G,\calT,k)$ be an instance of {\sc \SteinTree}
and let $(G_R,\calT_R,k)$ be the instance resulting from an exhaustive application of \Cref{rr:independent} to $(G,\calT,k)$.
If $G$ is $H_\ell$-induced-subgraph-free then $G_R$ is $H_{\ell+1}$-induced-subgraph-free.
\end{lemma}

\begin{proof}
Assume that $G[\calT]$ contains at least one connected component of size greater than 1 
(the lemma is trivial otherwise)
and let $C_1,\ldots,C_p$ be all such connected components of $G[\calT]$.
For every $i \in [p]$, denote by $v_{C_i} \in V(G_R)$ the vertex resulting from the contraction of $C_i$.
Now assume that $G$ is $H_\ell$-induced-subgraph-free
and suppose for a contradiction that $G_R$ contains an induced $H_{\ell+1}$.
Let $v_1,u_1,\ldots,v_{\ell+1},u_{\ell+1} \in V(G_R)$ be $2(\ell+1)$ vertices inducing an $H_{\ell+1}$ in $G_R$
where $\{v_i~|~i \in [\ell+1]\}$ is the clique and for every $i \in [\ell+1]$, $v_iu_i \in E(G_R)$.
Since $\{v_{C_i}~|~i \in [p]\}$ is an independent set in $G_R$,
$|\{v_i~|~i \in [\ell+1]\} \cap \{v_{C_i}~|~i \in [p]\}| \leq 1$:
let us assume without loss of generality that $\{v_i~|~i \in [\ell]\} \cap \{v_{C_i}~|~i \in [p]\} = \emptyset$.
On the other hand, if $v_{C_i} = u_{j_i}$ for some $i \in [p]$ and $j_i \in [\ell]$,
then, by construction, there exists $x_i \in V(C_i)$ such that $x_iv_{j_i} \in E(G)$:
let $I \subseteq [p]$ be the set of such indices.
Then $X = \{x_i~|~i \in I\} \cup \{u_i~|~i \in [\ell] \setminus I\}$ is an independent set in $G$ 
where each vertex in $X$ has exactly one neighbor in $K = \{v_i~|~i \in [\ell]\}$,
that is, $K \cup X$ induces an $H_\ell$ in $G$, a contradiction. 
\end{proof}

\begin{lemma}
\label{thm:steinertreefpt}
{\sc \SteinTree} parameterized by the leafage is \FPT.
\end{lemma}

\begin{proof}
As mentioned above, we reduce to \conRBDS: 
given an instance $(G,\calT,k)$ of \SteinTree,
we construct an instance $(H,(R,B),k_H)$ of \conRBDS as follows.
First, we assume that \Cref{rr:independent} has been exhaustively applied to $(G,\calT,k)$.
This implies, in particular, that $\calT$ is an independent set of $G$. 
Let us further assume that $|\calT| > 1$ (the problem is trivial otherwise).
Now let $G^*$ be the supergraph of $G$ obtained by making each terminal simplicial, that is,
for every $t \in \calT$, the neighborhood $N_G(t)$ $(= N_{G^*}(t))$ of $t$ induces a clique in $G^*$.
Observe that the leafage of $G^*$ is at most that of $G$:
indeed, a tree representation for $G^*$ can be obtained from a tree representation $(T,\calM)$ of $G$ as follows.
For every terminal $t \in \calT$, let $\alpha_t \in V(T)$ be a node of $T$ contained the model $\calM(t)$ of $t$.
If there exists a neighbor $x \in N_{G^*}(t)$ such that $\calM(x)$ does not contain $\alpha_t$,
then we extend $\calM(x)$ by adding to it the path in $T$ from $\alpha_t$ to $\alpha_x$
where $\alpha_x \in \calM(x)$ is the closest node to $\alpha_t$ in $T$.
By iterating this process and leaving all the other models intact,
we obtain a tree representation $(T^*,\calM^*)$ for $G^*$ where $T^*$ has the same number of leaves as $T$.  

\paragraph{Reduction.}
We may now construct the graph $(H,(R,B))$:
the set $R = \{r_x~|~x \in V(G)\}$ of red vertices contains a copy of each vertex in $V(G)$
and the set $B = \{b_t~|~t \in \calT\}$ of blue vertices contains a copy of each terminal.
The graph $H[R]$ is then isomorphic to $G^*$ and for every $t \in \calT$, $b_t$ is a true twin to $r_t$.
Finally, we set $k_H = k - |\calT|$.
We next show that the instances $(G,\calT,k)$ and $(H,(R,B),k_H)$ are equivalent.

\begin{claim}
If $(G,\calT,k)$ is a \yes-instance for \SteinTree, then $(H,(R,B),k_H)$ is a \yes-instance for \conRBDS.
\end{claim}

\begin{claimproof}
Assume that $(G,\calT,k)$ is a \yes-instance for \SteinTree and let $S$ be a solution.
Note that since $|\calT| > 1$ by assumption, necessarily $V(S) \setminus \calT \neq \emptyset$.
We contend that the set $D = \{r_x~|~x \in V(S) \setminus \calT\}$ is a solution for $(H,(R,B),k_H)$.
Indeed, it is clear that for every $t \in \calT$, $b_t$ has a neighbor in $D$.
To see that $D$ is connected, observe that if a terminal $t \in \calT$ is not a leaf of $S$,
then $t$ has at least two neighbors in $S$;
but the neighborhood of $r_t$ (and $b_t$) in $R$ is clique in $H$
and so, $D$ is connected.
Since $|D| \leq k - |\calT| = k_H$, we conclude that $D$ is indeed a solution for $(H,(R,B),k_H)$.
\end{claimproof}

\begin{claim}
If $(H,(R,B),k_H)$ is a \yes-instance for \conRBDS, then $(G,\calT,k)$ is a \yes-instance for \SteinTree.
\end{claim}

\begin{claimproof}
Assume that $(H,(R,B),k_H)$ is a \yes-instance for \conRBDS and let $D$ be a minimal solution.
We contend that the set $S = \{x~|~r_x \in D\} \cup \calT$ contains a solution for $(G,\calT,k)$, that is,
for every $t,t' \in \calT$, there is a path from $t$ to $t'$ in $G[S]$.
Observe first that for every $ t\in \calT$, $r_t \notin D$:
indeed, if there exists $t \in \calT$ such that $r_t \in D$, 
then surely $r_t$ has at least one neighbor in $D$ since $|B| = |\calT| > 1$ and $B$ is an independent set in $H$;
but $N[b_t] = N[r_t]$ and $N[r_t]$ is a clique and so, $D \setminus \{r_t\}$ is still a solution for $(H,(R,B),k_H)$,
a contradiction to the minimality of $D$.
This implies, in particular, that $|S| = |D| + |\calT| \leq k_H + |\calT| = k$.
Now since $D$ is dominating and connected,
for every terminal $t,t' \in \calT$, there exists a path $P_{t,t'}$ in $H[D \cup \{b_t,b_{t'}\}]$ from $b_t$ to $b_{t'}$;
but then, it is easy to see that the set 
$\{x \in V(G)~|~r_x \in P_{t,t'}\} \cup \{t'' \in \calT~|~V(P_{t,t'}) \cap N(r_{t''}) \neq \emptyset\} \subseteq S$ 
contains a path from $t$ to $t'$.
Therefore, $S$ is a solution for $(G,\calT,k)$.
\end{claimproof}

Observe finally that a tree representation for $(H,(R,B))$ 
can be obtained from the tree representation $(T^*,\calM^*)$ of $G^*$
by adding a copy of $\calM^*(t)$ for each terminal $t \in \calT$;
in particular, the leafage of $H$ is at most that of $G^*$ which concludes the proof. 
\end{proof}

\begin{lemma}
For every $\ell \geq 3$, {\sc \SteinTree} admits a $n^{\calO(\ell)}$-algorithm on $H_\ell$-induced-subgraph-free chordal graphs.
\end{lemma}

\begin{proof}
As mentioned above, we reduce to \conRBDS: 
given an instance $(G,\calT,k)$ of \SteinTree 
where $G$ is an $H_\ell$-induced-subgraph-free chordal graph,
we construct an instance $(H,(R,B),\allowbreak k_H)$ of \conRBDS as follows.
First, we assume that \Cref{rr:independent} has been exhaustively applied to $(G,\calT,k)$.
This implies, in particular, that $\calT$ is an independent set of $G$. 
Furthermore, by \Cref{lem:induceHl}, $G$ is $H_{\ell+1}$-induced-subgraph-free.
Now the set $R = \{r_x~|~x \in V(G)\}$ of red vertices contains a copy of each vertex in $V(G)$
and the set $B = \{b_t~|~t \in \calT\}$ of blue vertices contains a copy of each terminal.
The graph $H[R]$ is then isomorphic to $G$ and for every $t \in \calT$, $b_t$ is adjacent to only $r_t$.
Furthermore, we set $k_H = k$.
Now it is not difficult to see that these two instances are indeed equivalent:
if $S$ is a Steiner tree for $\calT$ in $G$ then $\{r_x~|~x \in V(S)\}$ is a connected red-blue dominating set of $H$;
and conversely, if $D$ is a connected red-blue dominating set then for every $t \in \calT$, $r_t \in D$
and so, $\{x~|~r_x \in D\}$ contains a Steiner tree for $\calT$.
Finally, it is easily seen that $H$ is $H_{\ell+2}$-induced-subgraph-free
since $\{r_t~|~t \in \calT\}$ is also an independent set in $H$
and for every $t \in \calT$, $N_H(b_t) = \{r_t\}$
(recall that $G$ is $H_{\ell+1}$-induced-subgraph-free after the exhaustive application of \Cref{rr:independent}),
which concludes the proof.
\end{proof}
\section{Conclusion}
In this article, we presented improved and new results regarding domination and cut problems on chordal graphs with bounded leafage.
We presented an \FPT\ algorithm running in time $2^{\calO(\ell)}\cdot \polyn$-time for the \textsc{Dominating Set} problem on chordal graphs, and used it to obtain similar results for the \textsc{Connected Dominating Set} and \textsc{Steiner Tree} problems.
Regarding cut problems, we proved that \textsc{MultiCut with Undeletable Terminals} on chordal graphs is \W[1]-\hard\ when parameterized by the leafage. 
We also presented a polynomial-time algorithm for \textsc{Multiway Cut with Undeletable Terminals} on chordal graphs.
We find it surprising that the complexity of this problem was not known before.
Finally, we examined these problems on $H_{\ell}$-induced-subgraph-free chordal graphs to check the extent of our approach. 

In the case of chordal graphs, we believe the leafage to be a more natural parameter than other popular parameters such as vertex cover, feedback vertex set or treewidth.
It would be interesting to examine the structural parameterized complexity of problems such as
{\sc Longest Cycle},
{\sc Longest Path}, 
{\sc Component Order Connectivity}, 
{\sc $s$-Club Contraction}, 
{\sc Independent Set Reconfiguration}, 
{\sc Bandwidth}, or
{\sc Cluster Vertex Deletion}.
These problems are known to be \NP-complete on split graphs and admit polynomial-time algorithms on interval graphs.
Hence it is plausible that they admit an \FPT\ or \XP\ algorithm on chordal graphs parameterized by the leafage.
We believe it is a representative list, though not exhaustive, of problems that exhibit this behavior.
In fact, it would be fascinating to find a natural problem that does not exhibit this behavior, i.e., a problem that is \NP-complete on interval graphs but admits a polynomial-time algorithm on split graphs.

\bibliographystyle{plainurl}
\bibliography{references}

\end{document}